\documentclass[a4paper,12pt,american]{article}
\usepackage{amsmath,amssymb,amsfonts,mathrsfs,accents, kpfonts,bm} 
\usepackage{arydshln}
\usepackage{enumerate} 
\usepackage[T1]{fontenc} 
\usepackage{palatino} 
\usepackage[utf8]{inputenc} 
\usepackage[kerning]{microtype} 
\usepackage[babel]{csquotes} 
\usepackage{booktabs} 
\usepackage[table]{xcolor}
\usepackage{lscape} 
\usepackage{multirow}
\usepackage[labelfont=bf]{caption} 
\usepackage[margin=1in]{geometry}
\usepackage{babel} 
\usepackage{amsthm} 

\newtheorem{theorem}{Theorem}
\newtheorem{lemma}{Lemma}
\newtheorem{corollary}{Corollary}
\newtheorem{prop}{Proposition}

\usepackage[colorlinks,citecolor=blue,urlcolor=burgundy,linkcolor=burgundy]{hyperref} 
\usepackage{booktabs, subcaption} 
\usepackage[bottom]{footmisc}
\usepackage[round,sort&compress,sectionbib]{natbib}
\setcitestyle{citesep={;}}
\usepackage{mathtools}
\usepackage{courier}
\definecolor{burgundy}{rgb}{0.5, 0.0, 0.13}
\usepackage{arydshln}
\usepackage{graphicx}
\usepackage{comment}
\usepackage{setspace}
\usepackage{tikz}
\usetikzlibrary{calc}
\usepackage{pgfplots}
\usepgfplotslibrary{fillbetween}
\usetikzlibrary{patterns}
\pgfplotsset{compat=1.11}
\usepackage{float}
\usepackage{soul}
\usepgfplotslibrary{fillbetween}
\definecolor{myblue}{rgb}{0.06, 0.2, 0.65}
\definecolor{burgundy}{rgb}{0.5, 0.0, 0.13}
\definecolor{myred}{rgb}{0.6, 0, 0.13}

\usepackage[]{appendix}
\usepackage{scrwfile}
\TOCclone[List of Appendices]{toc}{atoc}
\addtocontents{atoc}{\protect\value{tocdepth}=-1}
\newcommand\listofappendices{\listofatoc}
\newcommand*\savedtocdepth{}
\AtBeginDocument{%
  \edef\savedtocdepth{\the\value{tocdepth}}%
}
\let\originalappendix\appendix
\renewcommand\appendix{%
  \originalappendix
  \cleardoublepage
  \addtocontents{toc}{\protect\value{tocdepth}=-1}%
  \addtocontents{atoc}{\protect\value{tocdepth}=\savedtocdepth}%
}

\title{\textbf{\textcolor{myblue}{The Micro--Aggregated Profit Share}}\thanks{We thank Fil Babalievsky, V.V. Chari, Jose-Elias Gallegos, Eugenia Gonzalez--Aguado, Fatih Guvenen, Kyle Herkenhoff, Larry E. Jones, Loukas Karabarbounis, Rishabh Kirpalani, Hannes Malmberg, Simon Mongey, Sergio Ocampo, Chris Phelan, and James Traina for helpful comments and suggestions. We also wish to thank participants of the Labor--Firms--Macro reading group. Errors are our own. Comments are welcome. Emails: \href{mailto:hasen019@umn.edu}{hasen019@umn.edu}, \href{mailto:perez766@umn.edu}{perez766@umn.edu}.}}
\author{
    \textbf{Thomas Hasenzagl}\\ 
    \\
    \small{\textcolor{gray}{University of Minnesota}} \\
  \and
    \textbf{Luis Pérez} \\
    \textsc{Job Market Paper} \\
    \small{\textcolor{gray}{University of Minnesota}} \\
}
\date{\href{https://luiscanyamel.github.io/website_materials/JMP_LuisPerez.pdf}{Click here for the most recent version.} \\ \vspace{10pt} This version: \today. \\ \textcolor{gray}{First version: May 12, 2023.} \\ \vspace{10pt} }

\begin{document}
\maketitle

\vspace{-0.5cm}

\begin{abstract} 
\noindent How much has market power increased in the United States in the last fifty years? And how did the rise in market power affect aggregate profits? Using micro-level data from U.S. Compustat, we find that several indicators of market power have steadily increased since 1970. In particular, the aggregate markup has gone up from 10\% of price over marginal cost in 1970 to 23\% in 2020, and aggregate returns to scale have risen from 1.00 to 1.13. We connect these market-power indicators to profitability by showing that the aggregate profit share can be expressed in terms of the aggregate markup, aggregate returns to scale, and a sufficient statistic for production networks that captures double marginalization in the economy. We find that despite the rise in market power, the profit share has been constant at 18\% of GDP because the increase in monopoly rents has been completely offset by rising fixed costs and changes in technology. Our empirical results have subtle implications for policymakers: overly aggressive enforcement of antitrust law could decrease firm dynamism and paradoxically lead to lower competition and higher market power.\par 
\bigskip 

\end{abstract}


\newpage 
\section{Introduction}
How much has market power increased in the United States over the last fifty years? We use micro-level data from U.S. Compustat to address this question and find that several indicators of market power have steadily increased since 1970. We document that the aggregate markup, measured as the ratio of price to marginal cost, has gone up from 10\% in 1970 to 23\% in 2020 and that aggregate returns to scale have risen from 1.00 to 1.13. We map these market-power indicators to the aggregate profit share, which is the fraction of aggregate value added that is not used to compensate factors of production or cover fixed operating costs. More specifically, we show that the aggregate profit share can be expressed in terms of the aggregate markup, aggregate returns to scale, and a sufficient statistic for production networks that captures double marginalization in the economy (i.e., how profits propagate from downstream sellers to upstream suppliers.) \par 

By connecting indicators of aggregate market power to the profit share, we can not only quantify the profit share exploiting rich micro-level heterogeneity, but also understand its determinants. Studying the profit share is important for several reasons. First and foremost, profits, not markups, are the policy-relevant object for antitrust. Positive profits generally entail a loss of consumer surplus, while a markup above one, even if associated with a deadweight loss with respect to perfect competition, may have a positive welfare effect. Thus, the markup is less informative for antitrust policy than the profit share is.\footnote{To see this, notice that in the presence of fixed costs, a firm who cannot charge a markup at least as large as price over average cost would exit the market, destroying all consumer surplus. Hence, a price cap that limits the ability of firms to price above marginal cost can have detrimental welfare effects. On the other hand, a regulation that leads firms to make zero profits (i.e., to price at average cost) may be welfare enhancing, as it is easy to show that for a given technological environment, under general conditions, consumer surplus is highest when the profit share is lowest.} Leaving aside normative considerations, understanding the evolution of the profit share may also shed light on the decline of the labor share, the fall in business dynamism, the TFP slowdown, and other macroeconomic trends.\par

Our main goal is to understand if the rise in market power that we document for the US economy translated into higher profits. When firms gain market power, monopoly rents increase, but if fixed costs of production rise simultaneously, profits may remain unchanged. We argue that this narrative describes the experience of the United States between 1970 and 2020. Over that period, monopoly rents increased from 18\% to 41\% of GDP while the profit share remained constant at 18\% of GDP. We reconcile the increase in market power with a constant profit share by showing that the rise in monopoly rents associated with the increase in markups was counteracted by a proportional increase in returns to scale. Rising returns to scale reflect changes in technology and increases in fixed costs, such as R\&D expenses, advertising, regulatory compliance costs, and costs associated with information and communication technologies. \par 

Our findings are consistent with a view of the economy in which rising fixed costs lower competition by encouraging firm exit and discouraging firm entry.\footnote{This narrative is consistent with models of monopolistic competition \textit{\`{a} la} \cite{dixit1977monopolistic} and models of oligopolistic competition \textit{\`{a} la} \cite{atkeson2008pricing}.}$^,$\footnote{Business Dynamics Statistics from the US Census Bureau indicate that entry and exit rates of establishments and firms have both declined since 1980. See also \cite{decker2014role, decker2016declining}.} Rising fixed costs increase entry barriers.\footnote{\cite{de2019market}, \cite{aghion2019theory}, \cite{gutierrez2021entry}, \cite{de2022quantifying}, and others have pointed out the important role of rising fixed- and entry costs.} Monopoly rents are only sustainable in the presence of entry barriers since such rents will otherwise attract entrants and will be competed away. At the same time, firms that do not have enough market power to earn the monopoly rents required to cover rising fixed costs will exit. This lowers the degree of competition in the economy and allows surviving firms to increase their monopoly rents. However, firms use this increase in monopoly rents to cover the increasing fixed costs, and profits remain unchanged. \par 

Our empirical findings have subtle implications for policymakers. The increasing monopoly rents reflect increases in market power that could be counteracted by more aggressive enforcement of antitrust law. Such aggressive enforcement could, however, lower monopoly rents, making it unsustainable for some firms to operate given the high fixed costs. Therefore, overly aggressive enforcement of antitrust law could decrease firm dynamism and paradoxically lead to lower competition and higher market power. \par  

Our empirical analysis is built on three novel theoretical results. First, we show that the aggregate profit share can be constructed by weighing individual profit rates (i.e., profits over sales) using Domar weights (i.e., producer sales divided by GDP.)\footnote{It is worth emphasizing that Domar aggregation extends to the construction of other factor shares. When expenditures on factors of production are expressed relative to sales, labor- and capital shares can also be constructed from micro data using Domar weights.} Second, we use economic theory to derive a general expression of the profit rate in terms of a producer's markup, markdowns, and returns to scale. Third, we show that the profit share can be expressed in terms of several measures of aggregate market power---the aggregate markup, aggregate returns to scale, an aggregate monopsony term---and a sufficient statistic for production networks that captures double marginalization. \par 

Our first theoretical result establishes that the aggregate profit share can be constructed from micro-level data by weighting producers' profit rates---defined as profits over sales---using Domar weights (i.e., producer sales over GDP.) Domar weights are sufficient statistics for production networks that capture the influence of each producer on aggregate value added. If there are no production networks, Domar weights reduce to sales weights, sum up to one, and reflect the influence of producers on aggregate value added through the sale of final goods. With production networks, Domar weights differ from sales shares, their sum exceeds unity, and they reflect both the direct- and the indirect influence of producers on aggregate value added. Producers have a direct influence on aggregate value added through the sale of final consumption goods, and an indirect influence through the sale of intermediate goods to other producers. \par  

Our second theoretical result provides a general expression for a producer's \textit{economic} profit rate in terms of its markup, markdowns, and returns to scale.\footnote{A similar profit rate expression is provided by \cite{basu1997returns, basu2002aggregate}, \cite{basu2019price} and \cite{syverson2019macroeconomics}. Our expression is more general than theirs since we allow for explicit fixed costs and market power in factor markets.} We put emphasis on economic profits to distinguish them from accounting profits. Economic profits, unlike accounting profits, are unobserved and exclude the opportunity cost of using factors of production. We are able to recover economic profits, cast in the form of an economic profit rate, under very mild assumptions on producer behavior and technology. All that is needed to establish our second result is cost minimization and the existence of a production technology that satisfies standard regularity conditions (i.e., differentiability, quasiconcavity, and Inada conditions.) \par 

Why do economic profit rates depend on markups, markdowns, and returns to scale? A producer exerts market power if it charges a price that is above marginal cost---that is, if its markup is greater than unity---or if it compensates factors of production (say, workers) at a rate below their marginal revenue product---that is if its markdowns are lower than unity. Despite exerting market power, a producer only profits from market power when its rents from monopoly and monopsony power exceed its fixed costs of production, which are captured by the returns to scale. \par 

Our third theoretical result is an aggregation theorem that expresses the aggregate profit share in terms of several indicators of aggregate market power---the aggregate markup, aggregate returns to scale, an aggregate monopsony term---and a sufficient statistic for production networks that captures the degree of vertical integration in the economy. Production networks are important because of double marginalization: sales of upstream producers do not only capture the market-power rents of these producers but also those of downstream suppliers. \par 

This aggregation theorem is our main theoretical contribution, and we use it to guide our empirical analysis. Our formula presents several advantages over existing ones for computing the aggregate profit share. First and foremost, our theorem links aggregate indicators of market power to the aggregate profit share, allowing us to understand the determinants of the profit share, as well as the origins of aggregate profits. Does an increase in market power translate into an increase in profits? If not, why? Do aggregate profits originate from monopoly, monopsony, or both? And to what extent do each of these forces drive the level of aggregate profits? \par 

Apart from allowing us to answer these questions, our theorem implies aggregate measures of markups, markdowns, and returns to scale. Aggregation of firm-level markups is a topic that has generated substantial debate among academic economists. We contribute to this debate by suggesting that the aggregate markup is the harmonic sales-weighted markup, that aggregate markdowns are sales-weighted markdowns, and that aggregate returns to scale are sales-weighted returns to scale. \par 

Our theorem can also be used to assess the external validity of micro-level estimates of markups, markdowns, and returns to scale. Do firm-level estimates of these objects imply reasonable profit shares? In Appendix \ref{app:basu_deu}, we use our theorem to elucidate the back-and-forth discussion between \cite{de2020rise} and \cite{basu2019price} on whether \citeauthor{de2020rise}'s micro-level estimates have unreasonable macroeconomic implications. Under additional assumptions on technology, our theorem can also be used to easily calibrate economic models with monopolistic or monopsonistic wedges (e.g., obtaining a markup shock series for a New Keynesian model.) \par 

\subsection*{Related Literature}
The study of market power has been a subject of extensive research, with a rich and expansive literature dating back to the works of  \cite{pigou1920economics}, \cite{chamberlin1933theory}, and \cite{robinson1933}. \par 

Our paper mainly relates to a large literature on market power and macroeconomics. An influential paper in this literature is \cite*{de2020rise}. They estimate markups in Compustat data and find that the aggregate markup has increased from 21\% of price over marginal cost in 1980 to 61\% in 2016. Instead, we find that the aggregate markup has increased from 12\% to 23\% over the same period. There are three reasons why our series differ. The main reason is that we calculate the aggregate markup as the harmonic sales-weighted average of firm-level markups while \citeauthor*{de2020rise} use the sales-weighted average markup. Our aggregation scheme, unlike theirs, has a solid theoretical foundation.\footnote{Other studies in which the harmonic sales-weighted markup emerges as a natural measure of aggregate markup include \cite{baqaee2020productivity}, \cite{edmond2023costly}, and \cite{smith2022evolution}.} Had we used the sales-weighted average markup instead of the harmonic sales-weighted one, we would have found an increase in price over marginal cost of twenty two (rather than thirteen) percentage points over the last fifty years. The second reason is that we categorize selling, general, and administrative expenses (SG\&A,) as well as the costs of goods sold (COGS,) as variable costs, following \cite{traina2018aggregate}.\footnote{Markup estimates are sensitive to the choice of variable costs, but estimates of the profit share are not because profit rates are the same regardless of whether inputs are categorized as variable or fixed as long as all costs are accounted for.} The different categorization of variable costs explains less than half of our discrepancies. A third reason why our series of aggregate markup differs---one that does not matter much in practice---is that we use a measure of the capital stock that includes both physical- and intangible capital.\footnote{A recent literature relates the rise in market power to the increasing importance of intangible capital, such as software and intellectual property products \citep[see, for example,][]{crouzet2022, de2019market}. On the one hand, investment in intangibles increases fixed costs. On the other hand, the presence of intangibles in the production process lowers marginal costs, allowing firms to charge higher markups. Our empirical analysis captures these two mechanisms. We do so by including intangibles in the stock of capital in our production function estimation.} \par

While we are interested in the evolution of the aggregate markup and other indicators of aggregate market power, our main objective is to study the connection between market power and the profit share. We establish and study this connection through our three theoretical results. In doing so, our work relates and contributes to the literatures on production networks and the functional distribution of income. Here, we briefly review some of the most pertinent papers in those literatures and discuss our contributions in relation to them. \par 
\bigskip 

\noindent \textbf{Production Networks}. Research on production networks, which builds upon the pioneering work of \cite{quesnay1758tableau} and \cite{leontief1951structure, leontief1966input}, has significantly increased in recent decades.\footnote{A list of non-exhaustive papers in this literature is: \cite{long1983real}, \cite{acemoglu2012network, acemoglu2015systemic, acemoglu2016networks}, \cite{grassi2017io}, \cite{bigio2020distortions}, \cite{ghassibe2021endogenous}, \cite{baqaee2019macroeconomic, baqaee2019networks, baqaee2020productivity}.} Within this literature, our work relates most closely to \cite{baqaee2020productivity}. Our aggregation theorem, which expresses the aggregate profit share in terms of a sufficient statistic for production networks and aggregate measures of market power, generalizes the production side of that framework. It is possible to show that the profit share in their framework can be expressed as the input-output multiplier (i.e., the ratio of sales to GDP) times one minus the inverse of the harmonic sales-weighted markup. Thus, our theorem nests the production side of \cite{baqaee2020productivity}'s economy when there is no monopsony, no fixed costs, and returns to scale are constant for each producer. \par 
\medskip 

\noindent \textbf{The Functional Distribution of Income}. Income shares are important summary statistics for understanding many macroeconomic phenomena, including the global decline in the labor share, the stability of the Kaldor facts, and economic inequality. One such statistic that has historically caught the eye of both economists and policymakers is the profit share. \par 

Studying the profit share is important for several reasons. First, as argued above, the profit share in itself is more informative for antitrust policy than the markup is. Second, understanding the evolution of the profit share may shed light on the global decline in the labor share documented by \cite{karabarbounis2014global}, the fall in business dynamism reported by \cite{decker2014role}, and the TFP slowdown noted by \cite{gordon2012us}.\footnote{Income shares---that is, labor-, capital- and profit shares---add up to unity. Thus, if the labor share has declined, that necessarily means that either the capital or the profit share must have increased. }$^,$\footnote{In a world where technological change results from innovative activity, economic profits provide incentives for inventors and entrepreneurs to pursue new ideas and business projects.}$^,$\footnote{Economic profits result from monopolistic- and monopsonistic distortions, both of which may induce misallocation of production factors and thus affect aggregate productivity.} More generally, income shares (i.e., labor-, capital- and profit shares) are important for assessing the extent of wealth and income inequality \citep{piketty2003income, atkinson2011top}, and for calibrating economic models, whether it is in the areas of growth, development, or monetary policy \citep[e.g., see][]{boppart2023macroeconomics,gali2005monetary}. Finally, the profit share is informative for economic policies on corporate taxation, redistribution, and antitrust. \par 

Our theoretical results unify micro and macro approaches for computing income shares. Consistent aggregation of micro data permits obtaining income shares that will be identical to those obtained using macro data. That is, it does not matter whether one computes the aggregate profit share either as a residual of one minus the sum of labor and capital shares or by aggregating individual profit rates (defined as profits over sales) using Domar weights as long as one uses representative micro-level data. If the aggregate user cost of capital used to calculate the capital share via the macroeconomic approach is consistent with the underlying heterogeneity in producers' user costs of capital, then income shares would be identical regardless of the procedure---micro or macro---employed to compute these shares. \par 

Thus, an advantage of our theoretical results is that they can be used to make a direct connection between labor, capital, and profit shares using micro-level data. Two examples where our results can help include studying the role of the profit share in the decline of the labor share \citep{elsby2013decline, karabarbounis2014global, kehrig2021micro} and the stability of the Kaldor facts \citep{eggertsson2021kaldor}. In Appendix \ref{app:income_shares}, we calculate the labor share from US National Accounts and use our estimates of the micro-aggregated profit share to back out the capital share. The seven percentage-point decline in the labor share that occurred between 1970 and 2020, together with the constancy of the profit share that we document, imply that the capital share has absorbed most of the decline in the labor share. \cite{rognlie2016deciphering} argues that the capital share increased due to the housing sector. Overall, our results imply that income shares are much more in line with the Kaldor facts than what others have suggested.\par 
\bigskip 

\noindent \textbf{Layout of the paper}. The rest of the paper is organized as follows. Section \ref{sec:theory} presents our main theoretical results and connects them to the existing literature. Section \ref{sec:data} discusses data sources and the methodology. Section \ref{sec:empirics} offers our main empirical results. Section \ref{sec:conclusion} concludes. Appendices \ref{app:proofs}--\ref{app:results} provide omitted proofs, summary statistics for our empirical application, additional results, and robustness exercises. \par

\section{Aggregation Results on Market Power}\label{sec:theory}
In this section, we provide two aggregation results on market power. First, we show that the aggregate profit share can be constructed by weighting individual profit rates, defined as profits over sales, using Domar weights (i.e., sales divided by GDP.) Then, we show that under additional assumptions, a producer's profit rate can be written as a function of its markup, markdowns, and returns to scale. Second, we express the aggregate profit share in terms of the aggregate markup, an aggregate monopsony term, aggregate returns to scale, and a sufficient statistic for production networks.\par 

\subsection{Aggregate Profit Shares: Two Approaches}
In principle, one can compute aggregate profit shares using two different approaches: the macroeconomic- and the microeconomic approach. We use \textit{macro approach} to refer to the procedure employed to compute the profit share using aggregate data from National Accounts. And we use \textit{micro approach} to refer to the procedure that constructs the profit share from micro-level data. \par 
\bigskip 

\noindent \textbf{The Macro Approach}. The macroeconomic approach starts by recognizing that the aggregate value added of an economy is equal to the sum of compensation to factors of production and a residual term, here labeled ``wedge income,'' that captures all other sources of income. Under the usual assumption that there are only two factors of production---capital and labor---we can write 
    \begin{equation}
        \text{GDP} = WL + RK + \text{wedge income},
    \end{equation}
where $\text{GDP}$ is aggregate value added, $WL$ is labor compensation, and $RK$ are imputed (gross-of-depreciation) capital rents. In theory, the term ``wedge income'' does not only capture profit income but also any income that accrues to other wedges such as taxes on production and sales, tariffs, etc.\footnote{\cite{karabarbounis2019accounting} refer to the term we call wedge income as ``factorless income.''}$^,$\footnote{In practice, the term wedge income may reflect some factor income because of measurement error caused by mismeasuring capital or its user cost, as well as by incorrectly imputing payments to labor.} In the absence of wedges other than monopolistic or monopsonistic wedges, we can write
    \begin{align}
        \text{GDP} = WL + RK + \text{Profits},
    \end{align}
and then compute the aggregate profit share as
    \begin{align}\label{eq:Pimacro}
        \Lambda_\Pi^{\text{Macro}} = 1 - \Lambda_L - \Lambda_K,
    \end{align}
where $\Lambda_\Pi$ is the profit share, $\Lambda_L$ is the labor share, and $\Lambda_K$ is the capital share. \par 

Existing measures of the aggregate profit share have been constructed using this procedure. That is, one first obtains value added, labor compensation, and the stock of capital from National Accounts. Next, one imputes or estimates the aggregate user cost of capital to compute the capital share. Finally, equipped with labor- and capital shares, one infers the profit share using equation (\ref{eq:Pimacro}). \par 
\bigskip 

\noindent \textbf{The Micro Approach}. In this paper, we develop the microeconomic approach, which consists of aggregating individual producers' profit rates according to 
    \begin{align}\label{eq:profshare_microapproach}
        \Lambda_\Pi^{\text{Micro}} = \sum_{i\in\mathcal I}\omega_i s_{\pi_i}, 
    \end{align}
where $\omega_i$ is producer $i$'s weight, and $s_{\pi_i}$ is its profit rate. \par

Formulation (\ref{eq:profshare_microapproach}) permits different aggregation schemes, depending on the definition of the profit rate. Next, we show that when the profit rate is defined as profits over sales, one can construct the profit share from micro-level data using Domar weights---that is, producer sales over aggregate value added. Defining profit rates as profits divided by sales rather than value added is convenient because value added is not observed in most micro-level datasets, but sales are. \par 

\begin{lemma}[\textbf{The Micro--Aggregated Profit Share}]\label{lemma:microprofitshare}
If profit rates, defined as profits over sales, are aggregated using Domar weights (i.e., sales over GDP,) then micro- and macro approaches both yield the aggregate profit share. 
\end{lemma}
 \begin{proof}
 Let $\Pi$ denote aggregate profits. For each producer $i\in\mathcal I$, $\pi_i$ denotes profits, $p_iy_i$ are sales, and $s_{\pi_i}:= \pi_i/(p_iy_i)$ is $i$'s profit rate. Then, we have
    \begin{align*}
        \Lambda_\Pi^{\text{Macro}} &= \frac{\Pi}{\text{GDP}} = \frac{\sum_{i\in\mathcal I} \pi_i}{\text{GDP}} = \frac{\sum_i s_{\pi_i}p_iy_i}{\text{GDP}} = \sum_{i\in\mathcal I}\frac{p_iy_i}{\text{GDP}} \times s_{\pi_i} \equiv \Lambda_\Pi^{\text{Micro, DW}}.
    \end{align*}
\end{proof}

\noindent The intuition for using Domar weights to arrive at the aggregate profit share is that Domar weights, defined as producer sales over GDP, are ``summary statistics'' for input-output linkages that capture the influence of each producer on aggregate value added. If there are no production networks, Domar weights reduce to sales shares, sum up to one, and reflect the influence of producers on aggregate value added through the sale of final goods. With production networks, Domar weights differ from sales shares, their sum exceeds unity, and reflect both the direct- and the indirect influence of producers on aggregate value. Producers have a direct influence on aggregate value added through the sale of final consumption goods, and an indirect influence through the sale of intermediate inputs to other producers. \par 

Why do production networks matter in the construction of the aggregate profit share? In the presence of intermediate goods, because of double marginalization, the income generated by downstream producers (i.e., those directly selling to consumers) can be split into three terms: own profits, profits to others, and pure costs. Profits to others are accounted for in the costs of downstream producers that buy materials from other producers which may also charge positive markups and thus earn profits. In order to arrive at the aggregate profit share from micro-level data, one must take these input-output linkages into account, which is precisely what Domar weights do. \par 


Three important remarks are in order. First, establishing Lemma \ref{lemma:microprofitshare} requires no modeling assumptions, only fundamental accounting principles. If economic profit rates were observable at the individual level, the profit share could be easily constructed from micro data. Second, alternative aggregation schemes are possible. If individual profit rates are alternatively defined as profits over value added rather than sales, then the aggregate profit share can be constructed from micro-level data by weighting profit rates using value-added weights (See Lemma \ref{lemma:vaprofitshare} in Appendix \ref{app:proofs}.) Finally, although we have defined the profit share economy-wide, Lemma \ref{lemma:microprofitshare} applies at any desired level of aggregation as long as we make the necessary adjustments in terms of value added.

\subsection{Profit Rates in Terms of Monopoly and Monopsony}
While Lemma \ref{lemma:microprofitshare} is a useful aggregation result, it may not be so practical when confronted with the data. This is because profit rates are typically not observed, and, even when they are, they do not measure economic- but accounting profits.\footnote{Economic profits reflect the opportunity costs of using factors of production, whereas accounting profits do not. Businesses may have incentives to overstate capital depreciation when reporting accounting profits in balance-sheet data. The classic reference shedding light on the distinction between economic- and accounting profits is \cite{fisher1983misuse}.} To overcome this, we use economic theory to show that the \textit{economic} profit rate of any given producer can be written in terms of its markup, markdowns, and returns to scale. \par 

The economic environment we consider is very general. We think of firms as minimizing current-period costs to deliver a given amount of output, although this output could potentially result from complex strategic interactions among firms in dynamic environments, as in \cite{abreu1986extremal}. Firms operate production technologies which employ two types of inputs, intermediate goods and factor inputs, and face two types of costs, variable and fixed. Production technologies may exhibit arbitrary scale elasticities, the only (technical) requirements are that they be differentiable, quasiconcave, and satisfy Inada conditions. Intermediates include utilities and materials, and factor inputs are different types of capital and labor. Variable costs include purchases of intermediate inputs, as well as the compensation of some factors of production. Fixed costs can be either implicit or explicit. Explicit fixed costs do not affect production capacity and include investments towards future productive capacity, such as investment in R\&D. Implicit fixed costs affect production capacity, such as when a minimum amount of capital is necessary to start production, and may be captured by the scale elasticity if they cannot be separated from variable costs. \par 

We allow for firms to have pricing power, both in output and factor markets. Market power in output market is captured by the markup of price over marginal cost, and market power in input markets is captured by the markdown of the marginal revenue product of a factor over its rental rate. For convenience, we treat both monopolistic and monopsonistic wedges as exogenous, although it is not complicated to endogenize these. Under the general economic environment described above, we can establish the following proposition. \par 
\medskip 

\begin{prop}[Profit Rates, Monopoly, and Monopsony]\label{prop:indprofitrate}
Under the assumptions of cost-minimizing behavior, the existence of a continuously differentiable and quasiconcave production function, and monopsony power in factor markets, a producer's profit rate, defined as profits over sales, can be written as
    \begin{align}\label{eq:indprofitrate}
        s_{\pi} = 1 - \frac{\text{RS}}{\mu} = 1 - \frac{\text{SE}^{\text{adj}}}{\mu} + \frac{\mathcal M}{\mu},
    \end{align}
where $\text{RS} = \text{SE}^{\text{adj}} - \mathcal M$ are the returns to scale, $\text{SE}^{\text{adj}}$ is the scale elasticity of the production function adjusted for fixed costs, $\mathcal M$ is a monopsony term capturing market power in factor markets, and $\mu$ is the markup of price over marginal cost. \par 

The scale elasticity adjusted for fixed costs is defined as
    \begin{align}\label{eq:RSadj}
        \text{SE}^{\text{adj}} := \text{SE}\times\left(\frac{\text{TC}}{\text{TC} - \text{FC}}\right),
    \end{align}
where $\text{SE}$ is the scale elasticity, $\text{TC}$ are total costs, and $\text{FC}$ are explicit fixed costs. \par 

The monopsony term is given by
    \begin{align}\label{eq:monopsony_term}
        \mathcal M := \left(\frac{\text{TC}}{\text{TC} - \text{FC}}\right)\sum_{j\in \mathcal F}\theta_{j}(1-\nu_{j}),
    \end{align}
where $\theta_{j}\equiv \partial F/\partial x_{j} \times x_{j} / y$ is the elasticity of output with respect to input $j$, and $\nu_{j}$ is the markdown on factor $j\in\mathcal F$, where the markdown is defined as the ratio of input $j$'s rental rate to its marginal revenue product; that is, $\nu_{j}:= w_j(x_{j})/\text{MRP}_{j}$. 
\end{prop}

\begin{proof}
    See Appendix \ref{app:prop1}.
\end{proof}

\noindent \textbf{Remarks}. Under minimal assumptions on producer behavior and technology---namely, cost minimization and a continuously differentiable and quasiconcave production technology---equation (\ref{eq:indprofitrate}) states that a producer's profit rate can be written in terms of its markup $\mu$ and returns to scale $\text{RS}$.\footnote{Implicit in our formulation is the assumption that the firm's optimality condition for capital holds in every period, which also allows us to back out the user-specific cost of capital.} Alternatively, the profit rate can be written in terms of the markup, $\mu$, the scale elasticity of the production technology adjusted for fixed costs, $\text{SE}^{\text{adj}}$, and a monopsony term $\mathcal M$ that depends on markdowns $\nu_j$ and captures market power in factor markets.  \par 

The economic intuition for equation (\ref{eq:indprofitrate}) is the following. Profits originate from two sources: monopoly and monopsony. The term $\text{SE}^{\text{adj}}/\mu$ captures how monopoly power translates into profits. A producer that exerts monopoly power makes profits when it charges a markup above the scale elasticity adjusted for fixed costs; that is, if $\mu > \text{SE}^{\text{adj}}$. If the markup were to equal the scale elasticity adjusted for fixed costs and there were no markdowns, a producer would use his market power to cover fixed costs and compensate for scale economies. The term $\mathcal M/\mu$, captures how monopsony power in factor markets translates into profits. When a producer has market power in a factor market, the rental rate of that factor is below its marginal revenue product so that the producer earns monopsony rents from using that factor. Monopsony power translates into profits if monopsonistic rents exceed fixed costs.\footnote{Notice that the term $\mathcal M$ features a fixed-cost adjustment, $TC/(TC - FC)$, and input elasticities for each factor $j$. For simplicity, consider the case in which there are no fixed costs and focus on a particular factor $j$. If there is no monopoly power and $\nu_j$ is lower than unity, monopsony power translates into profits if $\theta_{j}(1-\nu_{j})\neq1$. With fixed costs, market power in input $j$ translates into profits if $\frac{TC}{TC-FC}\theta_{j}(1-\nu_{j})\neq1$.} \par 

Equation (\ref{eq:indprofitrate}) is a generalization of other expressions of profit rates in the literature. \cite{basu1997returns, basu2002aggregate}, \cite{basu2019price} and \cite{syverson2019macroeconomics} show that the profit rate (defined as profits over sales) can be written as
    \begin{align}\label{eq:profitrate_basu}
        s_{\pi} &= 1 - \frac{\text{RS}}{\mu}.
    \end{align}
Although equations (\ref{eq:profitrate_basu}) and (\ref{eq:indprofitrate}) look identical when expressed in terms of returns to scale, our expression is more general. This is because of two reasons. First, we explicitly allow for fixed costs. Second, we allow for monopsony power in factor markets. Both of these generalizations affect the profit rate through the returns to scale. When there are no explicit fixed costs (i.e., $\text{FC}=0$) and no monopsony power (i.e., $\nu_{j}=1$, $\forall j\in\mathcal F$,) the returns to scale equal the scale elasticity of the production technology and our expression and that of \citeauthor{basu1997returns}, \citeauthor{basu2019price} and \citeauthor{syverson2019macroeconomics} are identical.\footnote{Imposing constant returns to scale, equations (\ref{eq:indprofitrate}) and (\ref{eq:profitrate_basu}) both yield the familiar profit rate $s_{\pi} = 1 - \mu^{-1}$.} \par 

When there are explicit fixed costs (i.e., $\text{FC}>0$,) returns to scale are no longer equal to the scale elasticity and must be adjusted to incorporate these fixed costs. That explains why in the absence of monopsony power, our formula equates returns to scale to the scale elasticity of the production function adjusted for fixed costs. Further allowing for monopsony power in factor markets (i.e., $\nu_j\in(0,1)$ for some $j\in\mathcal F$) requires making one additional adjustment to the returns to scale. With both explicit fixed costs and monopsony power, returns to scale depend on the scale elasticity, explicit fixed costs, and a monopsony term. Returns to scale, defined as the ratio of average- to marginal cost, depend on the monopsony term because a producer can affect its marginal cost of production by exerting monopsony power in factor markets. This explains why, according to our formula, returns to scale equal the scale elasticity adjusted for fixed costs minus the monopsony term; that is, $\text{RS} =\text{SE}^{\text{adj}} -\mathcal M$. \par 

An important clarification is that the scale elasticity of the production function adjusted for fixed costs captures all fixed costs.\footnote{See Appendix \ref{app:FC_example} for a discussion on the relationship between returns to scale and fixed costs related to production factors. A related point is that changes in returns to scale can in principle capture both increasing fixed costs and technological change. Throughout the paper, we interpret rising returns to scale as increases in fixed costs for the ease of exposition.} In the data, it is not always possible to separate inputs that are directly used in the production of output from those that are not. If available data combines expenditures on production factors without separating those expenditures into variable and fixed costs, the estimated scale elasticity will capture all fixed costs. If, however, a direct measure of fixed costs is observable in the data, it can be beneficial to estimate the scale elasticity only using variable inputs. This is because the estimated scale elasticity can be adjusted to incorporate fixed costs not included in the estimation procedure, and the scale elasticity can rarely be estimated at the firm level. By using a fixed-cost adjustment factor, we can exploit firm-level heterogeneity in scale elasticities even within particular industries. We use this theoretical insight in our empirical exercise and estimate scale elasticities at the firm level. We do so by constructing a fixed-cost adjustment factor for each firm with observable R\&D expenditures, and by estimating the scale elasticity at the industry level.\par 
\bigskip 

\noindent \textbf{Remarks}. In Appendix \ref{app:proofs}, section \ref{app:Rident_proof}, we show that Lemma \ref{lemma:microprofitshare} combined with equation (\ref{eq:indprofitrate}) can be used to identify the aggregate user cost of capital. Finally, we note that expressing profit rates as profits divided by sales rather than value added is standard practice since value added is not observed in many datasets, but sales are. \par 

\subsection{Linking the Aggregate Profit Share to Market-Power Indicators}\label{sec:decomposition}
Our next result states that the aggregate profit share can be expressed in terms of a sufficient statistic for production networks that captures double marginalization in the economy and two indicators of aggregate market power---the aggregate markup and an aggregate monopsony term that captures market power in factor markets. \par 

\begin{theorem}\label{thm:profshare_decomp}
With cost-minimizing producers, continuously differentiable and quasiconcave production functions, fixed costs, and market power in factor- and output markets, the profit share can be expressed as
    \begin{align}
        \Lambda_\Pi &= \chi\left(1 - \mathbb{E}_{\omega}\left[\frac{\text{SE}^{\text{adj}}}{\mu}\right] + \mathbb{E}_{\omega}\left[\frac{\mathcal M}{\mu}\right]\right) \label{thm:agg_prof_general1} \\
            &= \chi\left(1 - \frac{\overline{SE}^{\text{adj}}}{\overline\mu_{hsw}} + \frac{\overline{\mathcal M}}{\overline\mu_{hsw}} - \text{Cov}_\omega\left[\text{SE}^{\text{adj}}, \frac{1}{\mu}\right] + \text{Cov}_\omega\left[\mathcal M, \frac{1}{\mu}\right]\right), \label{thm:agg_prof_general2}
    \end{align}
where $\chi =\sum_{k\in\mathcal I} \frac{p_ky_k}{\text{GDP}}$ is the input-output multiplier, $\mathbb{E}_{\omega}\left[\cdot\right]$ denotes a sales-weighted average where $\omega$ indexes sales weights, 
$\text{SE}^{\text{adj}}$ is the scale elasticity adjusted for fixed costs given by (\ref{eq:RSadj}), $\mu$ denotes the markup, and $\mathcal M$ is a monopsony term given by (\ref{eq:monopsony_term}).\par 

The term $\overline X$ is the sales-weighted average of $X$, $\overline\mu_{hsw}$ is the harmonic sales-weighted markup, and $\text{Cov}_\omega(X,Y)$ is the sales-weighted covariance of variables $X$ and $Y$.
\end{theorem}

\begin{proof}
    See Appendix \ref{app:thm1}.
\end{proof}

Theorem \ref{thm:profshare_decomp} links several indicators of aggregate market power to the profit share. Since these market-power indicators are obtained from micro data, Theorem \ref{thm:profshare_decomp} can be used to assess the macroeconomic implications of micro-level estimates of markups, markdowns, and returns to scale in terms of the aggregate profit share. This theorem also permits decomposing the aggregate profit share into different sources of market power and, hence, assessing to what extent aggregate profits originate from monopoly and monopsony. An important clarification is that while we have stated our theorem in terms of the economy-wide profit share, Theorem \ref{thm:profshare_decomp} can be adapted with ease to study profit shares at any desired level of aggregation, such as industry-level profit shares. By quantifying the sources of market power---monopoly \textit{vis-à-vis} monopsony---for particular industries, empirical applications of Theorem \ref{thm:profshare_decomp} can inform policy discussions on antitrust and help draft industry-specific regulations. We agree with \cite*{berry2019increasing} in that detailed industry analysis is needed, and believe that our theorem is useful in that respect as it bridges macro and industrial-organization approaches. \par 

The second line in Theorem \ref{thm:profshare_decomp}, equation (\ref{thm:agg_prof_general2}), states that the profit share can be computed using the input-output multiplier (i.e., the ratio of total sales to aggregate value added,) the aggregate markup, an aggregate monopsony term, and covariance terms. Since the aggregate profit share can be expressed as a function of the harmonic sales-weighted markup, this theorem provides a natural argument for using the harmonic sales-weighted markup as our measure of aggregate markup. We acknowledge that we are not the first to point to this measure as a natural choice for the aggregate markup. \cite{baqaee2020productivity} and \cite*{edmond2023costly}, among others, have noted this earlier. However, we show that the harmonic sales-weighted markup continues to be an appropriate choice for the aggregate markup in more general environments. While \citeauthor{baqaee2020productivity} allow for arbitrary production networks, as we do, neither they nor \citeauthor*{edmond2023costly} allow for increasing returns to scale, fixed costs, or monopsony. In the same manner that we established a measure of aggregate markup, we note that sales-weighted average returns to scale constitute a natural measure of aggregate returns to scale. Similarly, one could refer to the sales-weighted monopsony term, $\overline{\mathcal M}$, as the aggregate monopsony term. Finally, we note that the aggregate monopsony term captures market power on several factor markets through factor markdowns. For those interested in computing aggregate markdowns, our theorem suggests that the aggregate markdown of a particular factor is the sales-weighted markdown.\par 


The following corollaries are special cases to Theorem \ref{thm:profshare_decomp} and are useful for inferring profit shares and markups, as well as for assessing the validity of a researcher's own micro estimates or those of existing studies, as we demonstrate in the sequel. \par 

\begin{corollary}[No Monopsony]\label{corr:nomonop}
Assuming price-taking behavior in input markets (i.e., $\nu_{ij}=1$ for all $i,j$), Theorem \ref{thm:profshare_decomp} reduces to
    \begin{align}\label{eq:profshare_nomonop}
        \Lambda_\Pi = \chi\left(1 - \frac{\overline{RS}}{\overline\mu_{hsw}} - \text{Cov}_\omega\left[\text{RS}, \frac{1}{\mu}\right]\right),
    \end{align}
where $\text{RS} = \text{SE}^{\text{adj}}$.
\end{corollary}

\begin{corollary}[No Monopsony and No Fixed Costs]\label{corr:nomonopFC}
Assuming price-taking behavior in input markets (i.e., $\nu_{ij}=1$ for all $i,j$) and no fixed costs (i.e., $\text{FC}_i=0$ for all $i$,) Theorem \ref{thm:profshare_decomp} reduces to
    \begin{align}\label{eq:profshare_nomonopFC}
        \Lambda_\Pi = \chi\left(1 - \frac{\overline{RS}}{\overline\mu_{hsw}} - \text{Cov}_\omega\left[\text{RS}, \frac{1}{\mu}\right]\right),
    \end{align}
where $\text{RS} = \text{SE}$.
\end{corollary}

\begin{corollary}[No Monopsony, No Fixed Costs, and CRS]\label{corr:corr:nomonopFCRS}
Assuming price-taking behavior in input markets (i.e., $\nu_{ij}=1$ for all $i,j$,) no fixed costs (i.e., $\text{FC}_i=0$ for all $i$,) and constant returns to scale (i.e., $\text{RS} = 1$ for all $i$,) Theorem \ref{thm:profshare_decomp} reduces to
    \begin{align}\label{eq:profshare_nomonopFCRS}
        \Lambda_\Pi = \chi\left(1 - \frac{1}{\overline\mu_{hsw}}\right).
    \end{align}
\end{corollary}

Equation (\ref{eq:profshare_nomonopFCRS}) in Corollary \ref{corr:corr:nomonopFCRS} states that under the assumptions of cost minimization, price-taking behavior in input markets, no fixed costs, and continuously-differentiable production functions of homogeneity degree one, the aggregate profit share can be computed as the input-output multiplier---the ratio of sales to GDP---times one minus the inverse of the harmonic sales-weighted markup. This expression of the aggregate profit share obtains in the general framework of \cite{baqaee2020productivity}.\footnote{This profit share would be recorded in the first row and $f^*$-th column of the Leoentief inverse matrix augmented with a fictitious factor called ``profits'' and denoted $f^*$.} Such an expression can be used for inference provided one is willing to make strong assumptions on returns to scale. Equipped with markup estimates for all producers in the economy, one can compute the harmonic sales-weighted average markup, read the input-output multiplier from National Accounts, and then infer the aggregate profit share. Similarly, provided knowledge on the aggregate profit share, one can rearrange equation (\ref{eq:profshare_nomonopFCRS}) and back out the aggregate markup as
    \begin{align}\label{eq:markup_backout_cor3}
       \overline\mu_{hsw} = \left(1 - \frac{\Lambda_\Pi}{\chi}\right)^{-1}.
    \end{align}
Importantly, notice that taking input-output networks into account is crucial for inferring the aggregate profit share from micro data and for backing out the aggregate markup from macro data. Neglecting production networks would lead to underestimating the profit share by a factor of $\chi$, which is around two in the United States. Similarly, ignoring production networks would lead to overstating the aggregate markup by a factor of
    \begin{align*}
        \frac{(\chi - 1) \Lambda_\Pi}{(1-\Lambda_\Pi)(\chi-\Lambda_\Pi)}.
    \end{align*}
This bias can be substantial. To see this, suppose that $\chi = 2$ and $\Lambda_\Pi = 0.2$. Then, ignoring production networks would lead us to conclude that the aggregate markup is 113\% larger than it actually is. Finally, notice that if there is no vertical integration in the economy, $\chi = 1$, and neither inference of the aggregate markup nor that of the profit share are biased.\par 

Equation (\ref{eq:profshare_nomonopFC}) in Corollary \ref{corr:nomonopFC} presents a more interesting case by allowing for decreasing and increasing returns to scale in production. In that case, the aggregate profit share can be constructed as the input-output multiplier times one minus sales-weighted returns to scale divided by the harmonic sales-weighted markup and a covariance term.\par 

Finally, Corollary \ref{corr:nomonop} relaxes the assumption of no fixed costs. Under a no-monopsony assumption, the aggregate profit share can be constructed as the input-output multiplier times one minus sales-weighted returns to scale divided by the harmonic sales-weighted markup and a covariance term. Returns to scale equal the scale elasticity times a fixed-cost adjustment factor. In Appendix \ref{app:basu_deu}, we use Corollary \ref{corr:nomonop} to elucidate the back-and-forth discussion between \cite{basu2019price} and \cite{de2020rise} on how to map micro-levels estimates of market power to the aggregate profit share. \par

\section{Data and Methodology}\label{sec:data}
In this section, we review our data sources and the methodology employed for estimating markups and inferring the aggregate profit share. \par 

\subsection{Compustat}
Our main data comes from Compustat North America Fundamentals Annual, which we retrieve through the Wharton Research Data Services (WRDS.) Compustat provides balance-sheet data for US-incorporated publicly-traded firms. \par 
\medskip  

\noindent \textbf{Sample selection}. We download annual data for all available firms between 1966 and 2020. We deflate all financial variables using appropriate price deflators from US NIPA, and classify firms according to two-digit NAICS industries to estimate input elasticities. The NAICS industry classification is adapted to conform with the Bureau of Economic Analysis (BEA) when classifying industries.\footnote{This essentially means that 2-digit NAICS industries 31-33, 44-45, and 48-49 are grouped together.} We drop observations from our sample when firms have missing fiscal years, no NAICS code, or multiple NAICS codes. We also drop observations that register negative or missing values for sales, costs of goods sold, selling, general and administrative expenses, R\&D, or investment since all these variables are essential for estimating markups. Following \cite{de2020rise}, we drop observations with the ratio of sales to COGS in the 1st and 99th percentiles. We construct R\&D following \cite{peters2017intangible}, and treat it as a fixed cost contributing to the future stock of intangible capital. We measure a firm's total capital stock as the sum of physical and intangible capital.\footnote{Intangible capital, such as patents, trademarks, and software, has become an increasingly important input in firms' production processes over the course of our sample period \citep[see][]{crouzet2022}.} We use the K\_INT variable from \cite{peters2017intangible} to measure intangible capital. K\_INT is constructed by applying the perpetual inventory method to Compustat firms. We measure physical capital as Property, Plant, and Equipment (PPEGT.) We provide basic summary statistics in Appendix \ref{app:data}, Table \ref{tab:sumstats_compustat}.\par


\subsection{NIPA Tables and BEA}
We retrieve data on sectoral sales, value added, and GDP from the US National Income and Product Accounts (NIPA.) We use these data to compute Domar weights and input-output multipliers for the different industries. In essence, this procedure assumes the representativeness of Compustat data for the overall population of US firms. While this is a stark assumption, it allows us to develop a proxy for the aggregate profit share from micro-level data. \par
\newpage

\subsection{Production Function Estimation}\label{sec:PFE}
One popular approach for studying market power relies on the production function approach pioneered by \cite{hall1988relation}.\footnote{There is an alternative approach to obtain markups based on demand estimation. The classic reference in this literature is \cite{berry1995auto}. The vast majority of papers in this literature focus on market power in a narrowly defined market \cite[see, for example,][]{nevo2001measuring, berto2007vertical}. One exception in this literature is \cite{dopper2023rising}, who estimate markups for more than 100 distinct product categories in the United States from 2006 to 2019 using Kilts Nielsen Retail Scanner Data.} \citeauthor{hall1988relation}'s seminal contribution was to derive an expression for the markup by taking the first-order condition of a cost-minimizing producer with respect to a variable input, assuming the producer acts as a price taker in that input's market. Formally, he showed that the markup can be computed as
    \begin{align}\label{eq:markup_Hall}
        \mu &= \frac{\theta_j}{\alpha_j},
    \end{align}
where $\theta_j\equiv (\partial y/\partial x_j)$ is the elasticity of output with respect to variable input $j$, and $\alpha_j\equiv p_jx_j/(py)$ is the revenue share of input $j$. \par 

From equation (\ref{eq:markup_Hall}), it is clear that in order to obtain an empirical counterpart of the markup, one needs data on revenue and input spending $\{py, p_jx_j\}$, as well as on a variable input's output elasticity $\theta_j$. The main empirical challenge in markup estimation is obtaining reliable estimates of the variable input's output elasticity, $\hat\theta_j$.\footnote{More recently, the production function approach has been extended by \cite{brooks2021exploitation} and \cite{yeh2022monopsony} to allow for market power in input markets. The novel insight of these studies is to show that markdowns can be estimated using a two-step procedure where markups are obtained in the first step.}  In this respect, the literature has benefited from the work of economists who have developed specialized econometric methods tailored to this purpose.\footnote{See, for instance, \cite{olley1996}, \cite{levinsohn2003estimating}, \cite{wooldridge2009estimating}, \cite{loecker2012markups}, \cite{gandhi2014identification}, and \cite{ackerberg2015identification}.} \par 
\medskip 

\noindent \textbf{Control Function Approach}. We estimate elasticities using a control function approach. We assume that producers generate output $y$ using a Cobb-Douglas technology that takes as arguments one variable input $\ell$ and one quasi-fixed input $k$. The associated log regression is
    \begin{align}\label{eq:CD_prod_reg}
        \log y_{it} = \theta_\ell\log \ell_{it} + \theta_k\log k_{it} + \omega_{it} + \varepsilon_{it},
    \end{align}
where $(\theta_\ell, \theta_k$) are output elasticities, $\omega_{it}$ is the idiosyncratic productivity, which is observed by the producer (\textit{potentially} prior to making input choices) but not by the econometrician, and $\varepsilon_{it}$ is an unanticipated shock to output or productivity which is observed by neither the econometrician nor the producer and may simply capture measurement error. \par 

The presence of time-varying firm-level productivity, $\omega_{it}$, raises standard concerns regarding simultaneity bias. To address this source of bias, we follow \cite{ackerberg2015identification} and estimate the production function (\ref{eq:CD_prod_reg}) using a two-step procedure based on a control function for the productivity process.\footnote{\cite{ackerberg2015identification} argue that the first stage of the estimation procedures proposed by \cite{olley1996} and \cite{levinsohn2003estimating} may suffer from a functional-dependence problem that prevents identification of the variable input's elasticity. To solve this issue, they propose inverting conditional (rather than unconditional) demand functions for the proxy variable and then estimating the variable input's elasticity (along with all other production-function parameters) in the second stage.} In the first stage, we clean measurement error and the unanticipated productivity shock from output using a non-parametric regression. Then, we estimate production-function parameters in the second stage using appropriate moment conditions via GMM. \par 

More specifically, we estimate elasticities using \cite{olley1996}'s methodology coupled with \cite{ackerberg2015identification}'s correction. Our baseline estimates rely on Cobb--Douglas technologies, allow elasticities to be time-varying, and impose common elasticities across 2-digit NAICS industries, which we operationalize by running separate regressions for each industry. In each case, in the first stage of the estimation procedure, we clean out measurement errors and unanticipated shocks by regressing the outcome variable on a third-degree polynomial of the state, free- and proxy variables. In the second stage, we recover the underlying productivity measure and use an AR(1) process for productivity to estimate elasticities using investment as a proxy variable to instrument productivity. When we run our regression (\ref{eq:CD_prod_reg}), we use the capital stock from the previous period since Compustat records end-of-period capital assets. Table \ref{tab:CFA_structure} summarizes our estimation procedure.\footnote{In our baseline specification, the stock of capital is the sum of physical- and intangible capital. In Appendix \ref{app:no_itan}, we show that none of our main empirical results are driven by the inclusion of intangibles.}\par 
\begin{table}[h!]
    \centering
    \caption{Estimation Details in the Application of the Control--Function Approach.}
    \begin{tabular}{l | c}
    \toprule 
    \textbf{Technology} & Cobb--Douglas\\ 
    \textbf{Elasticities} & Time-varying, 9-year rolling windows \\
    \textbf{Method} & \multicolumn{1}{c}{\cite{olley1996}} \\
    \textbf{Productivity process} & \multicolumn{1}{c}{AR(1)} \\ 
    \textbf{Degree of polynomial} &\multicolumn{1}{c}{3rd} \\
    \textbf{\cite{ackerberg2015identification}'s correction} & \multicolumn{1}{c}{\checkmark} \\
    \textbf{Deflated variables} & \multicolumn{1}{c}{\checkmark} \\
    \cdashline{1-2}
    \textbf{Outcome}  & \multicolumn{1}{c}{SALE}\\
    \textbf{State}    & \multicolumn{1}{c}{PPEGT + K\_INT } \\
    \textbf{Free}     & \multicolumn{1}{c}{OPEX ($=$ COGS $+$ SG\&A)} \\
    \textbf{Proxy}    & \multicolumn{1}{c}{ICAPT} \\
    \bottomrule 
    \end{tabular}
    \label{tab:CFA_structure}
    \vspace{4pt}
    \begin{minipage}{1\textwidth}
        \scriptsize{\textbf{Table Notes}. Elasticities are estimated using 9-year rolling windows. That is, for year $t$, we use information from years $t-4,t-3,t-2, t-1, t, t+1, t+2, t+3, t+4$. Hence, to estimate elasticities for the period 1970--2020, we need data starting from 1966 to 2024. Data from 1966 to 2020 is available, but data from 2024 is not. Hence, for years 2016--2020, we set the elasticities to be those estimated in 2016--the last year for which a 9-year panel regression can be run.}  
    \end{minipage}
\end{table}

Our estimation procedure provides elasticity estimates for each industry. We remove outliers from our sample of estimated elasticities. We consider an elasticity to be an outlier if it is more than one standard deviation away from its group mean, where a group is a 2-digit NAICS industry. We set outliers to missing values and then replace them by linearly interpolating each group's input elasticities. Lastly, we winsorize elasticities at the 5th and 95th percentile. \par 

Equipped with elasticity estimates, we compute scale elasticities at the industry level as the sum of input elasticities. We obtain firm-specific returns to scale by multiplying the industry-level scale elasticity with a firm-specific fixed-cost adjustment factor. We compute markups for each producer as $\mu_i = \hat\theta_{j\ell}\alpha_{i\ell}^{-1}$, where $\hat\theta_{j\ell}$ is the estimated elasticity of the variable input for producer $i$ which belongs to industry $j$. In section \ref{sec:empirics}, Table \ref{tab:sectoral_estimates}, we report markup and return to scale estimates at the sectoral level for 2019. The evolution of variable input elasticities, returns to scale, and markups are provided for each industry in Appendices \ref{app:elast_het}--\ref{app:markup_het}. \par 
\bigskip 

\noindent \textbf{Revenue vs. Output Elasticities}. Because only revenue data are available, we cannot estimate output elasticities. Instead, we estimate revenue elasticities. As first noted by \cite{klette1996inconsistency}, using revenue elasticities generally results in biased markups. In Appendix \ref{app:revenue_bias}, we show that although using revenue elasticities may bias our markup estimates, our estimates of profit rates and the profit share are unaffected. In that appendix, we also discuss how recent criticisms raised by \cite{doraszelski2020inconsistency}, \cite{bond2021some} and \cite{de2022hitchhiker} affect our results. \par 
 

\section{Empirical Results}\label{sec:empirics}
In this section, we compute the micro-aggregated profit share using the theoretical results presented in Section \ref{sec:theory}. We benchmark our estimates of the micro-aggregated profit share to macro estimates and discuss their macroeconomic implications. We also provide some results at the industry level. \par 

Figure \ref{fig:agg_markup} plots the time series of the components that make up the micro-aggregated profit share as given by equation (\ref{eq:profshare_nomonop}); that is, assuming no monopsony power.\footnote{As noted earlier, future versions of the paper will include empirical results with monopsony power. The reason why we cannot offer such results with Compustat is that we have no data on materials and, hence, we cannot estimate markdowns. However, we do not expect very different estimates for the profit share since our current markup estimates are a composite of monopoly and monopsony power. Nevertheless, we recognize that accounting for monopsony can potentially affect our markup estimates.} Panel A displays three series: the aggregate markup, aggregate returns to scale, and the aggregate scale elasticity. Panel B depicts the sales-weighted covariance of firm-level returns to scale with inverse markups.  \par 
\begin{figure}[h!]
    \centering
    \caption{Aggregate Markup, Aggregate Returns to Scale, and its Covariance.}
    \includegraphics[scale=0.4]{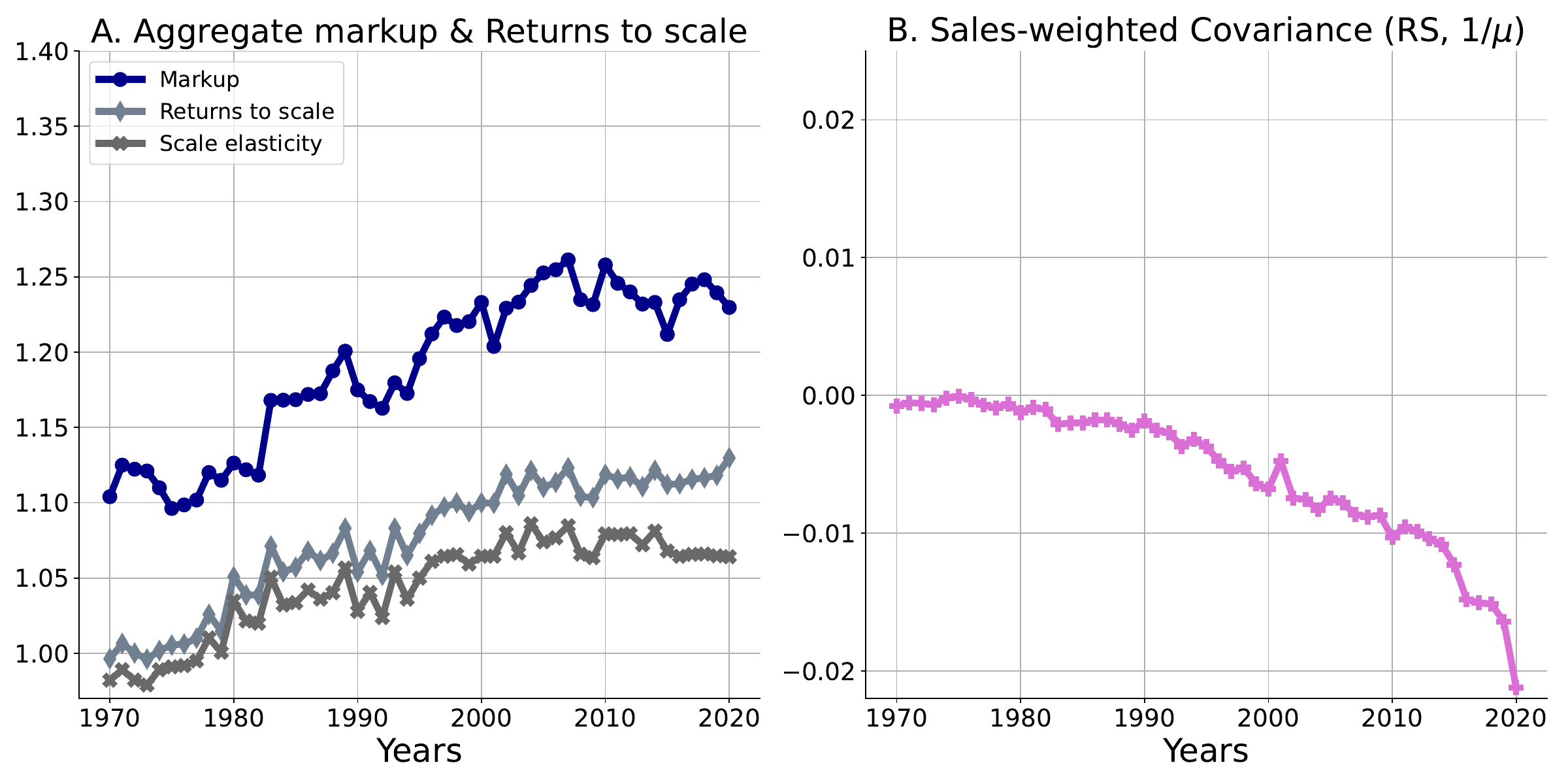}
    \label{fig:agg_markup}
    \begin{minipage}{1\textwidth}
        \scriptsize{\textbf{Figure Notes}. Firm-level data comes from US Compustat, 1970--2020. }  
    \end{minipage}
\end{figure}

The aggregate markup has increased from 1.10 in 1970 to 1.23 in 2020; that is, from 10\% of price over marginal cost to 23\%. Trends are the same as in \cite{de2020rise}, but levels are much lower. As noted earlier, this is mostly because the harmonic sales-weighted markup is much lower than the sales-weighted markup. Another reason why our markup estimates differ from theirs is that we categorize selling, general and administrative expenditures (SG\&A,) as well as the costs of goods sold (COGS,) as variable inputs, in line with \cite{traina2018aggregate}, while \citeauthor*{de2020rise} treat SG\&A as fixed costs. We also define the stock of capital to be a composite of physical- and intangible capital while they only consider physical capital.\footnote{In Appendix \ref{app:deu_rep} we show that making the same data choices as \citeauthor{de2020rise}, we replicate their results. We also compare our markup series to theirs, and explain the differences between them.} Looking at returns to scale, we can see that returns to scale have increased from 1.00 in 1970 to 1.13 in 2020. Half of this increase is explained by a rising scale elasticity, which has gone up from 0.98 in 1970 to 1.06 in 2020, and the remaining part by the fixed-cost adjustment factor, $(\text{TC}-\text{FC})/\text{FC}$, which has risen from 1.02 in 1970 to 1.06 in 2020. The divergence between the two series, returns to scale and the scale elasticity, thus reflects the increase in R\&D spending, which we treat as an explicit fixed cost. \par  

The covariance between returns to scale and inverse markups is slightly negative at the beginning of the sample and becomes more negative as time passes. We find this pattern to be reasonable. A negative covariance between inverse markups and returns to scale means that industries with higher fixed costs tend to charge higher markups.\par 

Equipped with the input-output multiplier---that is, the ratio of sales to GDP---and the measures provided in the previous figure, we compute the micro-aggregated profit share using equation (\ref{eq:profshare_nomonop}).\footnote{In Appendix \ref{app:IO_mult} we show that the input-output multiplier in the US---the ratio of sales to GDP---is mildly procyclical and roughly stable at around 1.8.} The results are depicted in Figure \ref{fig:microprofitshare}, which also shows the profit share that would obtain under zero covariance, and zero covariance and constant returns to scale. \par 
\begin{figure}[h!]
    \centering
    \caption{The Micro--Aggregated Profit Share in the United States, 1970--2020.}
    \includegraphics[scale=0.4]{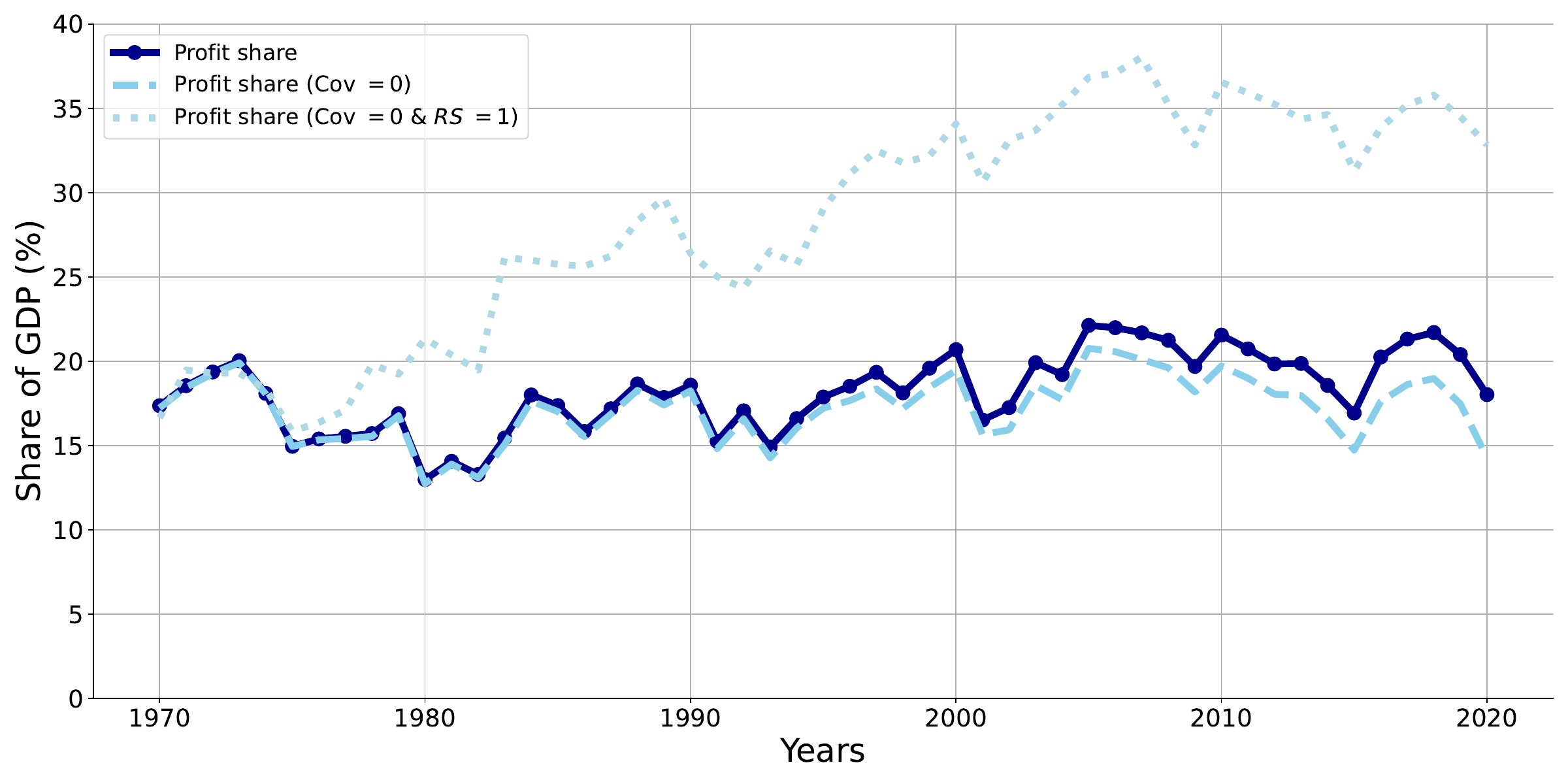}
    \label{fig:microprofitshare}
    \begin{minipage}{1\textwidth}
        \scriptsize{\textbf{Figure Notes}. Firm-level data comes from US Compustat, and data on value-added and GDP comes at the aggregate level from the BEA. The micro-aggregated profit share is expressed in terms of profits divided by value added.}  
    \end{minipage}
\end{figure}

Figure \ref{fig:microprofitshare} indicates that the profit share in the United States has been constant at around 18\% of GDP over the past fifty years.\footnote{Because of double marginalization, a profit share of 18\% implies that the average firm in the United States has a profit rate of about 10\%.} We note that our series likely provides an upper bound for the profit share for several reasons. First, Compustat data covers only publicly-listed firms, which are not representative of the US population of firms and likely have higher profit rates than smaller firms. Second, if intangibles or fixed costs are underestimated, the profit share will be overestimated. We have good reasons to believe that fixed costs, as well as organizational- and intangible capital, are underestimated as we will argue shortly on. The series in Figure \ref{fig:microprofitshare} which impose zero covariance and constant returns to scale stress the importance of accounting for both of these terms, especially non-constant returns to scale. If we had estimated the profit share imposing constant returns to scale, we would have concluded that the profit share in 2020 was twice as large as we report. \par
\bigskip

\noindent \textbf{Decomposing the Profit Share in Several Market-Power Forces}. An advantage of using the microeconomic approach (instead of the macroeconomic one) to compute the profit share is that it allows us to understand the determinants of the profit share. More specifically, the micro-aggregated profit share can be decomposed into fractions of value added pertaining to monopoly rents, fixed costs, and non-linearities. That is, 
    \begin{align}
        \Lambda_\Pi &= \chi\left(1 - \frac{\overline{\text{RS}}}{\overline\mu_{hsw}} - \text{Cov}_\omega\left[\text{RS}, \frac{1}{\mu}\right]\right) \tag{\ref{eq:profshare_nomonop}} \\ 
        &= \underbrace{\chi\left(1-\frac{1}{\overline\mu_{hsw}}\right)}_{\text{Monopoly rents}} + \underbrace{\chi\big(1-\overline{\text{RS}}\big)}_{\substack{\text{Fixed costs and} \\ \text{changing technology}}} + \underbrace{\chi\left\{\left(\frac{1}{\overline\mu_{hsw}}-1\right)\big(1-\overline{\text{RS}}\big) - \text{Cov}_\omega\left[\text{RS}, \frac{1}{\mu}\right]\right\}}_{\text{Non-linearities}}
    \end{align}
In Figure \ref{fig:profshare_decomp_mu_RS_cross}, we offer the results of this decomposition. The term ``monopoly rents'' has increased from 17\% to 33\% of aggregate value added because producers have increased their markups over time. The term ``fixed costs and changing technology'' has decreased from 0\% to $-$23\% of value added because of increasing fixed costs and a rising scale elasticity. And the term ``non-linearities'' has increased from 0\% to 8\% because the aggregate markup and aggregate returns to scale are both positive and increasing over time, and also because the covariance term is increasingly negative. As pointed out earlier, the profit share has been roughly constant at around 18\% of GDP, with movements around the trend reflecting cyclical variation. Thus, although market power increased, profitability did not because the increase in monopoly rents was entirely offset by rising fixed costs and changing technology. \par 
\begin{figure}[h!]
    \centering
    \caption{US Profit Share Decomposition into Sources of Market-Power, 1970--2020.}
    \includegraphics[scale=0.38]{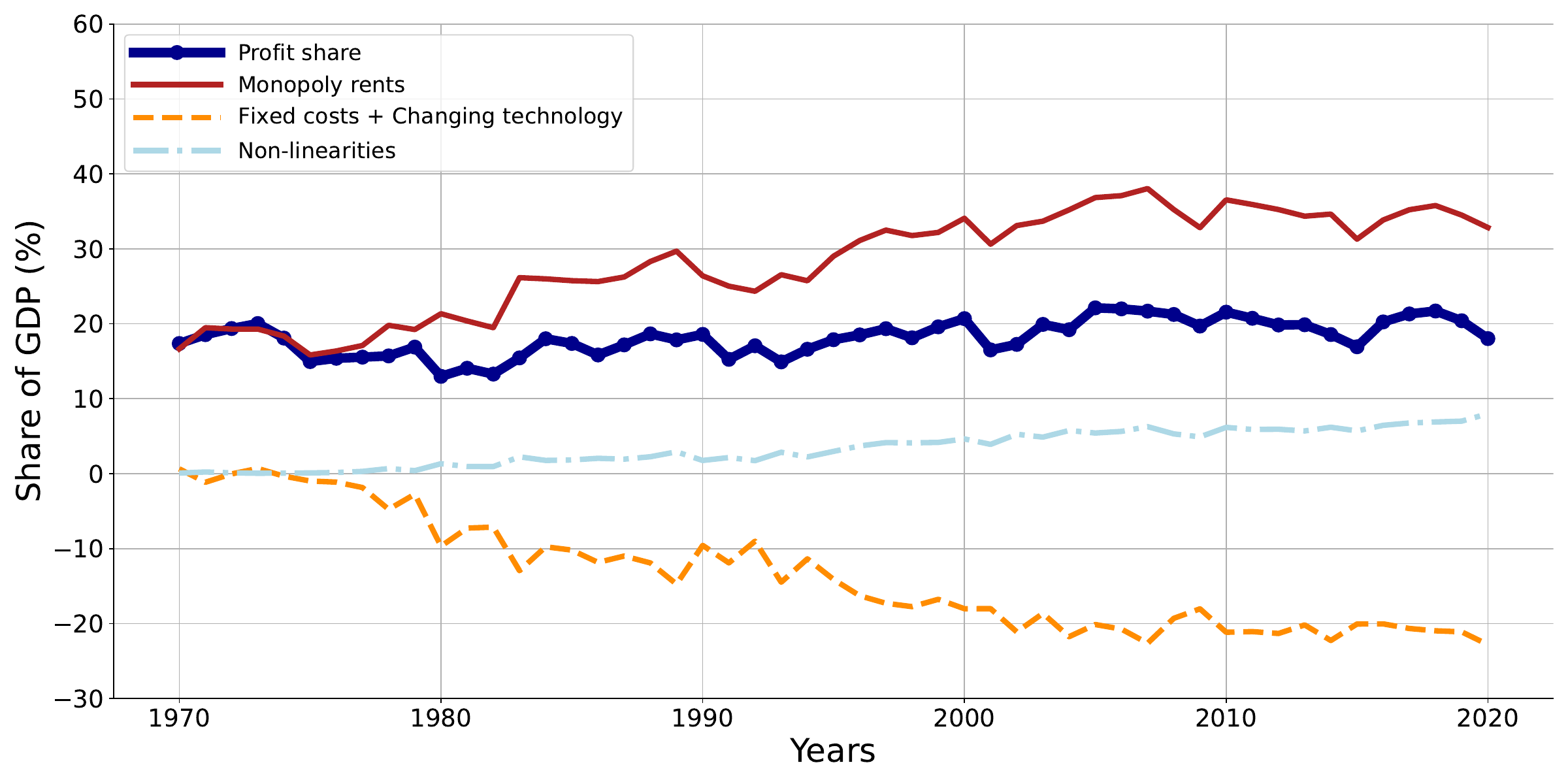}
    \label{fig:profshare_decomp_mu_RS_cross}
\end{figure}

\noindent \textbf{Benchmarking our results}. Figure \ref{fig:microprofitshare_benchmarking} benchmarks our estimates of the profit share to the profit share of \cite{karabarbounis2019accounting} using the macro approach.\footnote{We take the profit share from ``Case $\Pi$'' in their paper. To be clear, they refer to this series as the ``factorless income'' share.} \par 
\begin{figure}[h!]
    \centering
    \caption{The Profit Share in the United States: Micro vs. Macro Estimates.}
    \vspace{-10pt}
    \includegraphics[scale=0.4]{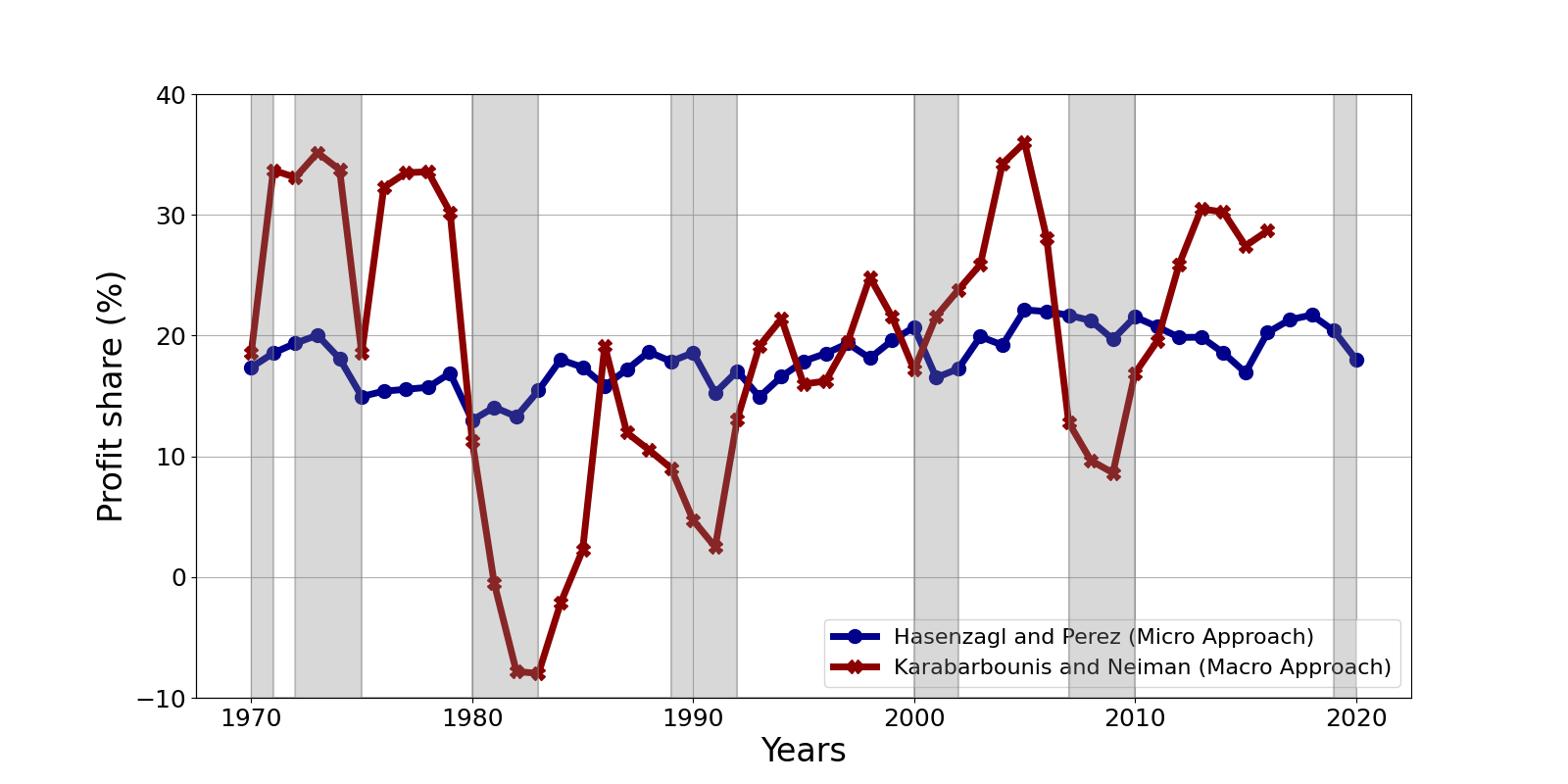}
    \label{fig:microprofitshare_benchmarking}
    \begin{minipage}{1\textwidth}
        \scriptsize{\textbf{Figure Notes}. The micro-aggregated profit share is computed using firm-level data from Compustat, assuming representativeness of these data, and using sales and GDP from the BEA. The \citeauthor{karabarbounis2019accounting}'s series is a non-smoothed version of the one reported in their paper and has been constructed using their replication files. They construct the profit share using aggregate data from National Accounts and imputing an aggregate user-cost of capital.}  
    \end{minipage}
\end{figure}

Our series of the profit share is a better measure of pure profits than the factorless-income series of \citeauthor{karabarbounis2019accounting}. First, their series is extremely volatile, and second, it implies unreasonably low capital shares. Their estimates of factorless-income in the early 1970s do, for example, suggest a capital share of 0\%, which is not realistic. Our series is more stable and implies more plausible capital shares. In line with \citeauthor{karabarbounis2019accounting}, and consistent with our countercyclical aggregate markup estimates, we find procyclical profit shares. Although our profit share series also fluctuates because of cyclical variation, we find these fluctuations to be within reasonable bounds. Finally, we follow \citeauthor{karabarbounis2019accounting} and compute the correlation between the profit share and the aggregate user cost of capital. With their series, we find a large negative correlation of 0.83; with our series, the correlation is also negative, but about a third in magnitude. As they note, these negative correlations suggest that both profit share series capture some form of unmeasured capital. \par 

We can also compare our results to those of \cite{barkai2020declining}, who also estimates capital- and profit shares for the United States using data from National Accounts. \citeauthor{barkai2020declining} estimates that the profit share has grown from $-$5\% in 1985 to 8\% in 2014. We find his profit share estimates substantially less reasonable than ours. First, \citeauthor{barkai2020declining}'s profit share is negative for roughly twenty years and then increases by fifteen percentage points in a matter of three years. Second, our profit share estimates are consistent with \cite{karabarbounis2023perspectives} estimates of the labor share while \citeauthor{barkai2020declining}'s are not. In 2014, for instance, our profit share is consistent with a labor share of 62\%, as estimated by \citeauthor{karabarbounis2023perspectives}, and a capital share of 20\%. This capital share estimate is five percentage points lower than \citeauthor{barkai2020declining}'s, which means that the implied labor share in \citeauthor{barkai2020declining} differs by six percentage points from that of \citeauthor{karabarbounis2023perspectives}. In Appendix \ref{app:income_shares}, we discuss the implications of our profit share estimates for capital income shares.\par 

\bigskip 

\noindent \textbf{Assessing \cite{de2020rise}'s Estimates}. \citeauthor{de2020rise}'s estimates of markups and profitability have generated substantial debate in academic and policy forums. Scholars have scrutinized their methodology and their estimates. Here, we discuss the most relevant issues in calculating aggregate markups and the implications of their estimates in terms of profitability. \par 

\textit{Weighting firm-level markups}. An important choice is how to aggregate markups. \cite{karabarbounis2019accounting} noted that aggregating markups using cost (rather than sales) weights dramatically changes the magnitude of the markup series. This goes back to the discussion on which notion of aggregate markup is appropriate to use. In this paper, we provide an argument for using the harmonic sales-weighted markup based on our aggregation theorems. As pointed out by \cite*{edmond2023costly}, the harmonic sales-weighted markup is equivalent to a cost-weighted markup when cost elasticities are homogeneous across producers, in which case the point raised by \citeauthor{karabarbounis2019accounting} remains valid.\par 

\textit{Measuring variable costs}. \cite{traina2018aggregate} argues that \citeauthor{de2020rise}'s estimates substantially overstate true markups because they do not correctly account for variable costs.\footnote{\cite{de2020rise} measure variable costs as the costs of goods sold (COGS). \cite{traina2018aggregate} and others argue that variable costs do not only include costs of goods sold but also selling, general, and administrative (SG\&A) expenses. The latter represent indirect inputs to production, such as sales or marketing expenditures, management expenditures, and accounting or legal expenditures.} Failing to properly account for variable costs may affect estimates of the variable input's elasticity and thus impact markup estimates, which are sensitive to these elasticities. Using a more appropriate categorization of variable and fixed costs, \citeauthor{traina2018aggregate} finds a more modest increase in markups. Following a dichotomy of fixed and variable costs similar to \citeauthor{traina2018aggregate}, we find much lower levels for aggregate markups compared to \citeauthor*{de2020rise}. Our results point to an increase in the aggregate markup from 1.13 in 1980 to 1.23 in 2016, whereas \citeauthor{de2020rise} report a hike in markups from 1.21 to 1.61. \par  

In Figure \ref{fig:cogsvsopex}, we plot aggregate markups, returns to scale, and the sales-weighted covariance of firm-level inverse markups with returns to scale using two different categorizations of variable costs, OPEX and COGS. As argued earlier, correctly categorizing variable costs is important for estimating input elasticities, markups, and returns to scale. This can be seen in panels A and B. Panel C displays the covariance term, which we use to compute micro-aggregated profit shares. \par 
\begin{figure}[h!]
    \centering
    \caption{Markups, Returns to Scale, and Variable Costs (OPEX vs. COGS).}
    \includegraphics[scale=0.38]{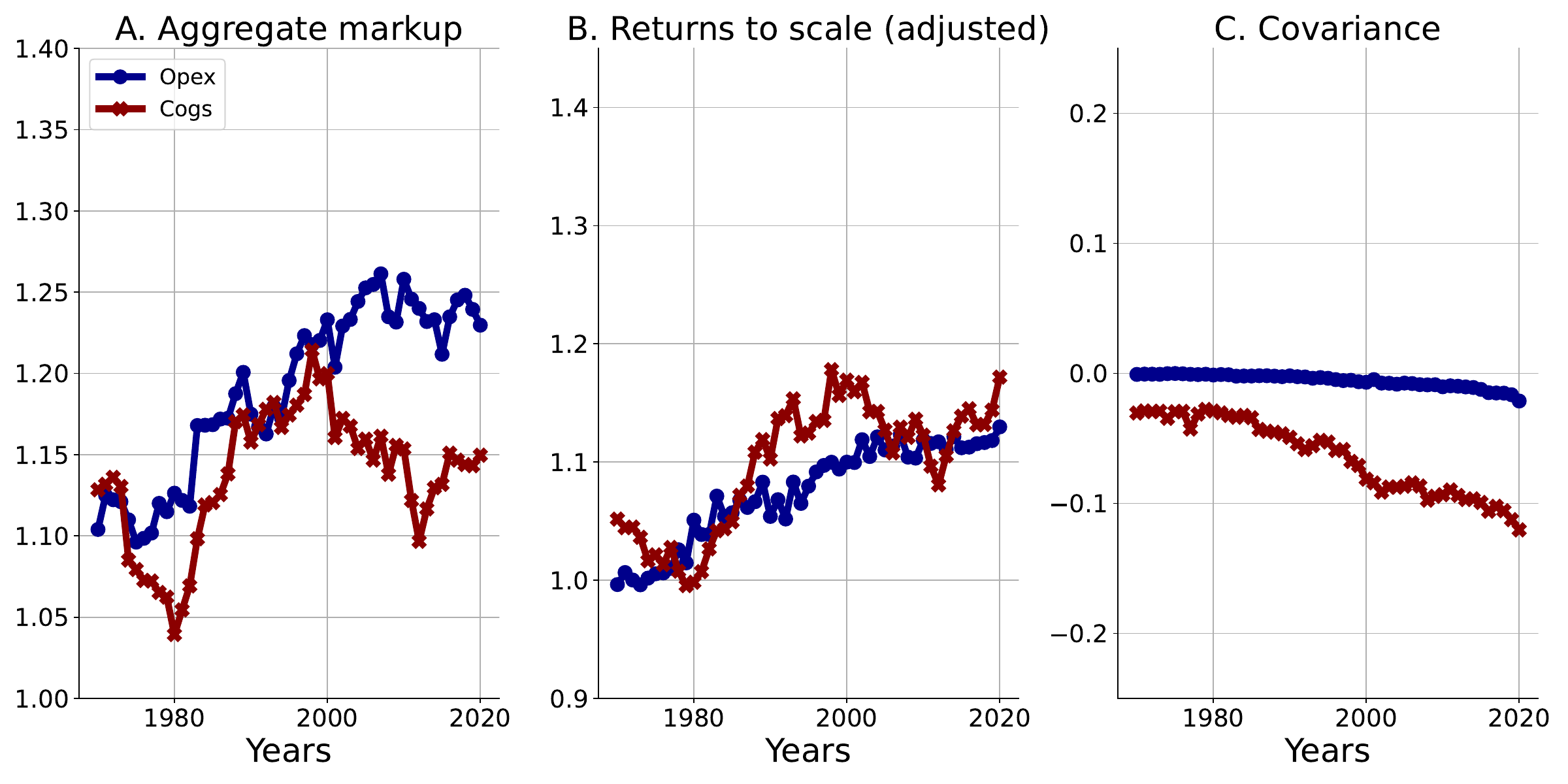}
    \label{fig:cogsvsopex}
    \begin{minipage}{1\textwidth}
        \scriptsize{\textbf{Figure Notes}. Aggregate markups are computed as harmonic, sales-weighted averages of individual firms' markups. The COGS series treats the costs of goods sold as the only source of variable costs. In contrast, the OPEX series considers costs of goods sold and selling, general and administrative expenses, as variable costs. Returns to scale are sales-weighted averages of firm-level returns to scale adjusted for fixed costs. The covariance is a sales-weighted covariance of firm-level inverse markups with returns to scale adjusted for fixed costs.}  
    \end{minipage}
\end{figure}

In this paper, we have dealt with both of these issues, categorizing variable costs as suggested by \cite{traina2018aggregate} and aggregating markups as implied by our aggregation theorem, and have shown that the aggregate markup in the US economy has increased from 1.10 in 1970 to 1.23 in 2020. Discrepancies in markups arising from different categorizations of variable inputs can be seen in Figure \ref{fig:cogsvsopex}. \par 

Finally, we touch upon the implications of \citeauthor{de2020rise}'s estimates of aggregate markups on the aggregate profit share, since this has been an area of debate in the literature. \citeauthor{basu2019price}, \citeauthor{barkai2020declining}, and others have argued that the micro estimates of \citeauthor*{de2020rise} have unreasonable implications for aggregate income shares. Specifically, \citeauthor{basu2019price} has argued that \citeauthor*{de2020rise}'s markup estimates imply an unrealistically high profit share of about 70\% (See Appendix \ref{app:basu_deu}.) Consistently with what we argued in Section \ref{sec:theory}, we have shown that firm-level estimates of markups, once properly measured and aggregated, no longer have unreasonable implications when fixed costs are accounted for. It is also easy to verify that while the categorization of variable costs does matter for markup estimates, it does not matter for estimates of the aggregate profit share, simply because profit rates are the same regardless of whether inputs are categorized as variable or fixed as long as all costs are accounted for. \par  

\subsection{The Profit Share and The User Cost of Capital}\label{sec:implied_r}
A critical input for obtaining profit rates and, thus, the profit share is the user cost of capital. This is true regardless of whether the profit share is computed using the macroeconomic or the microeconomic approach. In this section, we construct the profit share using three different user costs of capital: ours (henceforth, Hasenzagl--Perez,) \cite{de2020rise}'s, and \cite{karabarbounis2019accounting}'s. \par 

Theorem \ref{thm:profshare_decomp} implicitly assumes that we can compute the user cost of capital for each individual producer using its first-order condition for capital. The first-order condition of producer $i$ at time $t$ permits obtaining $i$'s user cost of capital as
    \begin{equation}\label{eq:r_HP}
        r_{it} = \frac{\theta^k_{jt}}{\mu_{it}}\times \frac{p_{it}y_{it}}{k_{it}}
    \end{equation}
where $\theta^k_{j}$ is the output elasticity of capital for the industry $j$ in which producer $i$ operates, $\mu$ is the markup, $py$ are sales, and $k$ is the capital stock. \par 

We believe that allowing for heterogeneous user costs of capital is an attractive feature of our method. First, if the composition of capital is different across producers and different types of capital generate heterogeneous returns, producers' user cost of capital will be different. Second, in the presence of financial frictions, rental rates of capital may be different across producers. Third, risk premia and other wedges may differ across individual producers or sectors. We capture these sources of heterogeneity by backing out the user cost of capital for each producer using its first-order condition. \par  

Equipped with heterogeneous user costs of capital given by equation (\ref{eq:r_HP}), we can compute the aggregate profit share using either Lemma \ref{lemma:microprofitshare} and equation (\ref{eq:DEU_profitrate}) or, alternatively, using equation (\ref{eq:profshare_nomonop}).\par 

\cite{de2020rise}'s user cost of capital is assumed to be equal across \textit{all} producers in the economy and computed as
    \begin{equation}\label{eq:DEU_r}
        r_t = i_t - \pi_t + \delta,    
    \end{equation}
where $i$ is the nominal interest rate, $\pi$ is the inflation rate, and $\delta$ is the depreciation rate.\par 

\cite{de2020rise} compute the user cost of capital using the federal funds effective rate (FF) as a proxy for the nominal interest rate, the growth rate of the GDP implicit price deflator (GDPDEF) as a proxy for the inflation rate, and an exogenous depreciation rate which they set to 0.12 for the entire sample period. Given this user cost, they compute producers' profit rates as
    \begin{equation}\label{eq:DEU_profitrate}
        s_{\pi_{it}} = 1 - \frac{\theta^v_{jt}}{\mu_{it}} - \frac{r_tk_{it}}{p_{it}y_{it}} - \frac{\text{FC}_{it}}{p_{it}y_{it}},
    \end{equation}
where $\theta_{j}^v$ is the elasticity of the variable input for the industry $j$ in which producer $i$ operates, $\mu$ is the markup, $py$ are sales, $k$ is capital, $\text{FC}$ are fixed costs, and $r$ is the common user-cost of capital. Once we have obtained profit rates for all producers, we can compute the aggregate profit share using Domar weights as shown in Lemma \ref{lemma:microprofitshare}.\par 

\cite{karabarbounis2019accounting} compute several user costs of capital. We focus on the user cost of capital calculated under what they label ``Case $R$.'' To calculate the user cost of capital, they use the equation
    \begin{equation}\label{eq:KN_usercost}
        P^QQ = WN + \tilde R^IK^I + \tilde R^NK^N + \Pi^Q, 
    \end{equation}
where $P^QQ$ is the business sector's value added, $WN$ are labor payments, $\tilde R^IK^I$ are capital payments to IT capital, $\tilde R^NK^N$ are capital payments to non-IT capital, and $\Pi^Q$ are business profits. They use $\tilde R^j$ for $j\in\{I, N\}$ to denote the ``revised'' rental rates of IT and non-IT capital and compute these taking $\{P^QQ, WN, K^I, K^N, \Pi^Q\}$ as given so that equation (\ref{eq:KN_usercost}) holds while being consistent with the Hall--Jorgenson formula for the measured (or non-revised) user costs of capital. They use the word ``revised'' because this way of computing the user cost of capital creates a wedge with measured rental rates. As is standard, they attribute the difference between measured- and revised rental rates to time-varying risk premia, adjustment costs, financial frictions, and the like. We borrow the time series of the revised rental rate from their paper and compute the profit share using this user cost for all producers using Lemma \ref{lemma:microprofitshare} and equation (\ref{eq:DEU_profitrate}). \par 

While we recognize that it is difficult to argue for one particular way of computing the user cost of capital, we believe that using the firm's first-order condition with respect to capital presents several advantages over other methods. First, and in contrast to both \citeauthor{de2020rise} and \citeauthor{karabarbounis2019accounting}, we allow for heterogeneous user costs of capital across producers. Second, we do not impose a homogeneous and constant depreciation rate for all producers as \citeauthor{de2020rise}.\footnote{There is substantial evidence that depreciation rates are heterogeneous across capital types (e.g., IT vs. non-IT) and that industries differ in the composition of their capital stocks. There is also evidence that rates of depreciation have increased over time as industries have become more reliant on technology.} Third, an advantage of estimating user costs of capital as we do is that these should be gross of wedges, meaning that risk premia, adjustment costs, financial frictions, and the like should be factored in at the level of each producer. \par 
\newpage 

Figure \ref{fig:usercosts} depicts the three different users costs of capital. The Hasenzagl--Perez series corresponds to the aggregate user cost of capital resulting from aggregating individual users costs.\footnote{One can think of this user cost of capital as the user cost of capital that a representative firm would face for the total capital stock in the economy.} Our series registers abrupt jumps around recessions and times of financial stress (e.g., 1973, 1980, 2000, 2007, 2012,) which is consistent with higher risk premia, adjustment costs, and heightened financial frictions during these times.\footnote{While we find evidence in favor of higher risk premia and heightened financial frictions during recessions, in line with \cite{gilchrist2012credit}, \cite{duarte2015equity}, \cite{caballero2017safe}, \cite{jorda2019rate} and others, and consistently with many theories of financial frictions, there is also evidence on the contrary \citep[see, for example][]{chari2007business, brinca2016accounting}.} The level of our series is close to that of \citeauthor{karabarbounis2019accounting}, which should also capture the aforementioned wedges. The level of the \citeauthor{de2020rise}'s series is generally higher than those of the Karabarbounis--Neiman and Hasenzagl--Perez series because they impose a constant depreciation rate of 0.12.\par 
\begin{figure}[h!]
    \centering
    \caption{Comparison of Aggregate User Costs of Capital Across Studies.}
    \includegraphics[scale=0.4]{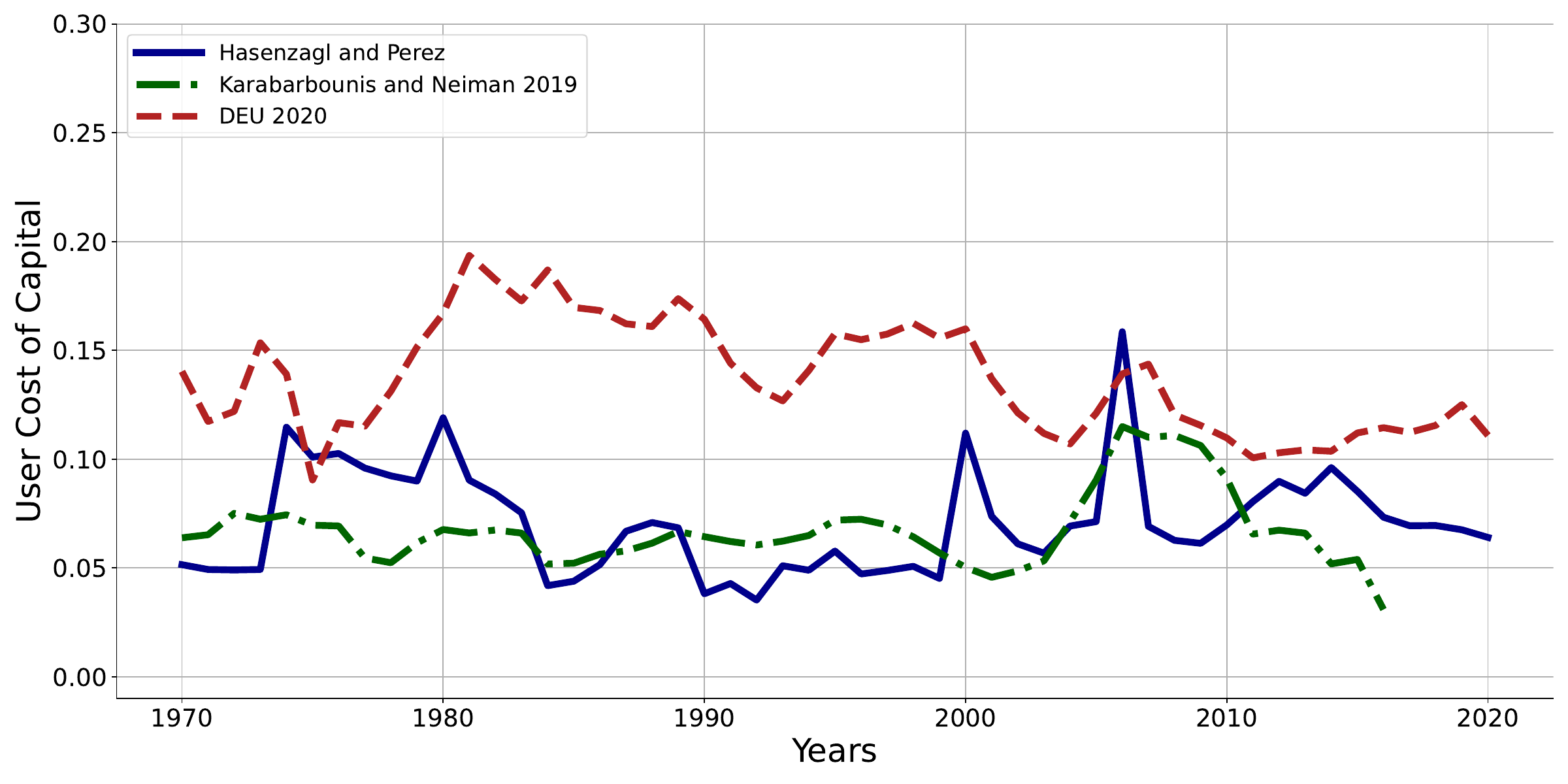}
    \label{fig:usercosts}
    \begin{minipage}{1\textwidth}
        \scriptsize{\textbf{Figure Notes}. The DEU 2020 series is calculated as the federal funds effective rate minus the GDP implicit price deflator plus a exogenous depreciation rate of twelve percent. The Karabarbounis and Neiman 2019 corresponds to the rental rate measure for ``Case $R$'' and is directly taken from their replication package. The Hasenzagl and Perez series is obtained by aggregating individual user cost of capital as implied by our theory, and the individual user cost of capital are recovered exploiting the firms' optimality conditions. All series for the user cost of capital are expressed as rates. }  
    \end{minipage}
\end{figure}

Figure \ref{fig:shares_usercosts} displays three different profit shares. We constructed these series by first computing firm-level profit rates using the different measures of the user cost of capital from figure \ref{fig:usercosts} and then aggregating these profit rates. The series of Hasenzagl--Perez and Karabarbounis--Neiman have similar trends, except for the period 2000--2010.\footnote{Even if the user cost of capital of Hasenzagl--Perez and Karabarbounis--Neiman are similar, profit shares may not be because of the underlying heterogeneity in user costs of capital in Hasenzagl--Perez.} The profit share implied by \citeauthor{de2020rise}'s user cost of capital exhibits a different trend: the profit share increases dramatically from the early 1980s to 2020. \par 
\begin{figure}[h!]
    \centering
    \caption{Profit Shares and User Costs of Capital.}
    \includegraphics[scale=0.4]{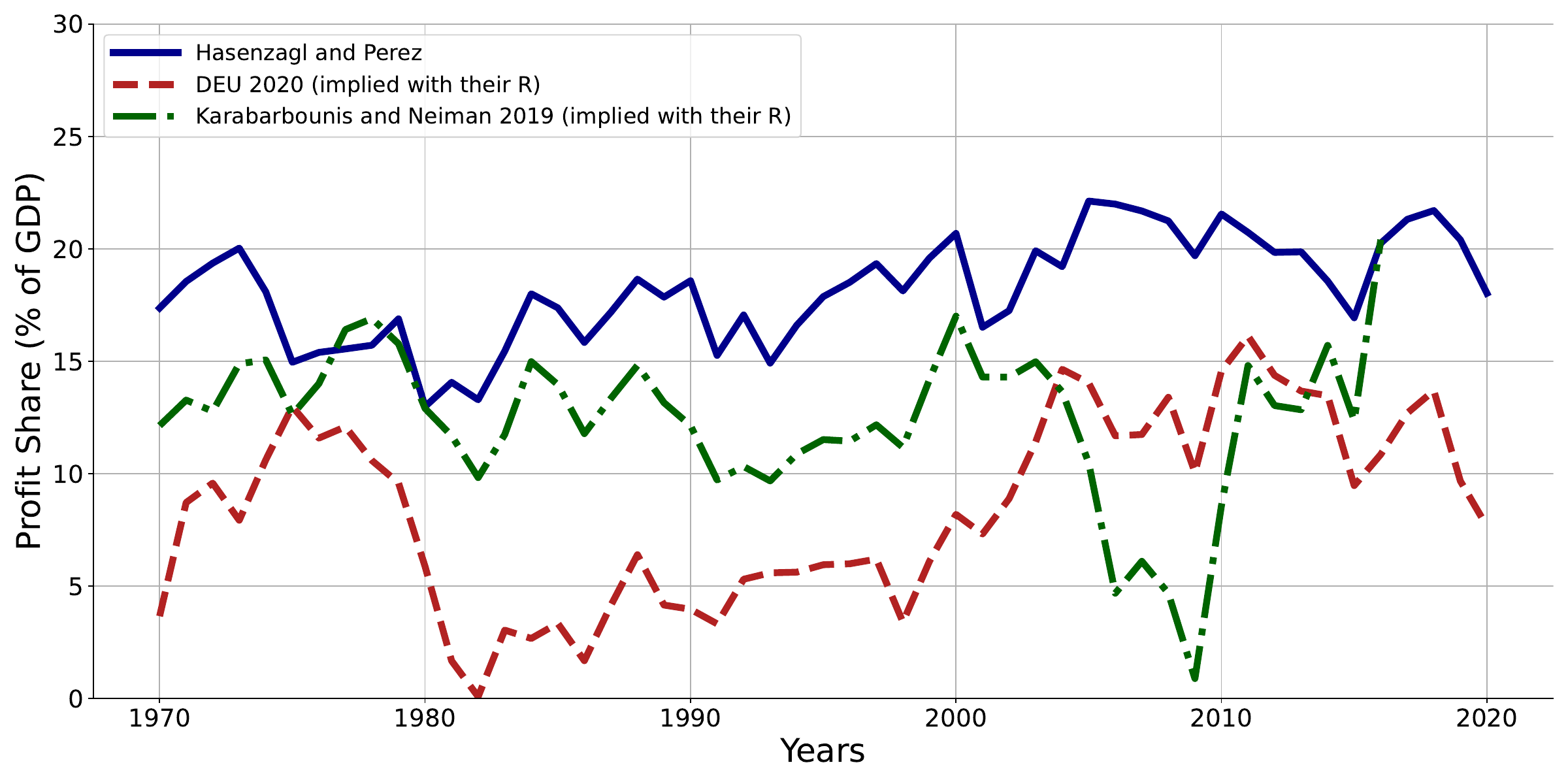}
    \label{fig:shares_usercosts}
    \begin{minipage}{1\textwidth}
        \scriptsize{\textbf{Figure Notes}. Each profit share series is constructed following the microeconomic approach and data from US Compustat. Each of these series is calculated using a different user cost of capital; those presented in Figure \ref{fig:usercosts}. Only the Hasenzagl and Perez series allows for heterogeneity in user costs of capital.  }  
    \end{minipage}
\end{figure}

Figures \ref{fig:usercosts} and \ref{fig:shares_usercosts} together shed light on the controversy surrounding the profit share: what is the most appropriate user cost of capital to use? While we believe that our aggregate user cost of capital displays a reasonable level and presents important advantages over other user cost measures (e.g., it captures risk premia, financial frictions, adjustment costs, and heterogeneity in financing opportunities across producers,) we acknowledge that any measure of the user cost of capital is subject to controversy.\par   

\subsection{Understanding the Rise in the Aggregate Markup}\label{sec:markup_decomp}

Our empirical results so far have shown that the rise in markups has been counteracted by rising fixed costs and changing technologies. But, has the increase in markups been homogeneous across firms? Or does the rise in the aggregate markup mask important micro-level heterogeneity? In this section, we tackle these questions, and also try to understand the compositional forces behind the rise in the aggregate markup. Did the aggregate markup increase because economic activity shifted towards firms with higher markups? Or because incumbent firms increased their markups over time? Or is it because new firms charge relatively higher markups? \par 

Figure \ref{fig:markup_kden} provides kernel densities of unweighted firm-level markups at different points in time. Clearly, the distribution of markups has become more spread over time, signaling an increase in variance, as first reported by \cite{de2020rise}. The flattening of the markup distribution is mostly driven by a fatter right tail. \par 
\begin{figure}[h!]
    \centering
    \caption{Kernel Densities of Unweighted Firm-level Markups.}
    \includegraphics[scale=0.48]{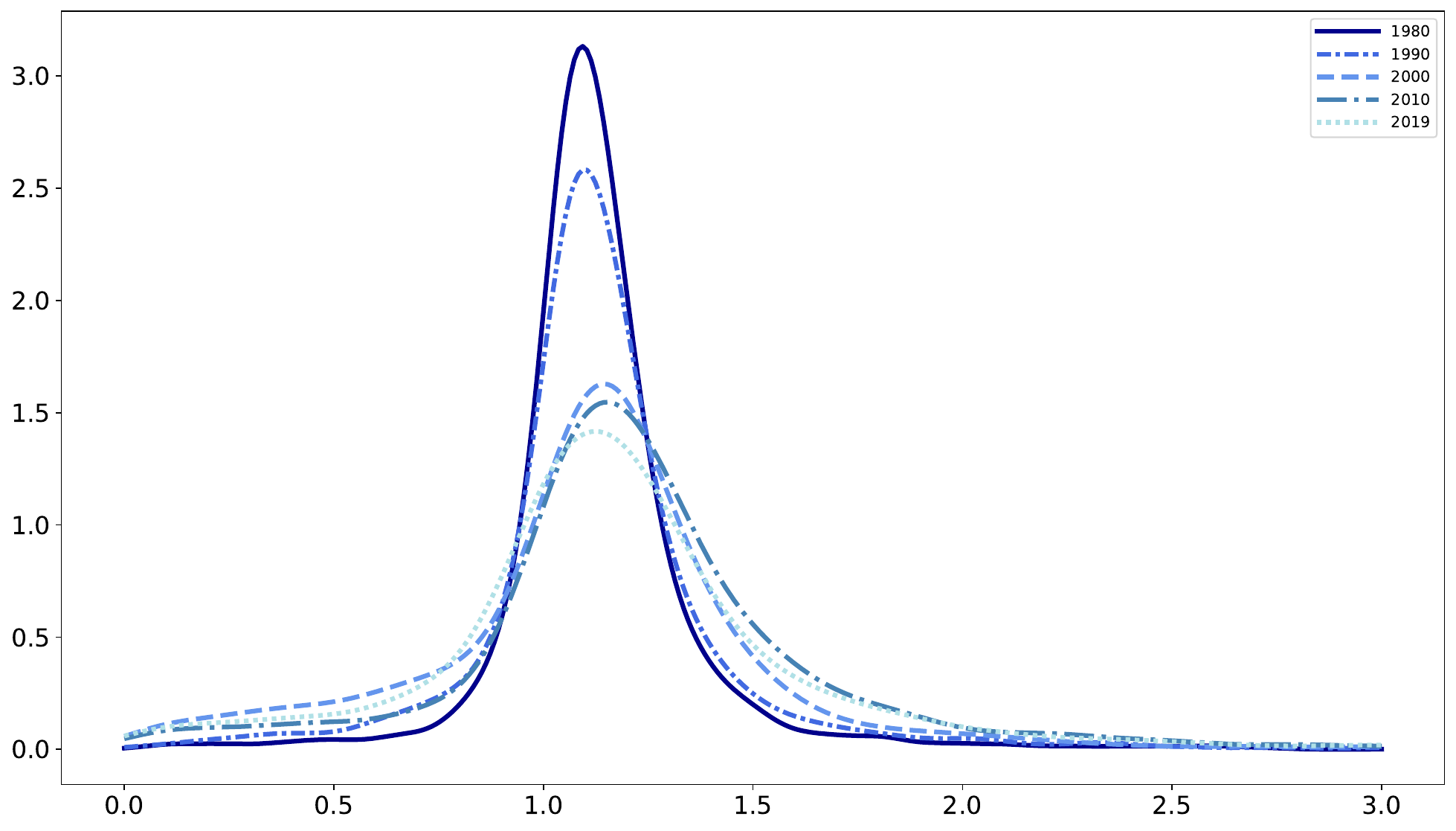}
    \label{fig:markup_kden}
    \begin{minipage}{1\textwidth}
        \scriptsize{\textbf{Figure Notes}. Firm-level markups are calculated using data from US Compustat. Both the costs of goods sold and selling, general and administrative expenses, are treated as variable costs.  }  
    \end{minipage}
\end{figure}

Figure \ref{fig:markup_pctl} plots the evolution of different percentiles of the distribution of firm-level markups. This figure reveals that the the aggregate markup has increased mostly because firms with markups at the median, and especially above the median, have sharply increased their markups since 1970. This fact is also known from the work of \cite{de2020rise} and others. What it is less known is that firms at the 25th percentile and below have registered markups below unity. This fact is quite interesting because it reveals that a large number of firms has been selling at prices below marginal cost, and thus making economic losses. It is particularly striking how firms at the 10th percentile and below have priced increasingly below marginal cost starting in 1980. Firms can price below marginal cost for several reasons. They could be trying to liquidate inventories before exiting the market, or they could be engaging in predatory pricing to drive competitors out of business, just to name two popular reasons. \par  
\begin{figure}[h!]
    \centering
    \caption{Percentiles of the Markup Distribution.}
    \includegraphics[scale=0.48]{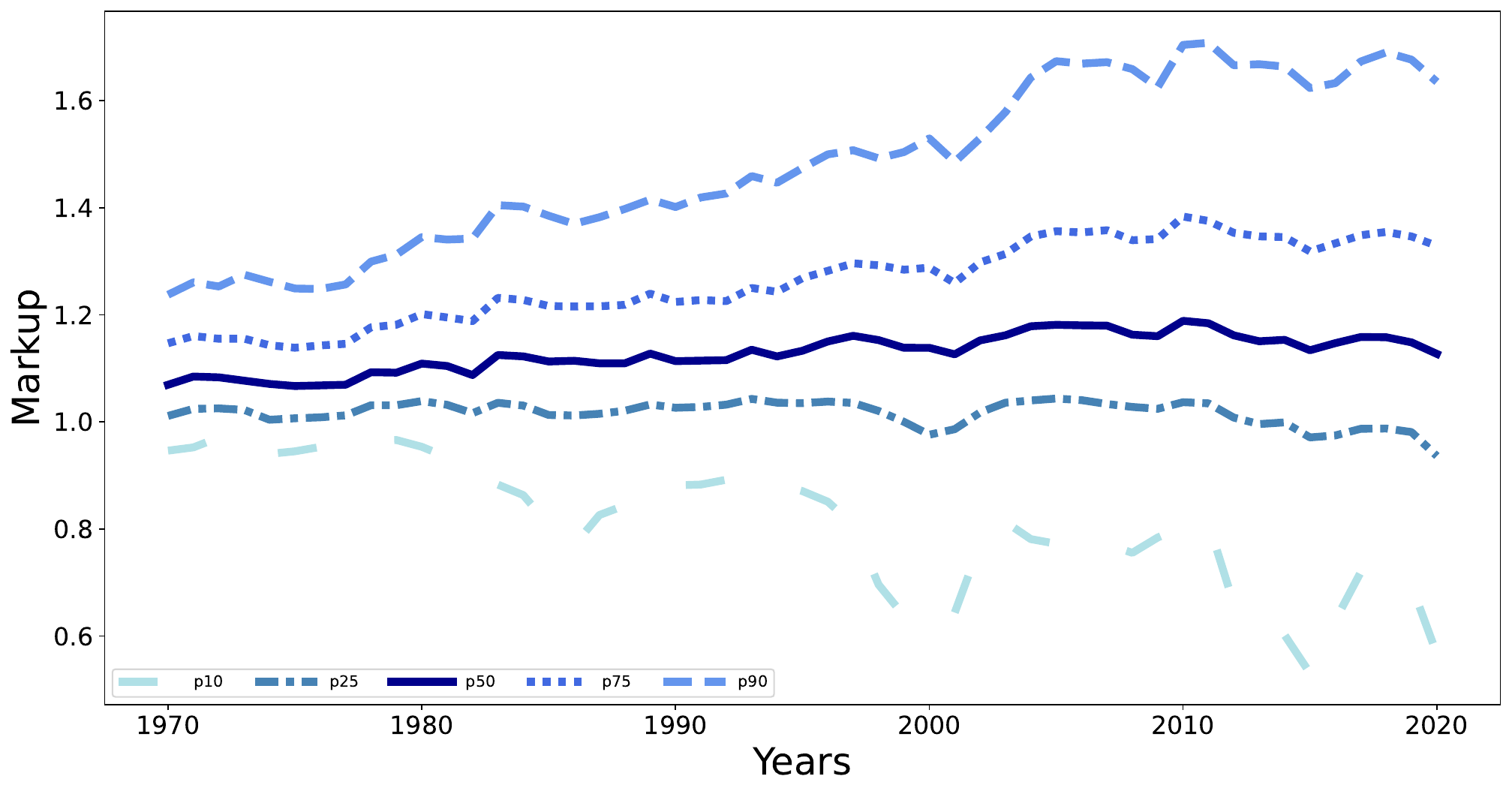}
    \label{fig:markup_pctl}
    \begin{minipage}{1\textwidth}
        \scriptsize{\textbf{Figure Notes}. Firm-level markups are calculated using data from US Compustat. Both the costs of goods sold and selling, general and administrative expenses, are treated as variable costs.  }  
    \end{minipage}
\end{figure}

The increasing heterogeneity in markups across firms, in tandem with rising fixed costs, presents additional challenges for antitrust policy. In the absence of heterogeneity, designing and enforcing antitrust law is relatively easy since policies are naturally targeted. In the presence of heterogeneity and with scarce resources, designing antitrust policies at the level of the firm is generally not feasible. In these situations, increasing heterogeneity makes harder for policymakers to implement policies that combat excessive monopoly power without compromising survival of some firms. 

\par 

Next, we study to what extent the rise in the aggregate markup is due to markup increases within incumbent firms, reallocation of economic activity toward firms with higher markups, and the appearance of firms that charge higher markups than those currently in place. To do so, we provide a statistical decomposition for the harmonic sales-weighted markup, as formalized in the following proposition. \par 

\begin{prop}[Aggregate Markup Decomposition]\label{prop:markup_decomp}
The aggregate markup---the harmonic sales-weighted markup $\overline\mu_{hsw}$---can be decomposed according to
    \begin{align}\label{eq:markup_decomp}
        \Delta\overline\mu_{\text{hsw}} &= \underbrace{- \frac{\sum_{i\in\mathcal C_t}\overline\omega_i\Delta\mu_i^{-1}}{\overline\mu_{\text{hsw}, t}\times \overline\mu_{\text{hsw}, t-1}}}_{\text{within}} \underbrace{- \frac{\sum_{i\in\mathcal C_t}\Delta\omega_i\overline{\mu_i^{-1}}}{\overline\mu_{\text{hsw}, t}\times \overline\mu_{\text{hsw}, t-1}}}_{\text{between}} \nonumber \\ 
        &\qquad \qquad \underbrace{- \frac{\sum_{i\in\mathcal E_t}\omega_{it}\left(\mu_{it}^{-1} - \overline\mu_{hsw}^*\right)}{\overline\mu_{\text{hsw}, t}\times \overline\mu_{\text{hsw}, t-1}} + \frac{\sum_{i\in X_{t-1}}\omega_{it-1}\left(\mu_{it-1}^{-1} - \overline\mu_{hsw}^*\right)}{\overline\mu_{\text{hsw}, t}\times \overline\mu_{\text{hsw}, t-1}}}_{\text{net entry}}
    \end{align}
where $i$ and $t$ index producers and time, respectively, $\mathcal C_t := \big\{i\in\mathcal I_t \ \wedge \ i\in\mathcal I_{t-1}\big\}$ is the set of incumbents, $\mathcal E_t := \big\{i\in\mathcal I_t \ \wedge \ i\notin\mathcal I_{t-1}\big\}$ is the set of entrants, $\mathcal X_{t-1} := \big\{i\in\mathcal I_{t-1} \ \wedge \ i\notin\mathcal I_{t}\big\}$ is the set of exiting firms, $\mathcal I_t = \mathcal C_t \cup \mathcal E_t$, $\mathcal I_{t-1} = \mathcal C_t \cup \mathcal X_{t-1}$, $\overline\mu_{hsw,t}$ is the harmonic sales-weighted markup in period $t$, $\Delta X := X_t - X_{t-1}$, $\overline X:= \frac{1}{2}\big(X_t + X_{t-1}\big)$, and $\overline\mu_{hsw}^* = \frac{1}{2}\big(\overline\mu_{\text{hsw}, t} + \overline\mu_{\text{hsw}, t-1}\big)$. 
\end{prop}
\begin{proof}
    See Appendix \ref{app:markup_decomp}.
\end{proof}

The results of this statistical decomposition are depicted in Figure \ref{fig:markup_decomp}.\footnote{Since the last year of our sample is 2020, 2019 is the last year we use for the markup decomposition. This is because we identify exiting firms in year $t$ using information from years $t$ and $t+1$. Hence, we cannot identify exiting firms in 2020.} 
Here, we decompose the change in the harmonic sales-weighted markup for 1980--1981, 1980--1982, and so on. As mentioned before, the aggregate markup in the United States has increased from around 1.10 in 1970 to 1.24 in 2019; that is, from 10\% of price over marginal cost to 24\%. Firm entry and exit accounts for about 50\% of that increase (i.e., 7\% percentage points,) and the remaining part is accounted for in equal parts by changes in markups within firms and reallocation of economic activity between firms.\par 
\begin{figure}[h!]
    \centering
    \caption{Cumulative Decomposition of the Rise in the Aggregate Markup, 1970--2019.}
    \includegraphics[scale=0.4]{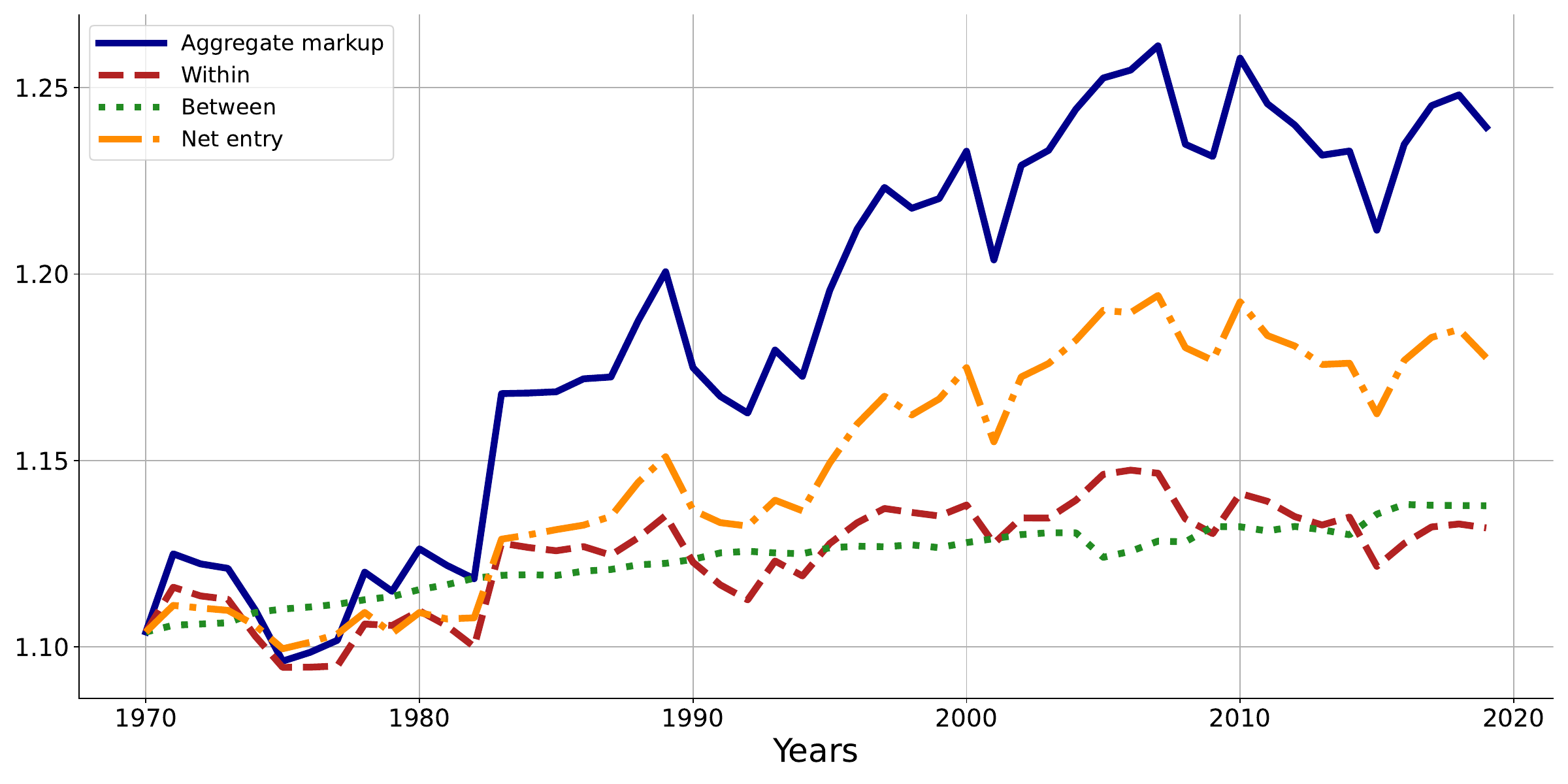}
    \label{fig:markup_decomp}
    \vspace{4pt}
    \begin{minipage}{1\textwidth}
        \scriptsize{\textbf{Figure Notes}. The aggregate markup is the harmonic, sales-weighted markup. The within term captures the aggregate effect of changes in markups within firms, the between term captures how the reallocation of economic activity between firms has contributed to the change in the aggregate markup, and the net-entry component captures the effect of entry and exit. }  
    \end{minipage}
\end{figure}

The fact that the aggregate markup has increased mostly because of net entry is consistent with our narrative that rising fixed costs increase entry barriers, discouraging firm entry and encouraging firm exit. Importantly, we note that the increasing importance of net entry in the markup decomposition is not informative about the number of firms that enter or exit. This term should be interpreted as new entrants having to charge relatively high markups in order to operate without losses, whereas exiting firms are those that charge relatively lower markups.\par 

\subsection{Industry-Level Heterogeneity}
Finally, we investigate the existing heterogeneity in markups, scale elasticities, returns to scale, and profitability across US industries. This exercise serves two purposes. First, it allows us to learn more about the profitability, the markups, and the technologies of particular industries. Second, it provides a sanity check for the plausibility of our empirical results, at least qualitatively. \par 

Table \ref{tab:sectoral_estimates} reports our estimates of aggregate markups and scale elasticities at the level of 2-digit NAICS industries for 2019.\footnote{We have purposely chosen not to report markup estimates in 2020 because of Covid-19.} Markups are harmonic sales-weighted averages of firm-level markups, and scale elasticities are the sum of revenue elasticities. \par 
\begin{table}[h!]
    \centering
    \caption{Markups and Scale Elasticity Estimates from US Compustat, 2019.}
    \begin{tabular}{l c | c c | c c}
    \toprule 
    \multicolumn{2}{l}{}  & \multicolumn{2}{c}{\textsc{Markups}} & \multicolumn{2}{c}{\textsc{Scale Elasticity}} \\
    \textbf{Sector} & \textbf{NAICS code} & \textbf{Mean} & \textbf{Std. dev.} & \textbf{Mean} & \textbf{Std. dev.}  \\
    \midrule 
     Agriculture \& Fishing    & 11     & 1.09 & 0.20 & 1.01 & $-$ \\
     Mining \& Quarrying       & 21     & 1.26 & 0.62 & 1.11 & $-$ \\
     Utilities                 & 22     & 1.25 & 1.14 & 1.07 & $-$ \\
     Construction              & 23     & 1.11 & 0.25 & 1.06 & $-$ \\
     Manufacturing             & 31--33 & 1.34 & 0.47 & 1.09 & $-$ \\
     Wholesale                 & 42     & 0.98 & 0.48 & 1.00 & $-$ \\
     Retail                    & 44--45 & 1.07 & 0.23 & 1.00 & $-$ \\
     Transportation            & 48--49 & 0.98 & 0.66 & 1.00 & $-$ \\
     IT                        & 51     & 1.51 & 0.49 & 1.08 & $-$ \\
     Financial Services        & 52     & 1.19 & 1.07 & 1.05 & $-$ \\
     Real Estate               & 53     & 1.19 & 1.08 & 1.03 & $-$ \\
     Professional Services     & 54     & 1.17 & 0.38 & 1.06 & $-$ \\
     Administration            & 56     & 1.23 & 0.27 & 1.07 & $-$ \\
     Education                 & 61     & 1.16 & 0.58 & 1.01 & $-$ \\
     Health Care               & 62     & 1.17 & 0.38 & 1.11 & $-$ \\
     Arts \& Entertainment     & 71     & 1.18 & 0.36 & 1.04 & $-$ \\
     Hospitality               & 72     & 1.14 & 0.25 & 0.99 & $-$ \\
     Other                     & 81     & 1.15 & 0.36 & 1.03 & $-$ \\
    \bottomrule 
    \end{tabular}
    \label{tab:sectoral_estimates}
    \vspace{4pt}
    \begin{minipage}{1\textwidth}
        \scriptsize{\textbf{Table Notes}. Markups are harmonic sales-weighted markups. Scale elasticities are the sum of revenue elasticities. Since scale elasticities are estimated at the industry level, standard deviations cannot be reported.}  
    \end{minipage}
\end{table}
Since scale elasticities capture implicit fixed costs---those associated with production factors---we expect industries with relatively large minimal input requirements to have relatively high scale elasticities. This is indeed the case for industries like mining, utilities, manufacturing, and IT, all of which are industries with relatively large capital requirements. Similarly, we expect industries with larger scale elasticities to charge higher markups since larger markups are required to operate without loss. In line with our expectations, these numbers indicate that industries with higher scale elasticities charge higher markups (e.g., mining, utilities, manufacturing, and IT.) \par 

Figure \ref{fig:industry_est} offers a scatter plot of industry-level markups and returns to scale, where larger dots reflect higher industry profit shares. Recall that the difference between scale elasticities and returns to scale is that returns to scale take into account explicit fixed costs while scale elasticities do not. Explicit fixed costs include R\&D expenditures toward future productivity or next-period's stock of intangible capital, as well as regulatory compliance costs and the like. Here again we see a positive correlation between returns to scale and industry-level markups: firms in industries with larger overall fixed costs have to charge higher markups to recoup those costs. Figure \ref{fig:industry_est} also shows that industries with higher markups and returns to scale generally have higher profit shares. While our results suggest that industries with higher fixed costs charge markups way above and beyond what is necessary to cover those costs, these results should be interpreted with care. It is possible that these industries charge higher markups today to cover fixed costs that were incurred in the past. \par 
\begin{figure}[h!]
    \centering
    \caption{Industry-level Markups, Returns to Scale, and Profitability in the US, 2019.}
    \includegraphics[scale=0.46]{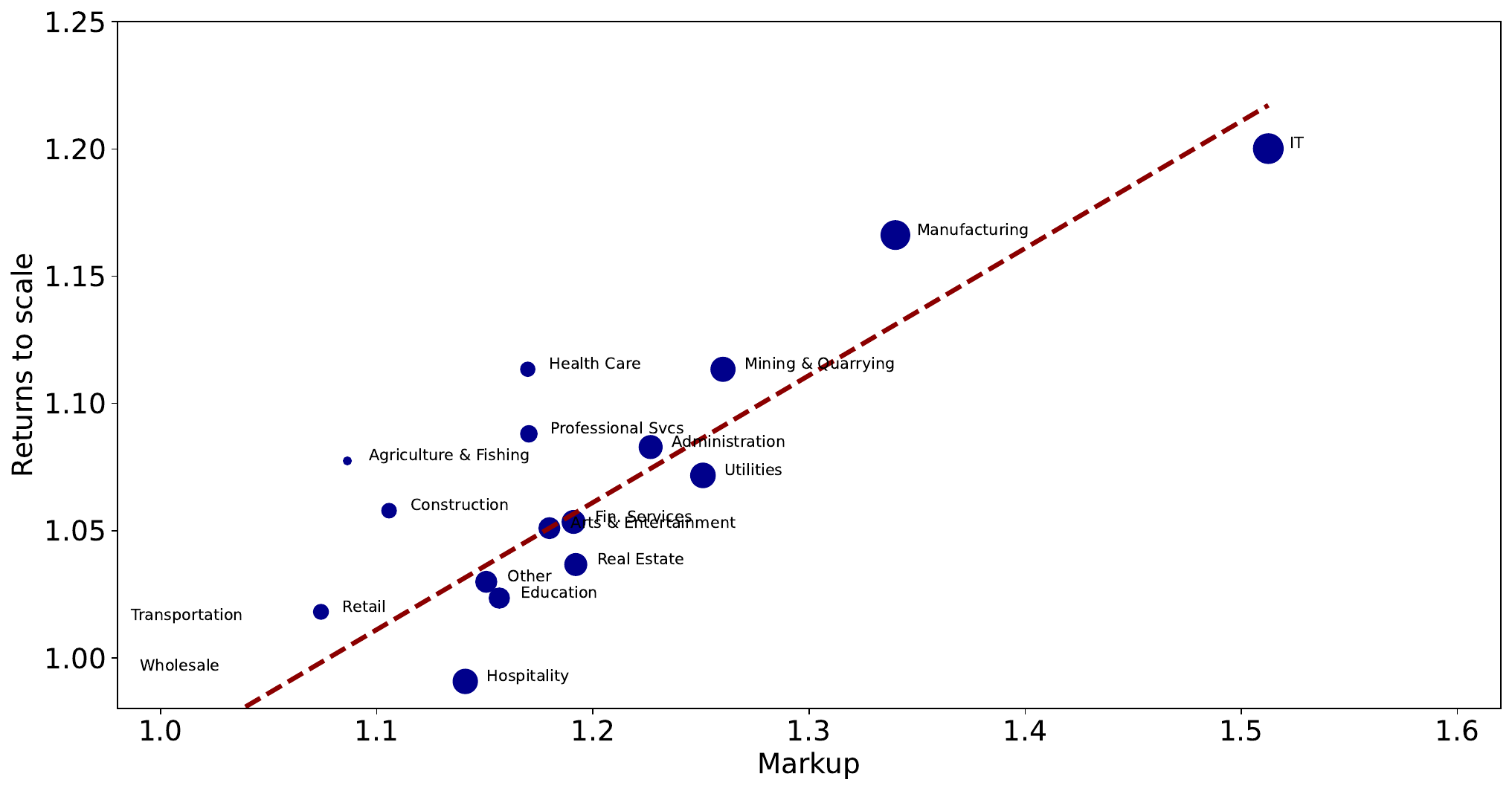}
    \label{fig:industry_est}
    \begin{minipage}{1\textwidth}
        \scriptsize{\textbf{Figure Notes}. Markups are harmonic sales-weighted markups. Returns to scale are sales-weighted scale elasticities adjusted for fixed costs. Larger dots are associated with higher profit shares.}  
    \end{minipage}
\end{figure}

In Appendices \ref{app:elast_het}--\ref{app:markup_het} we report time series of input elasticities, returns to scale, and markups for each industry. \par 

\section{Conclusion}\label{sec:conclusion}
Income shares are important summary statistics for understanding many macroeconomic phenomena, including the global decline in the labor share, the stability of the Kaldor facts, and economic inequality, to name just a few. One such statistic that has historically caught the eye of both economists and policymakers is the profit share. \par 

In this paper, we show that the economy-wide profit share can be constructed by weighting individual profit rates using Domar weights (i.e., producer sales over GDP.) We refer to the profit share so constructed as the \textit{micro-aggregated profit share} and note that no modeling assumptions---only fundamental accounting principles---are needed to establish this result. This procedure for calculating the aggregate profit share brings micro data and macro analysis closer, potentially allowing researchers to leverage interesting features of micro-level datasets for gaining further insights. \par 

Since profits or profit rates are rarely observable---and even when they are, they reflect accounting- rather than economic profits---economic theory is often needed to compute economic measures of profits. Imposing minimal assumptions on producer behavior and technology, we provide a generic expression for a producer's profit rate---defined as profits over sales---in terms of monopoly and monopsony terms. Our formula is a generalization of standard formulas in the literature because it allows for explicit fixed costs and relaxes the assumption of price-taking behavior in factor markets. Apart from standard regularity conditions on production functions, our result only requires producers to minimize costs. \par 

The main theoretical contribution of this paper is to show that the profit share can be expressed in terms of several indicators of aggregate market power---the aggregate markup and an aggregate monopsony term---and a sufficient statistic for production networks that captures double marginalization (i.e., how profits propagate from downstream sellers to upstream suppliers.) Our theorem, which can be applied at any desired level of aggregation, is useful for (i) understanding the determinants of the profit share; (ii) quantifying different sources of profits---that is, monopoly \textit{vis-à-vis} monopsony; (iii) assessing the plausibility of micro-level estimates of markups, markdowns, and returns to scale; (iv) using certain indicators of aggregate market power (e.g., the harmonic sales-weighted markup, sales-weighted markdowns, sales-weighted returns to scale); and (v) calibrating economic models with monopolistic and monopsonistic wedges. \par 

We use the aforementioned theoretical results to guide our empirical analysis and reassess the extent of market power and profitability in the United States from 1970 to 2020 using firm-level data from US Compustat. Our empirical results are the following. First, we find that an economically meaningful notion of aggregate markup---the harmonic, sales-weighted markup---has increased from 10\% of price over marginal cost in 1970 to 23\% in 2020. Second, we find that the rise in the aggregate markup is mostly explained by the appearance of firms that charge higher markups. Reallocation of economic activity towards firms with higher markups and markup increases within firms each explain about 25\% of the rise in the aggregate markup. Third, we find that aggregate returns to scale have risen from 1.00 in 1970 to 1.13 in 2020, reflecting both rising fixed costs and changes in technology. Fourth, despite of the rise in market power, we find that the aggregate profit share in the United States has been constant at 18\% of GDP over the past five decades. We reconcile the increase in market power with the constancy of the aggregate profit share by decomposing the profit share into monopoly rents and fixed costs, both as shares of GDP, and showing that the increase in monopoly rents has been entirely offset by rising fixed costs and changing technology. \par 

Our empirical findings have subtle policy implications. In particular, they suggest that a rise in market power does not necessarily call for stricter enforcement of antitrust law. If firms use their increase in market power to cover rising fixed costs, aggregate profits as a share of GDP remain stable. Although stricter enforcement of antitrust law could lower monopoly rents, it could also cause firms that cannot cover the rising fixed costs to exit the market. Fewer firms means less competition, which translates into higher market power for surviving firms. Thus, stricter enforcement of antitrust law can paradoxically lead to higher market power.\par

Our current empirical analysis is not complete and presents several limitations. We see our estimates of the profit share as providing an upper bound for the profit share since certain forms of capital are not correctly measured, and Compustat data are not representative of the US population of firms. In future versions of the paper, we plan to use data from the US Census of Manufactures to provide a decomposition of this industry's profit share into monopoly and monopsony rents, as well as to assess the representativeness assumption of Compustat firms for this particular sector.  \par 


\newpage 
\bibliography{ref}
\addcontentsline{toc}{section}{References}

\newpage 
\appendix 
\appendixpage
\begin{center}
    \large Thomas Hasenzagl \hspace{50pt} Luis Pérez \\ 
    \vspace{5pt}
    \today 
\end{center}
\listofappendices 
\numberwithin{equation}{section} 
\newpage 

\section{Omitted Proofs}\label{app:proofs}

\subsection{Proof of Lemma \ref{lemma:vaprofitshare}}
\begin{lemma}[The Micro--Aggregated Profit Share]\label{lemma:vaprofitshare}
The aggregate profit share can be constructed from micro-level data by aggregating individual producers' profit rates, defined as profits over value added, using valued-added weights.
\end{lemma}

\begin{proof}
Let $\Pi$ denote aggregate profits. For each producer $i\in\mathcal I$, $\pi_i$ denotes profits, $p_iy_i$ are sales, and $s_{\pi_i}:= \pi_i /(\text{va}_i)$ is $i$'s profit rate, where $\text{VA}_i$ is the value added generated by producer $i$---that is, sales minus material costs. Then, we have
\begin{align*}
    \Lambda_\Pi^{\text{Macro}} &= \frac{\Pi}{\text{GDP}} \\
    &= \frac{\sum_{i\in\mathcal I}\pi_i}{\text{GDP}} \\
    &= \sum_{i\in\mathcal I} \frac{\text{VA}_i}{\text{GDP}}\times \frac{\pi_i}{\text{VA}_i}\equiv \Lambda_\Pi^{\text{Micro, VA}}.
\end{align*}
\end{proof}

\subsection{Proof of Proposition \ref{prop:indprofitrate}}\label{app:prop1}
\noindent \textbf{Proposition \ref{prop:indprofitrate}} (Profit Rates, Monopoly, and Monopsony). \textit{Under the assumptions of cost-minimizing behavior, the existence of a continuously differentiable and quasiconcave production function, and monopsony power in factor markets, a producer's profit rate, defined as profits over sales, can be written as}
    \begin{align*}
        s_{\pi} = 1 - \frac{\text{RS}}{\mu} = 1 - \frac{\text{SE}^{\text{adj}}}{\mu} + \frac{\mathcal M}{\mu}, \tag{\ref{eq:indprofitrate}}
    \end{align*}
\textit{where $\text{RS} = \text{SE}^{\text{adj}} - \mathcal M$ are the returns to scale, $\text{SE}^{\text{adj}}$ is the scale elasticity of the production function adjusted for fixed costs, $\mathcal M$ is a monopsony term capturing market power in factor markets, and $\mu$ is the markup of price over marginal cost.} \par 

\textit{The scale elasticity adjusted for fixed costs is defined as
}
    \begin{align*}
        \text{SE}^{\text{adj}} := \text{SE}\times\left(\frac{\text{TC}}{\text{TC} - \text{FC}}\right),
    \end{align*}
\textit{where $\text{SE}$ is the scale elasticity, $\text{TC}$ are total costs, and $\text{FC}$ are explicit fixed costs.} \par

\textit{The monopsony term is given by}
    \begin{align*}
        \mathcal M := \left(\frac{\text{TC}}{\text{TC} - \text{FC}}\right)\sum_{j\in \mathcal F}\theta_{j}(1-\nu_{j}), \tag{\ref{eq:monopsony_term}}
    \end{align*}
\textit{where $\theta_{j}\equiv \partial F/\partial x_{j} \times x_{j} / y$ is the elasticity of output with respect to input $j$, and $\nu_{j}$ is the markdown on factor $j\in\mathcal F$, where the markdown is defined as the ratio of input $j$'s rental rate to its marginal revenue product; that is, $\nu_{j}:= w_j(x_{j})/\text{MRP}_{j}$. }

\begin{proof}
The cost minimization problem of a producer may be written as
    \begin{align*}
        \min_{\bm{x}_j\geq\bm{0}}\quad & \text{TC} \equiv \sum_{j\in\mathcal N}w_j(x_{j})x_{j} + \text{FC} \\
        \text{s.t. }\quad & y = F\left(\{x_{j}\}_{j\in N}; A\right) \geq \overline y,
    \end{align*}
where $\text{TC}$ is the total cost function, $w_j(x_{j})$ is the rental rate of input $j$, which depends on the quantity demanded by producer of input $j$, $x_{j}$, $\text{FC}$ are fixed costs, $y$ is gross output, $F: \mathbb R_+^N \to \mathbb R_+$ is a continuously differentiable and quasiconcave production technology which satisfies Inada conditions and is parametrized by productivity parameter $A$, and $\overline y$ is the minimum required amount of output. $\mathcal N$ denotes the set of total inputs, which can be either materials or production factors, and $\mathcal F\subseteq \mathcal N$ denotes the set of factors. \par 

The Lagrangian of this programming problem is
    \begin{align*}
        \mathcal L\big(\bm x; \lambda\big) = \sum_j w_j(x_{j})x_{j} + \text{FC} + \lambda\left\{\overline y - F\big(\bm x; A\big)\right\}, 
    \end{align*}
where $\lambda$, the Lagrange multiplier, is the shadow value of relaxing the output constraint. \par 

A generic first-order condition for an interior demand reads as
    \begin{align*}
        w_j(x_{j}) + w_j'(x_{j})x_{j} - \lambda \frac{\partial F}{\partial x_{j}} = 0,
    \end{align*}
which implies
    \begin{align*}
        \left(1 + \frac{w_j'(x_{j}) x_{j}}{w_j(x_{j})}\right) w_j(x_{j})x_{j} = \lambda \frac{\partial F}{\partial x_j}\frac{x_{j}}{y}y,
    \end{align*}
where $\varepsilon_{Sj}^{-1} \equiv \frac{w_j'(x_{j}) x_{j}}{w_j(x_{j})}$ is the producer's (perceived) inverse elasticity of supply of input $j$.\par 

Using standard duality arguments, we can establish that $1 + \varepsilon_{Sj}^{-1}$ is the inverse markdown on input $j$---that is, the marginal revenue product of the producer with respect to $j$ divided by the rental rate of input $j$. We denote the markdown on $j$ by $\nu_{j}$, and note that $\nu_{j}\in(0,1]$. By the Envelope Theorem, it follows that the Lagrange multiplier $\lambda$ is the marginal cost. Thus, we can write:
    \begin{align*}
        w_j(x_{j})x_{j} = \text{mc}\times y\theta_{j}\nu_{j},
    \end{align*}
where $\theta_{j}\equiv \partial F/\partial x_{j}\times x_{j}/y$ is the output elasticity of input $j$.\par 

Summing over all inputs $j\in\mathcal N$, 
    \begin{align*}
        \sum_{j\in\mathcal N} w_j(x_{j})x_{j} &= \text{mc}\times y\left(\sum_{j\in\mathcal N}\theta_{j}\nu_{j}\right) \\
        &= \text{mc}\times y\left(\text{SE} - \sum_{j\in\mathcal F}\theta_{j}\{1 - \nu_{j}\}\right), 
    \end{align*}
where $\text{SE}$ is the scale elasticity of the production function and the second equality follows by imposing the restriction that there is monopsony power only in factor markets, i.e., $\nu_j=1$ for $j\notin\mathcal F$. \par 

Using the definition of the markup, we can write the producer's profit rate $s_{\pi}$, defined as profits $\pi$ over sales $py$ as
    \begin{align*}
        s_{\pi} &= 1 - \frac{\text{SE}}{\mu}\left(\frac{\text{TC}}{\text{TC} - \text{FC}}\right) + \frac{1}{\mu}\left(\frac{\text{TC}}{\text{TC} - \text{FC}}\right)\sum_{j\in\mathcal F}\theta_{j}(1-\nu_{j}).
    \end{align*}
Next, we show that returns to scale, defined as the ratio of average- to marginal cost, are given by $\text{RS} = \text{SE}^{\text{adj}} - \mathcal M$. Recall that 
    \begin{align*}
        \sum_{j\in\mathcal N} w_j(x_{j})x_{j} 
        &= \text{mc}\times y\left(\text{SE} - \sum_{j\in\mathcal F}\theta_{j}\{1 - \nu_{j}\}\right).
    \end{align*}
Doing simple algebraic manipulations, 
    \begin{align*}
        \frac{\text{AC}}{\text{mc}}\frac{\text{TC}- \text{FC}}{\text{AC}\times y} = \text{SE} - \sum_{j\in\mathcal F}\theta_{j}\{1 - \nu_{j}\},
    \end{align*}
where $\text{AC} := \text{TC} / y$ is the average cost. \par 

The last expression implies 
    \begin{align*}
        \text{RS} = \text{SE}\left(\frac{\text{TC}}{\text{TC}-\text{FC}}\right) - \left(\frac{\text{TC}}{\text{TC}-\text{FC}}\right)\sum_{j\in\mathcal F}\theta_{j}\{1 - \nu_{j}\},
    \end{align*}
which establishes the result. 
\end{proof}

\subsection{Proof of Proposition \ref{prop:markup_decomp}} \label{app:markup_decomp}
\noindent \textbf{Proposition \ref{prop:markup_decomp}}. \textit{The aggregate markup---the harmonic sales-weighted markup $\overline\mu_{hsw}$---can be decomposed according to}
    \begin{align}\label{eq:markup_decomp}
        \Delta\overline\mu_{\text{hsw}} &= \underbrace{- \frac{\sum_{i\in\mathcal C_t}\overline\omega_i\Delta\mu_i^{-1}}{\overline\mu_{\text{hsw}, t}\times \overline\mu_{\text{hsw}, t-1}}}_{\text{within}} \underbrace{- \frac{\sum_{i\in\mathcal C_t}\Delta\omega_i\overline{\mu_i^{-1}}}{\overline\mu_{\text{hsw}, t}\times \overline\mu_{\text{hsw}, t-1}}}_{\text{between}} \nonumber \\ 
        &\qquad \qquad \underbrace{- \frac{\sum_{i\in\mathcal E_t}\omega_{it}\left(\mu_{it}^{-1} - \overline\mu_{hsw}^*\right)}{\overline\mu_{\text{hsw}, t}\times \overline\mu_{\text{hsw}, t-1}} + \frac{\sum_{i\in X_{t-1}}\omega_{it-1}\left(\mu_{it-1}^{-1} - \overline\mu_{hsw}^*\right)}{\overline\mu_{\text{hsw}, t}\times \overline\mu_{\text{hsw}, t-1}}}_{\text{net entry}} \nonumber
    \end{align}
\textit{where $i$ and $t$ index producers and time, respectively, $\mathcal C_t := \big\{i\in\mathcal I_t \ \wedge \ i\in\mathcal I_{t-1}\big\}$ is the set of incumbents, $\mathcal E_t := \big\{i\in\mathcal I_t \ \wedge \ i\notin\mathcal I_{t-1}\big\}$ is the set of entrants, $\mathcal X_{t-1} := \big\{i\in\mathcal I_{t-1} \ \wedge \ i\notin\mathcal I_{t}\big\}$ is the set of exiting firms, $\mathcal I_t = \mathcal C_t \cup \mathcal E_t$, $\mathcal I_{t-1} = \mathcal C_t \cup \mathcal X_{t-1}$, $\overline\mu_{hsw,t}$ is the harmonic sales-weighted markup in period $t$, $\Delta X := X_t - X_{t-1}$, $\overline X:= \frac{1}{2}\big(X_t + X_{t-1}\big)$, and $\overline\mu_{hsw}^* = \frac{1}{2}\big(\overline\mu_{\text{hsw}, t} + \overline\mu_{\text{hsw}, t-1}\big)$.}

\begin{proof}
    \begin{align*}
        \Delta\overline\mu_{\text{hsw}} &= \overline\mu_{\text{hsw},t} - \overline\mu_{\text{hsw},t-1} \\
        &= \left(\sum_{i\in\mathcal I_t}\omega_{it}\mu_{it}^{-1}\right)^{-1} - \left(\sum_{i\in\mathcal I_{t-1}}\omega_{i,t-1}\mu_{i,t-1}^{-1}\right)^{-1} \tag{by definition} \\ 
        &= \frac{\sum_{i\in\mathcal I_{t-1}}\omega_{i,t-1}\mu_{i,t-1}^{-1} - \sum_{i\in\mathcal I_t}\omega_{it}\mu_{it}^{-1}}{\overline\mu_{\text{hsw},t} \times \overline\mu_{\text{hsw},t-1}} \tag{taking common denominator}\\ 
        &= - \frac{1}{\overline\mu_{\text{hsw},t} \times \overline\mu_{\text{hsw},t-1}} \left\{\sum_{i\in\mathcal I_t}\omega_{it}\mu_{it}^{-1} - \sum_{i\in\mathcal I_{t-1}}\omega_{i,t-1}\mu_{i,t-1}^{-1}\right\} \tag{rearranging} 
    \end{align*} 
Noting that
    \begin{align*}
        \sum_{i\in\mathcal I_t} \omega_{it}\mu_{it}^{-1} &= \sum_{i\in\mathcal C_t}\omega_{it}\mu_{it}^{-1} + \sum_{i\in\mathcal E_t}\omega_{it}\mu_{it}^{-1}
    \end{align*}
and 
    \begin{align*}
        \sum_{i\in\mathcal I_{t-1}} \omega_{it-1}\mu_{it-1}^{-1} &= \sum_{i\in\mathcal C_t}\omega_{it-1}\mu_{it-1}^{-1} + \sum_{i\in\mathcal X_{t-1}}\omega_{it-1}\mu_{it-1}^{-1},
    \end{align*}
we can write
    \begin{align*}
        \Delta\overline\mu_{\text{hsw}} &= - \frac{1}{\overline\mu_{\text{hsw},t} \times \overline\mu_{\text{hsw},t-1}} \left\{\sum_{i\in\mathcal C_t}\omega_{it}\mu_{it}^{-1} - \sum_{i\in\mathcal C_t}\omega_{it-1}\mu_{it-1}^{-1} + \sum_{i\in\mathcal E_t}\omega_{it}\mu_{it}^{-1} - \sum_{i\in\mathcal X_{t-1}} \omega_{it-1}\mu_{it-1}^{-1} \right\}.
    \end{align*}
The result follows from algebraic manipulations to this expression.  
\end{proof}
\newpage

\subsection{Identification Proof for the Aggregate User Cost of Capital}\label{app:Rident_proof}
\begin{prop}
    Suppose labor compensation, the capital stock, and value added are observable at the aggregate level. Further suppose that representative micro-level data on sales, gross output, variable inputs, and capital stocks are available. Then, under the assumptions of cost-minimizing behavior, homogeneous, continuously differentiable, and quasiconcave production functions, and no monopsony power for one variable input, we can identify the aggregate user cost of capital, $R$.
\end{prop}

\begin{proof}
    By Lemma \ref{lemma:microprofitshare}, the aggregate profit share can be computed as Domar-weighted profit rates. By Proposition \ref{prop:indprofitrate}, economic profit rates are given by equation (\ref{eq:indprofitrate}). Markups, markdowns, elasticities, and returns to scale for each producer, $\{\mu_i, \bm \nu_i, \bm\theta_i, \text{RS}_i\}_i$, can be estimated using the production function approach and the control function approach explained in Section \ref{sec:data}. Importantly, notice that individual user costs of capital are not required in this estimation procedure. Markups are computed in a first step, using the first-order condition for the variable input no subject to monopsony, and markdowns for all the remaining inputs are estimated in a second step using markup estimates and first-order conditions. This is basically the procedure of \cite{yeh2022monopsony}.  \par 

    Using the macro approach, we can obtain the aggregate user cost of capital as
        \begin{align*}
            R = \frac{\text{GDP}}{K}\left(1 - \Lambda_L - \Lambda_\Pi\right), 
        \end{align*}
    where $\text{GDP}$ is aggregate value added, $K$ is the aggregate capital stock, $\Lambda_L$ is labor compensation over value, and $\Lambda_\Pi$ is the Domar-weighted profit share, computed using Lemma \ref{lemma:microprofitshare} and Proposition \ref{prop:indprofitrate}.
\end{proof}

\newpage 
\subsection{Proof of Theorem \ref{thm:profshare_decomp}}\label{app:thm1}

\noindent \textbf{Theorem \ref{thm:profshare_decomp}}. \textit{With cost-minimizer producers, continuously differentiable and quasiconcave production functions, fixed costs, and market power in factor- and output markets, the profit share can be expressed as}
    \begin{align}
        \Lambda_\Pi &= \chi\left(1 - \mathbb{E}_{\omega}\left[\frac{\text{SE}^{\text{adj}}}{\mu}\right] + \mathbb{E}_{\omega}\left[\frac{\mathcal M}{\mu}\right]\right) \tag{\ref{thm:agg_prof_general1}} \\
            &= \chi\left(1 - \frac{\overline{SE}^{\text{adj}}}{\overline\mu_{hsw}} + \frac{\overline{\mathcal M}}{\overline\mu_{hsw}} - \text{Cov}_\omega\left[\text{SE}^{\text{adj}}, \frac{1}{\mu}\right] + \text{Cov}_\omega\left[\mathcal M, \frac{1}{\mu}\right]\right), \tag{\ref{thm:agg_prof_general2}}
    \end{align}
\textit{where $\chi =\sum_{k\in\mathcal I} \frac{p_ky_k}{\text{GDP}}$ is the input-output multiplier, $\mathbb{E}_{\omega}\left[\cdot\right]$ denotes a sales-weighted average where $\omega$ indexes sales weights, 
$\text{SE}^{\text{adj}}$ is the scale elasticity adjusted for fixed costs given by (\ref{eq:RSadj}), $\mu$ denotes the markup, and $\mathcal M$ is a monopsony term given by (\ref{eq:monopsony_term}).}\par 

\textit{The term $\overline X$ is the sales-weighted average of $X$, $\overline\mu_{hsw}$ is the harmonic sales-weighted markup, and $\text{Cov}_\omega(X,Y)$ is the sales-weighted covariance of variables $X$ and $Y$.}

\begin{proof}
By Lemma \ref{lemma:microprofitshare}, the aggregate profit share can be computed as
        \begin{align*}
            \Lambda_\Pi = \sum_{i\in\mathcal I}\frac{p_iy_i}{\text{GDP}} \times s_{\pi_i}.
        \end{align*}
Proposition \ref{prop:indprofitrate} can be used to establish 
    \begin{align*}
        s_{\pi_i} &= 1 - \frac{\text{SE}_i^{\text{adj}}}{\mu_i} + \frac{\mathcal M_i}{\mu_i}.
    \end{align*}
Hence, 
    \begin{align*}
        \Lambda_\Pi &= \sum_{i\in\mathcal I}\frac{p_iy_i}{\text{GDP}}\left(1 - \frac{\text{SE}_i^{\text{adj}}}{\mu_i} + \frac{\mathcal M_i}{\mu_i}\right) \\
        &= \underbrace{\left(\sum_{k\in\mathcal I} \frac{p_ky_k}{\text{GDP}}\right)}_{\equiv \chi}\sum_{i\in\mathcal I}\underbrace{\frac{p_iy_i}{\sum_k p_ky_k}}_{\substack{\equiv\omega_i}}\left(1 - \frac{\text{SE}_i^{\text{adj}}}{\mu_i} + \frac{\mathcal M_i}{\mu_i}\right) \\
        &= \chi\left(1 - \mathbb E_\omega\left[\frac{\text{SE}^{\text{adj}}}{\mu}\right] + \mathbb E_\omega\left[\frac{\mathcal M}{\mu}\right]\right) \\
        &= \chi\left(1 - \frac{\overline{\text{SE}^{\text{adj}}}}{\overline\mu_{hsw}} + \frac{\overline{\mathcal M}}{\overline\mu_{hsw}} - \text{Cov}_\omega\left[\text{SE}^{\text{adj}}, \frac{1}{\mu}\right] + \text{Cov}_\omega\left[\mathcal M, \frac{1}{\mu}\right]\right)
    \end{align*}
\end{proof}

\newpage
\section{Example Linking Increasing Returns and Fixed Costs}\label{app:FC_example}
We provide a simple example to support our claim that increasing returns to scale can be thought of as fixed costs arising from employing factors of production. \par 

Suppose there is a firm with production technology
    \begin{equation}\label{eq:FC_ex_prod}
        y = \begin{cases} A(\ell-\overline\ell)^\alpha, & \ell>\overline\ell \\ 0, & \text{otherwise} \end{cases},
    \end{equation}
where $y$ is output, $A\in\mathbb{R}_{++}$ is a productivity shifter, $\ell\in\mathbb{R}_+$ is labor, $\overline\ell\in\mathbb{R}_{++}$ is a minimum requirement on labor to produce positive output, and $\alpha\in\mathbb{R}_{++}$ is a parameter governing the returns to scale in production.\footnote{An example of a firm with production function (\ref{eq:FC_ex_prod}) could be a pharmaceutical company that needs to hire a minimum amount of labor hours $\bar\ell$ (say, hours provided by attorneys) in order to obtain approval by the US Food and Drug Administration (FDA).} \par 

Assuming the production function is continuously differentiable and that there is no monopsony power, we can write the derivative of output with respect to labor $\ell$ as
    \begin{align}\label{eq:derivative_FC_example}
        \frac{\mathrm d y}{\mathrm d\ell} = \alpha A(\ell-\overline\ell)^{\alpha-1}.
    \end{align}
Multiplying both sides of (\ref{eq:derivative_FC_example}) by $(\ell-\overline\ell)$ we get
    \begin{align*}
        \frac{\mathrm d y}{\mathrm d\ell}(\ell-\overline\ell) = \alpha A(\ell-\overline\ell)^{\alpha} = \alpha y.
    \end{align*}
Rearranging this expression and doing simple algebraic manipulations yields
    \begin{align*}
        \frac{\mathrm d y}{\mathrm d\ell}\frac{\ell}{y} &= \alpha + \frac{\mathrm d y}{\mathrm d\ell}\frac{\overline\ell}{y} \noindent \\
                                                        &= \alpha + \frac{\mathrm d y}{\mathrm d\ell}\frac{\ell}{y}\times\frac{\overline\ell}{\ell},
    \end{align*}
which implies
    \begin{align*}
        \frac{\mathrm d y}{\mathrm d\ell}\frac{\ell}{y} &= \alpha\frac{\ell}{\ell-\overline\ell} > \alpha.
    \end{align*}
Thus, if $\alpha = 1$ and $\ell > \overline \ell > 0$, there are increasing returns to scale because of fixed costs. \par 

Next, we argue that the returns to scale, $\text{RS}$, that we would estimate for this production function are
    \begin{equation*}
        \text{RS} = \alpha \frac{\ell}{\ell-\overline\ell}.
    \end{equation*}
To see this, suppose we have cross-sectional data on $\{y, (\ell+\overline\ell)\}_i$, where $y$ is output, $(\ell+\overline\ell)$ is total labor, and $i$ indexes firms. Then, consider the following regression
    \begin{align*}
        \log y_i = \beta\log(\ell+\overline\ell)_i + \tilde \omega_{i},
    \end{align*}
where $\beta$, the returns to scale, is the parameter of interest, and $\tilde\omega$ is white noise. \par 

Note that 
    \begin{align*}
        \beta &= \frac{\mathrm d\log y}{\mathrm d\log(\ell + \overline\ell)} \\
              &= \frac{\mathrm d\log y}{\mathrm d\log \ell} \\
              &= \frac{\mathrm d y}{\mathrm d \ell} \frac{\ell}{y}.
    \end{align*}
Given production function (\ref{eq:FC_ex_prod}), it follows that
    \begin{align*}
        \beta &= \alpha \frac{\ell}{\ell-\overline\ell},
    \end{align*}
which is the parameter we would recover from a regression. \par 

From this simple example, we argue that increasing returns to scale may capture increasing fixed costs. We use this logic to interpret the returns to scale we estimate in our empirical analysis. \par


\newpage 
\section{The DEU--Basu Controversy}\label{app:basu_deu}
\cite{basu2019price} was the first to suggest that \cite{de2020rise}'s markup estimates had implausible macroeconomic implications. Through an informal argument, he noted that a markup of 1.61 in 2016 and returns to scale of 1.05 imply a profit share of about 70\% of value added in the United States once we recognize that the ratio of sales to value added is approximately two. He argued that such a profit share was implausible since the labor share, which can be obtained from NIPA tables, was around 62\% of GDP. \cite{barkai2020declining} made \citeauthor{basu2019price}'s argument more formally by explicitly using expression 
    \begin{align}\label{eq:basu_deu}
        \Lambda_\Pi = \chi\left(1 - \frac{\text{RS}}{\mu}\right),
    \end{align}
where $\Lambda_\Pi$ is the profit share, $\chi$ is the ratio of sales to aggregate value added, $\text{RS}$ are returns to scale, and $\mu$ is the markup. \par 

\citeauthor*{de2020rise}'s response to \citeauthor{basu2019price} (and indirectly to \citeauthor{barkai2020declining}) was to say that such an argument could not be made because it incorrectly relied on a representative-firm assumption and imposed no fixed costs. We contribute to this discussion by providing an exact aggregation result that clarifies where the arguments of \citeauthor{basu2019price} and \citeauthor{barkai2020declining} fail and why the aggregate sales-weighted markup of \citeauthor*{de2020rise} cannot be used to infer the aggregate profit share.\par 

Corollary \ref{corr:nomonop} provides the following expression for the profit share under the assumption of no monopsony:
    \begin{align*}
        \Lambda_\Pi = \chi\left(1 - \frac{\overline{\text{RS}}}{\overline\mu_{hsw}} - \text{Cov}_\omega\left[\text{RS}, \frac{1}{\mu}\right]\right), \tag{\ref{eq:profshare_nomonop}}
    \end{align*}
where $\overline{\text{RS}} = \overline{\text{SE}}^{\text{adj}}$ are sales-weighted returns to scale, which are equal to the sales-weighted scale elasticity adjusted for fixed costs, $\overline\mu_{hsw}$ is the harmonic sales-weighted markup, and $\text{Cov}_\omega(\cdot,\cdot)$ is the sales-weighted covariance operator. \par 

Inspection of equation (\ref{eq:profshare_nomonop}) reveals several problems with \citeauthor{basu2019price} and \citeauthor{barkai2020declining}'s arguments. First, as noted by \citeauthor*{de2020rise}, if there is heterogeneity in markups and returns to scale across producers, aggregation is nonlinear as by Jensen's inequality. That explains why the covariance term appears in equation (\ref{eq:profshare_nomonop}). Second, our aggregation result demands the harmonic sales-weighted markup and sales-weighted returns to scale, not the sales-weighted markup. If there is dispersion in markups, the sales-weighted markup is larger than the harmonic sales-weighted markup, and the discrepancy between the two markup notions depends on the variance of the cross-sectional distribution of markups. Thus, using the sales-weighted markup in equation (\ref{eq:profshare_nomonop}) would generally lead to overstating the aggregate profit share. 
While \citeauthor{basu2019price} and \citeauthor{barkai2020declining} used what was available to them---that is, the estimates provided by \citeauthor*{de2020rise}---the sales-weighted markup is not an appropriate measure of aggregate markup to draw conclusions about the aggregate profit share; instead, the harmonic sales-weighted markup is. A third problem with \citeauthor{basu2019price} and \citeauthor{barkai2020declining}'s arguments in inferring the profit share from \citeauthor{de2020rise}'s estimates is to not account for fixed costs. What they use as returns to scale are not returns to scale, but rather the scale elasticity of the production function. In the absence of monopsony power, returns to scale are given by the scale elasticity times a fixed-cost adjustment factor as our formula shows. The fixed-cost adjustment factor is important because \citeauthor{de2020rise} treat SG\&A as an explicit fixed cost. Once we aggregate correctly and account for all fixed costs, we get equation (\ref{eq:profshare_nomonop}), which is well suited to assess the validity of micro estimates of markups and returns to scale under the assumption of no monopsony.\footnote{Equation (\ref{eq:profshare_nomonop}) nests \citeauthor{basu2019price}'s and \citeauthor{barkai2020declining}'s expression (\ref{eq:basu_deu}) as the special case in which there are no explicit fixed costs and no heterogeneity in markups.} \par 

In Figure \ref{fig:profshare_Basu_comparison}, we map micro-level estimates of markups and returns to scale to the aggregate profit share using equation (\ref{eq:profshare_nomonop}). To relate to the discussion between \citeauthor{basu2019price} and \citeauthor*{de2020rise}, we use the production function parameters that replicate \cite{de2020rise}'s results. That is, we use the cost of goods sold (COGS) as a variable input, treat selling, general, and administrative expenses (SG\&A) as fixed costs, and only consider physical capital when defining the capital stock. The Hasenzagl--Perez series maps micro-level estimates of markups and returns to scale to an aggregate profit share that is constant at around 15\% of GDP. The other three series map firm-level estimates of markups to the aggregate profit share using flawed aggregation schemes and ignoring relevant features of the data such as fixed costs.\par 

\begin{figure}[h!]
    \centering
    \caption{Mapping Micro-Level Estimates of Markups to the US Profit Share.}
    \includegraphics[scale=0.35]{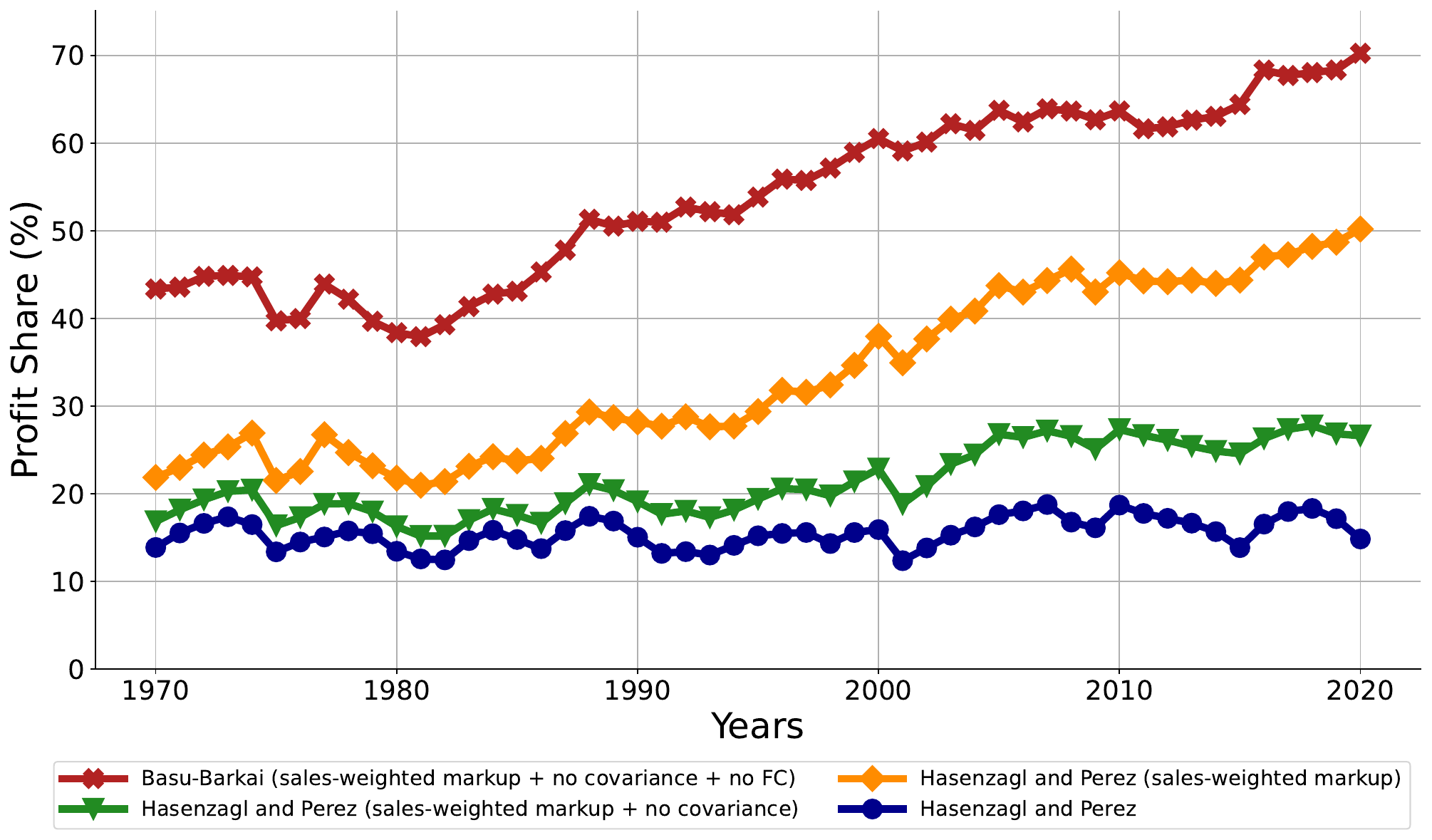}
    \label{fig:profshare_Basu_comparison}
    \vspace{7pt}
\end{figure}

\newpage 
\section{Data Appendix}\label{app:data}
Table \ref{tab:sumstats_compustat} provides key summary statistics for the variables we use in our estimation procedure. Although the relevant years of analysis are 1970--2020, we provide summary statistics for 1966--2020 given that information from 1966 onward is used to estimate input elasticities, as noted in Section \ref{sec:PFE}. \par 
\begin{table}[h!]
    \centering
    \caption{Summary Statistics from US Compustat, 1966--2020.}
    \begin{tabular}{l | c c c c c}
    \toprule 
    \textbf{Variable} & \textbf{Mean} & \textbf{Median} & \textbf{Min} & \textbf{Max} & \textbf{Std. dev.} \\
    \midrule 
    SALE     & 2261.51 & 157.67 & 0 & 489529.70 & 11930.17 \\
    COGS     & 1544.21 & 92.69  & 0 & 428240.60 & 9118.36  \\
    SG\&A    & 332.07  & 28.22  & 0 & 98210.41  & 1619.11  \\
    R\&D     & 51.64   & 0      & 0 & 37567.35  & 436.57   \\
    PPEGT    & 1658.22 & 54.94  & 0 & 542207.80 & 10979.76 \\
    K\_INT   & 1013.56 & 45.12  & 0 & 334349.90 & 6091.17  \\
    CAPX     & 161.80  & 6.58   & 0 & 50038.86  & 1042.36   \\
    \cdashline{1-6}
    \multicolumn{1}{l}{Observations}  & \multicolumn{5}{l}{261,945} \\
    \bottomrule 
    \end{tabular}
    \label{tab:sumstats_compustat}
\end{table}

\newpage 
\section{Biases Associated with the Use of Revenue Elasticities}\label{app:revenue_bias}
In this section, we discuss how the use of revenue- rather than output data affects elasticity estimates and measures of market power and profitability. More specifically, we show that although using revenue-based elasticities (instead of output ones) may introduce downward biases in markups, our estimates of profit rates and profit shares are unaffected by this source of bias.\par 

\subsection{Markups}
As first noted by \cite{klette1996inconsistency}, the use of deflated revenue data as a proxy for real output introduces an omitted-variable bias in the estimated scale elasticities when there is price dispersion in prices within the industry for which scale elasticities are estimated. This bias is caused by the correlation of prices with inputs and generally results in a downward bias in the scale elasticities of the production function and, consequently, in markups. More recently, \cite{bond2021some} and \cite{de2022hitchhiker} have provided a more general discussion on the direction of this bias and raised additional concerns about using revenue data to estimate output elasticities. \par 

The major criticism put forth by \citeauthor{bond2021some} is that when the revenue elasticity for a flexible input is used in place of the output elasticity, the estimated markup of a firm that maximizes current-period profits is equal to one and thus not informative of the true markup. The key to this argument is that firms maximize static profits. In more general environments in which firms maximize the discounted sum of profits, static profit maximization need not apply, although firms may still minimize costs statically.\footnote{Several papers in the repeated-games literature have shown that with with low-enough discount factors and strategic considerations, it is possible to sustain equilibrium outcomes that look very different from those that arise under static profit maximization \citep[e.g., see][]{green1984noncooperative, abreu1986extremal}. } Imposing only cost-minimizing behavior, our next proposition shows, similarly to \citeauthor{klette1996inconsistency} and \citeauthor{bond2021some}, that revenue-based markups are downward biased when firms face downward-sloping demand curves. \par 

\begin{prop}\label{prop:biased_markup}
    Revenue-based markups $\mu^R$ understate true markups $\mu$ in situations of monopolistic competition; that is, $\mu^R\leq\mu$.
\end{prop}

\begin{proof}
    Cost minimization implies that the true markup can be written as $\mu = \theta_\ell \alpha_\ell^{-1}$, where $\theta_\ell$ is the output elasticity of variable input $\ell$, and $\alpha_\ell^{-1}$ is its inverse revenue share. The revenue-based markup is given by
        \begin{align*}
            \mu^R = \frac{\theta^R_\ell}{\alpha_\ell},
        \end{align*}
    where $\theta^R_\ell:= \frac{\partial R}{\partial \ell}\frac{\ell}{R}$ is the revenue-based elasticity of input $\ell$, and $R=py$ is revenue. \par 

    Using the chain rule it follows that $\mu^R = \theta^R_y \mu$, where $\theta^R_y:=\frac{\mathrm d R}{\mathrm d y}\frac{y}{R}$ is the revenue elasticity of output. \par 
    
    Under perfect competition, firms have no ability to influence prices, so $\theta^R_y = 1$ and $\mu^R = \mu$. When firms face downward-sloping demand curves and have the ability to influence prices, it is easy to show that
        \begin{align*}
            \theta^R_y = \left(\frac{\epsilon-1}{\epsilon}\right) < 1,
        \end{align*}
    since the firm's inverse (perceived) elasticity of demand is $\epsilon < \infty$. Hence, $\mu^R\leq\mu$.    
\end{proof}

Proposition \ref{prop:biased_markup} establishes that revenue-based markups understate true markups when firms face downward-sloping demand curves. The magnitude of this bias depends on the elasticity of the demand curve. This bias vanishes in the limiting case of a perfectly-elastic demand curve ($\epsilon\to\infty$). \par 

\subsection{Profit Rates and Profit Share}
Our next proposition shows that even when revenue-based markups understate true markups, revenue-based profit rates, which rely on revenue-based markups, are unaffected by this source of bias. \par 

\begin{prop}\label{prop:unbiased_profitrates}   
    Revenue-based and output-based profit rates are equal; that is, $s_\pi^R = s_\pi$. 
\end{prop}

\begin{proof}
    By proposition \ref{prop:indprofitrate}, the (output-based) profit rate of a monopolistic producer with no market power in factor markets is given by
        \begin{align*}
            s_\pi = 1 - \frac{\text{RS}}{\mu},
        \end{align*}
    where $\text{RS} = \text{SE}\left(\frac{\text{TC}}{\text{TC}-\text{FC}}\right)$ are the returns to scale which are given by the scale elasticity $\text{SE} = \sum_j\theta_j$ times a fixed-cost adjustment factor, and $\mu$ is the output-based markup. \par 

    The revenue-based elasticity of each input can be shown to be $\theta_j^R = \theta^R_y \times\theta_j$, where $\theta^R_y$ is the revenue elasticity of output and $\theta_j$ is the output elasticity of input $j$. Hence, both the returns to scale and the revenue-based markup are biased by the same factor. When computing the profit rate, these two biases cancel. 
\end{proof}

Propositions \ref{prop:biased_markup} and \ref{prop:unbiased_profitrates} together imply that although our estimates of markups may be downward biased, our profit rates estimates, and consequently the profit share, are unaffected by the use of revenue elasticities. \par  

\subsection{Other Sources of Biases}
The literature has identified additional sources of bias in markup estimates against which our results are not immune. \par 

\cite{de2022hitchhiker} show that using a Cobb-Douglas production function to estimate revenue-based elasticities, as we do, generally leads to overestimating the variance of the markup distribution. \par 

\citeauthor{bond2021some} make two additional criticisms that apply even to the estimation of output elasticities. The first is that if the variable input used to estimate the markup is distorted, the estimated markup will reflect not only market power in output markets but also other frictions (e.g., adjustment costs.) A friction of particular interest is monopsony power. If a producer has monopsony power in the variable input used to construct the markup, then its markup captures both market power in output- and input markets, as noted by \cite{morlacco2019market}, \cite{brooks2021exploitation}, and \cite{yeh2022monopsony}. The second criticism set forth by \citeauthor{bond2021some} is that if markups are estimated using a variable input that firms use for both the production of output and influencing demand (e.g., marketing expenses,) then they will be downward biased. \par 

\cite{doraszelski2020inconsistency} argue that the first-stage of the \cite{ackerberg2015identification}'s estimation procedure employed requires researchers to observe markups while markup estimation is in itself the goal of the estimation procedure. \cite{de2022hitchhiker} argue that not controlling for markups in the first-stage of this estimation procedure does not seem to matter much in practice. \par

\newpage 
\section{Additional Results}\label{app:results}

\subsection{Replication of \cite*{de2020rise}}\label{app:deu_rep}
In this section, we show that computing and aggregating markups as \cite{de2020rise} do, we almost perfectly replicate their results. The replication of markups is shown in Figure \ref{fig:deu_rep}.\footnote{We find very similar results when we estimate time-varying input elasticities which are stable and oscillate around 0.84--0.86.} Small quantitative discrepancies result from updates in US Compustat and minor differences in sample selection, data cleaning, and input elasticities.  \par 
\begin{figure}[h!]
    \centering
    \caption{Replication of \citeauthor{de2020rise}'s Markup Series, 1970--2020.}
    \includegraphics[scale=0.37]{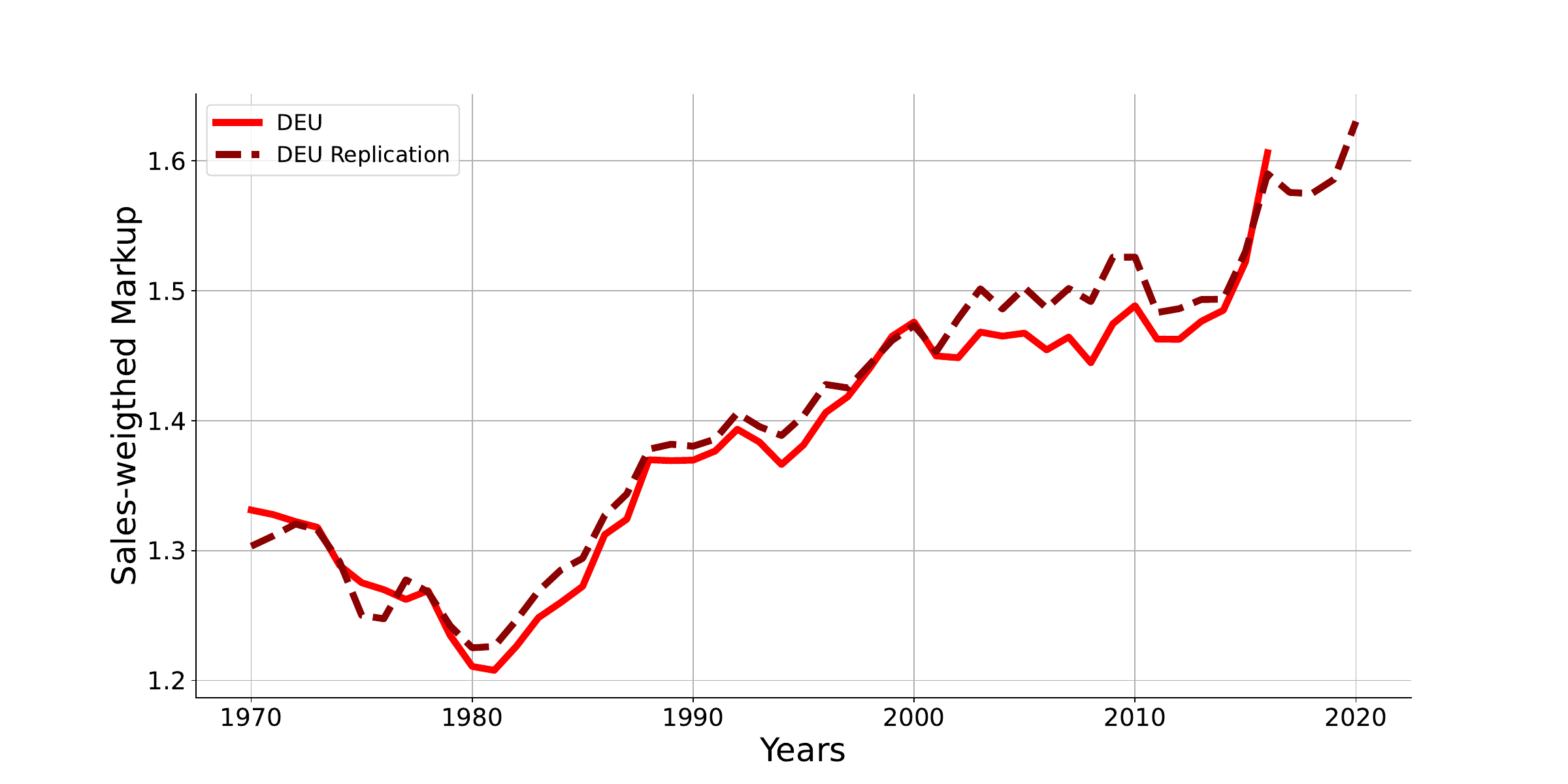}
    \label{fig:deu_rep}
    \begin{minipage}{1\textwidth}
        \scriptsize{\textbf{Table Notes}. Firm-level data comes from US Compustat. Markups are computed using the costs of goods sold (COGS) as variable costs, time-varying elasticities, and aggregating markups using sales weights. }  
    \end{minipage}
\end{figure}

Figure \ref{fig:DEU_mu_differences} shows four different aggregate markup series: the one from \citeauthor*{de2020rise} (DEU replication) and three other aggregate markup series which rely on the aggregation scheme implied by our theorem. The difference between cross-marked and diamond-marked series illustrates that more than half of the differences between \citeauthor{de2020rise}'s markup series and ours arises due to different aggregation schemes. We use the harmonic sales-weighted markup rather than the sales-weighted markup. The difference between cross-marked and triangle-marked series indicates in light of the above that the remaining difference between \citeauthor{de2020rise}'s markup series and ours is due to a different categorization of variable inputs. We treat the cost of goods sold (COGS) and selling, general, and administrative expenses (SG\&A) as variable inputs, while \cite{de2020rise} categorize only COGS as variable inputs. In theory, the markup could be recovered from any variable input so it should not matter whether one uses COGS or OPEX ($=$ COGS + SG\&A) to estimate markups. In practice, however, the categorization of variable inputs matters for markup estimation. We follow \cite{traina2018aggregate} and most of the subsequent literature and treat both COGS and SG\&A as variable inputs. Finally, we note that the inclusion of intangible capital does not affect markup estimates, and that the discrepancies between triangle-marked and dot-marked lines reflect measurement error. This is to be expected since markup estimates do not depend on the elasticity of output with respect to capital. \par 
\begin{figure}
    \centering
    \includegraphics[scale=0.37]{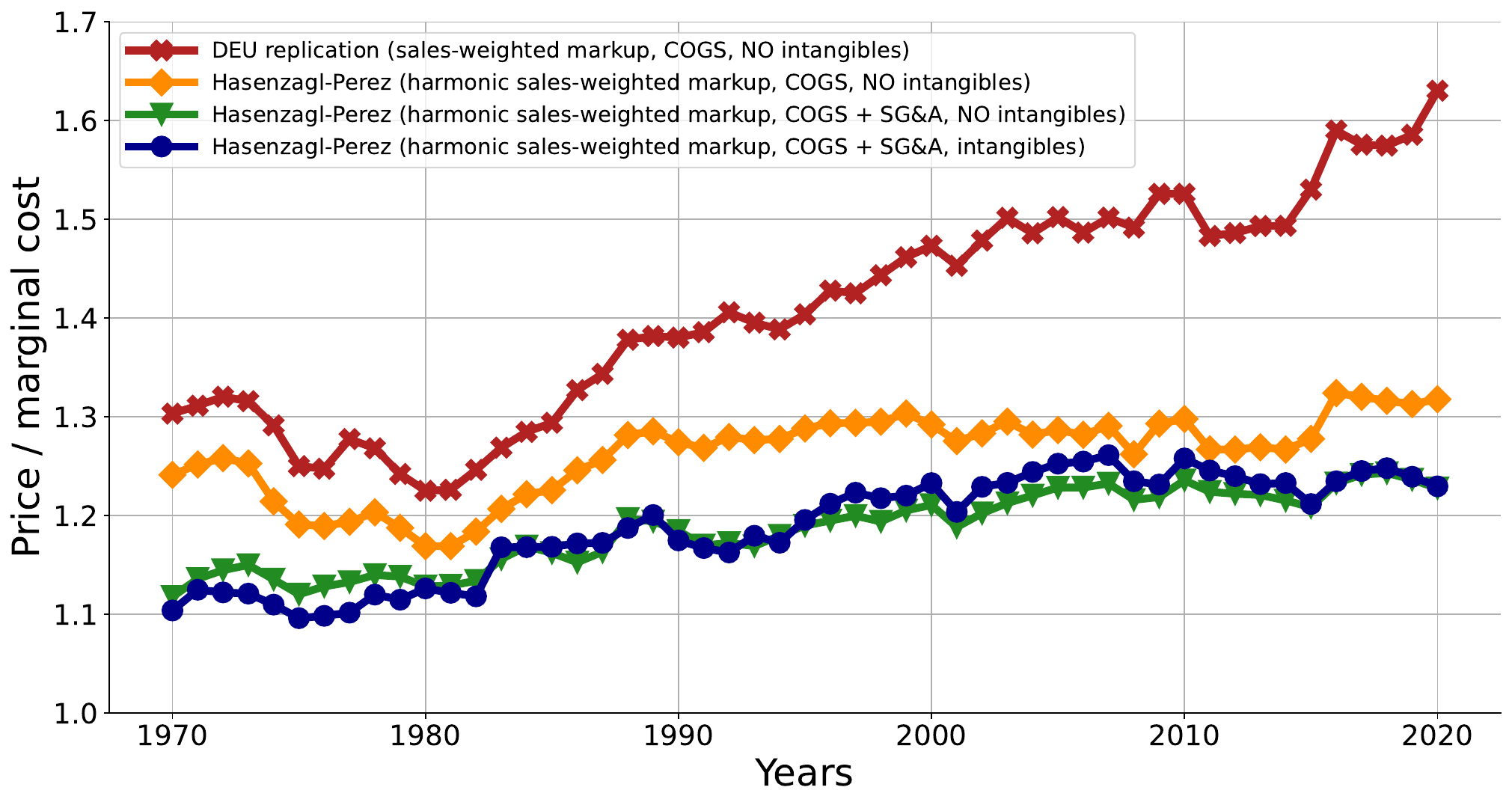}
    \caption{Differences Between \cite{de2020rise}'s Markup Series and Ours.}
    \label{fig:DEU_mu_differences}
\end{figure}

Figure \ref{fig:deu_rep_profit} shows the replication of \cite{de2020rise}'s sales-weighted average profit rate. Again, we almost perfectly replicate their series. This is not surprising in the light of the almost perfect replication of markups and input elasticity estimates, since we we also use their exogenous user cost of capital.\par 
\begin{figure}[h!]
    \centering
    \caption{Replication of \cite{de2020rise} Sales-weighted Profit Rate, 1970--2020.}
    \vspace{-10pt}
    \includegraphics[scale=0.37]{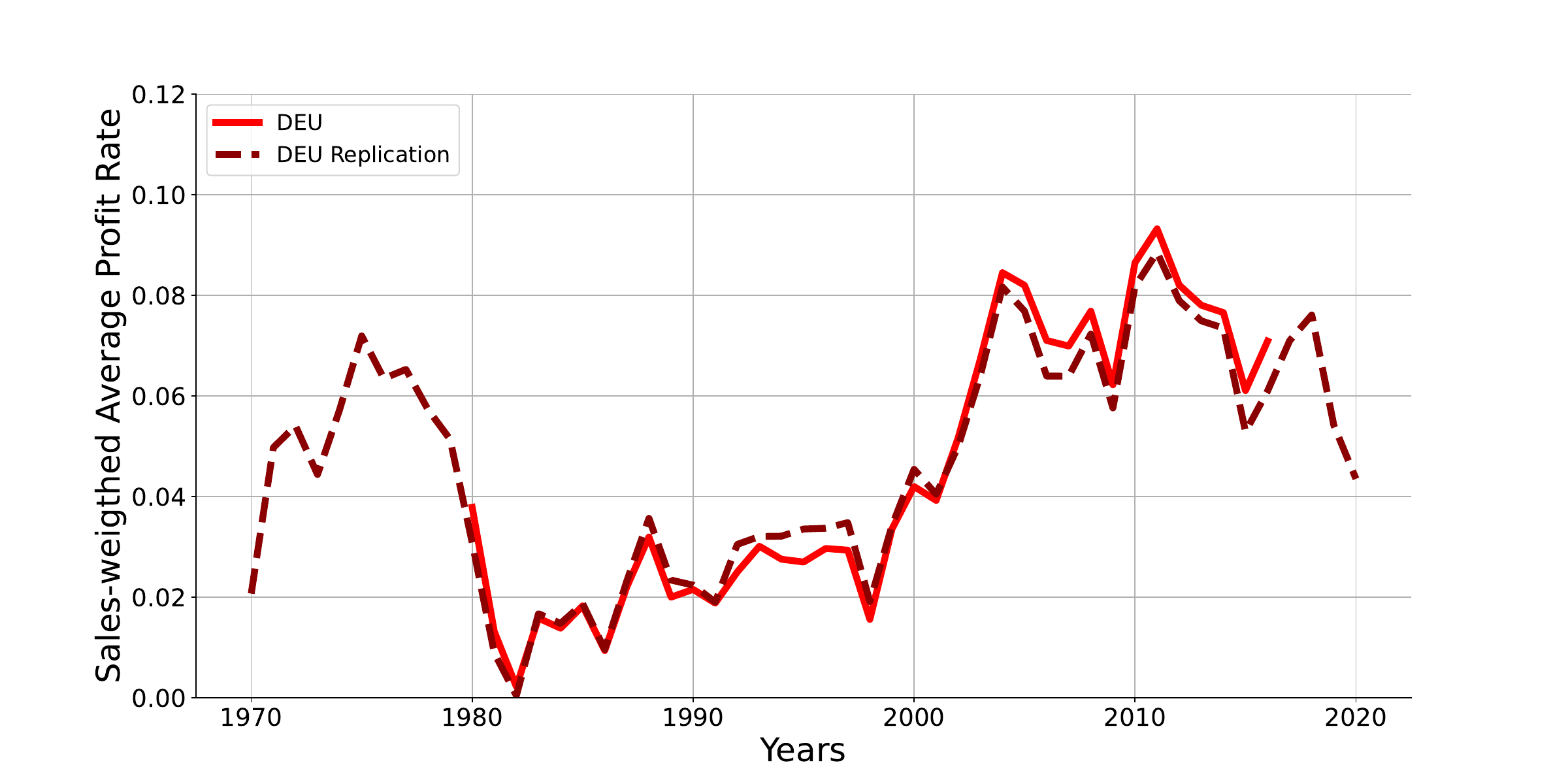}
    \label{fig:deu_rep_profit}
    \begin{minipage}{1\textwidth}
        \scriptsize{\textbf{Table Notes}. Firm-level data comes from US Compustat. Markups are computed using the costs of goods sold (COGS) as variable costs, time-varying elasticities, and aggregating markups using sales weights. Profit rates for each producer are computes as one minus the time-varying input elasticity of the variable input over the markup, minus fixed costs divded by sales.}  
    \end{minipage}
\end{figure}

\clearpage 
\newpage 
\subsection{The Input--Output Multiplier in the United States}\label{app:IO_mult}
Figure \ref{fig:IO_mult} plots the time series for the input-output multiplier, which is the sum of Domar weights or, in other words, the ratio of economy-wide sales to GDP. Clearly from this figure, the IO multiplier for the United is roughly stable at around 1.8 and mildly procyclical. This implies that, absent changes on markups, markdowns and returns to scale during recessions, the aggregate profit share should decline during recessions. \par 
\begin{figure}[h!]
    \centering
    \caption{The US Input--Output Multiplier, 1970--2020.}
    \vspace{-10pt}
    \includegraphics[scale=0.42]{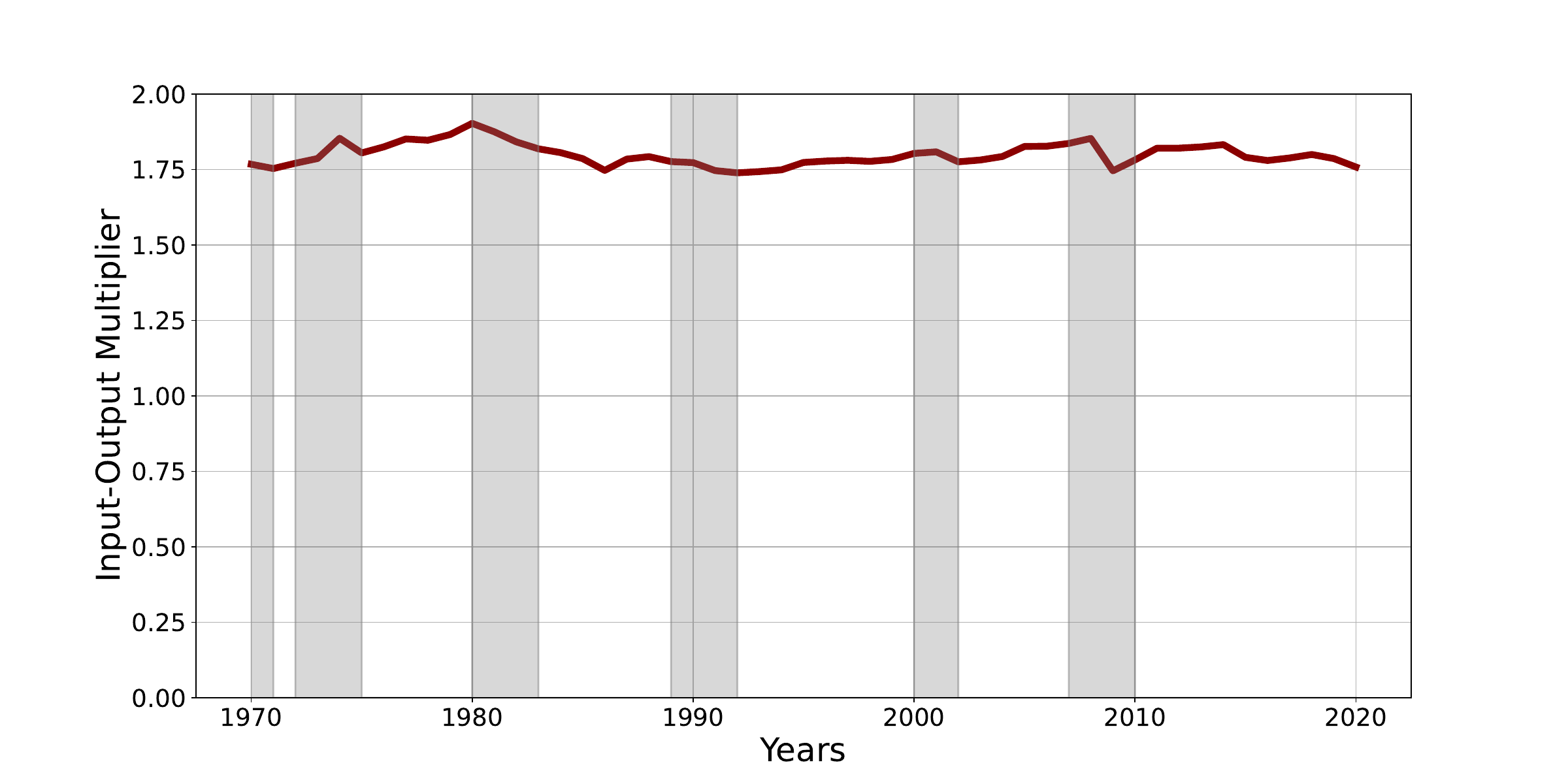}
    \label{fig:IO_mult}
\end{figure}

\newpage 
\subsection{Heterogeneity in Variable-Input Elasticities}\label{app:elast_het}
\begin{figure}[h!]
    \centering
    \caption{Variable Input Elasticities Across Industries, 1970--2020.}
    \vspace{-10pt}
    \includegraphics[scale=0.4]{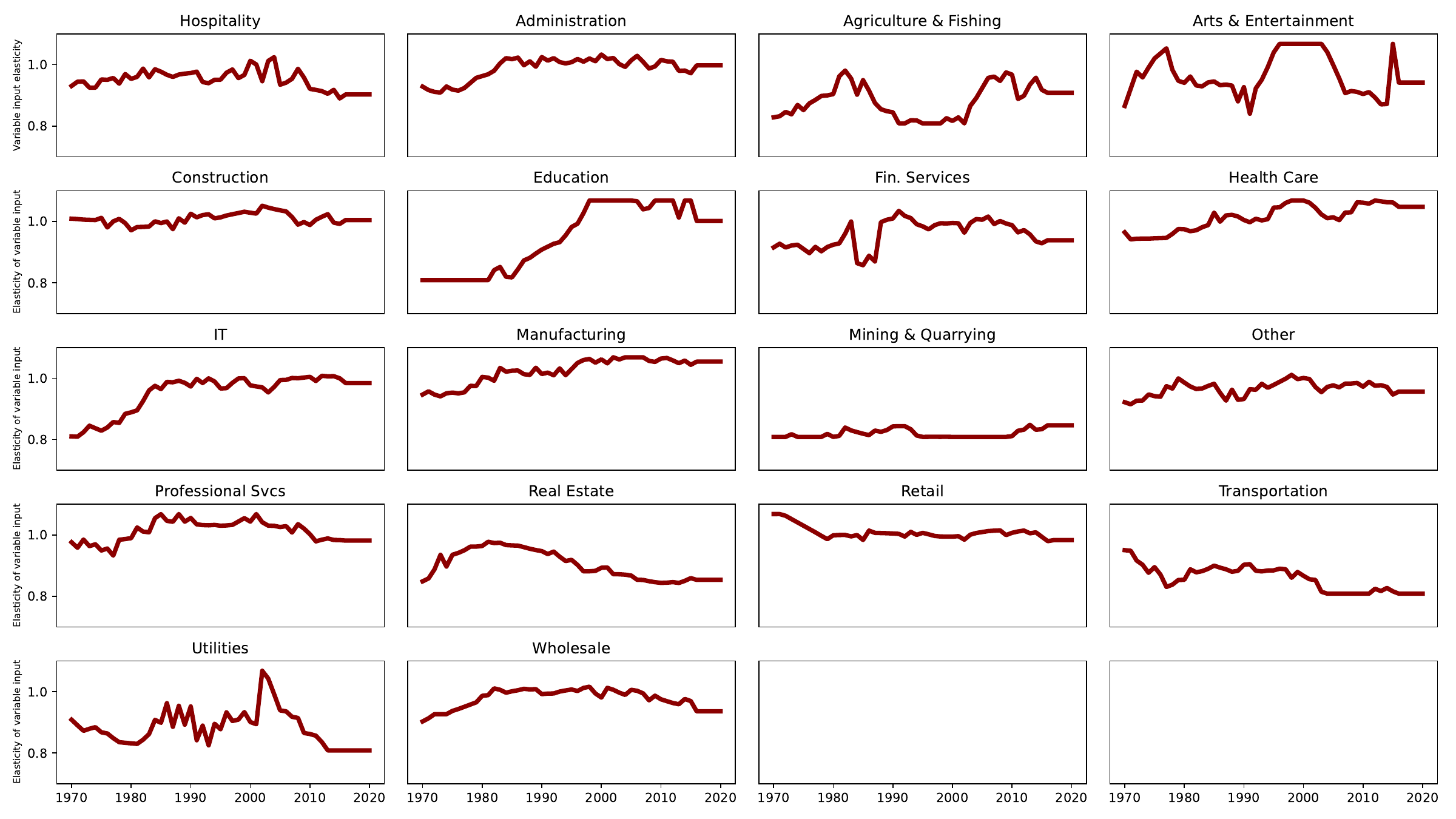}
    \label{fig:Elast_int}
\end{figure}

\subsection{Heterogeneity in Returns to Scale}\label{app:RS_het}
\begin{figure}[h!]
    \centering
    \caption{Returns to Scale Across Industries, 1970--2019.}
    \vspace{-10pt}
    \includegraphics[scale=0.4]{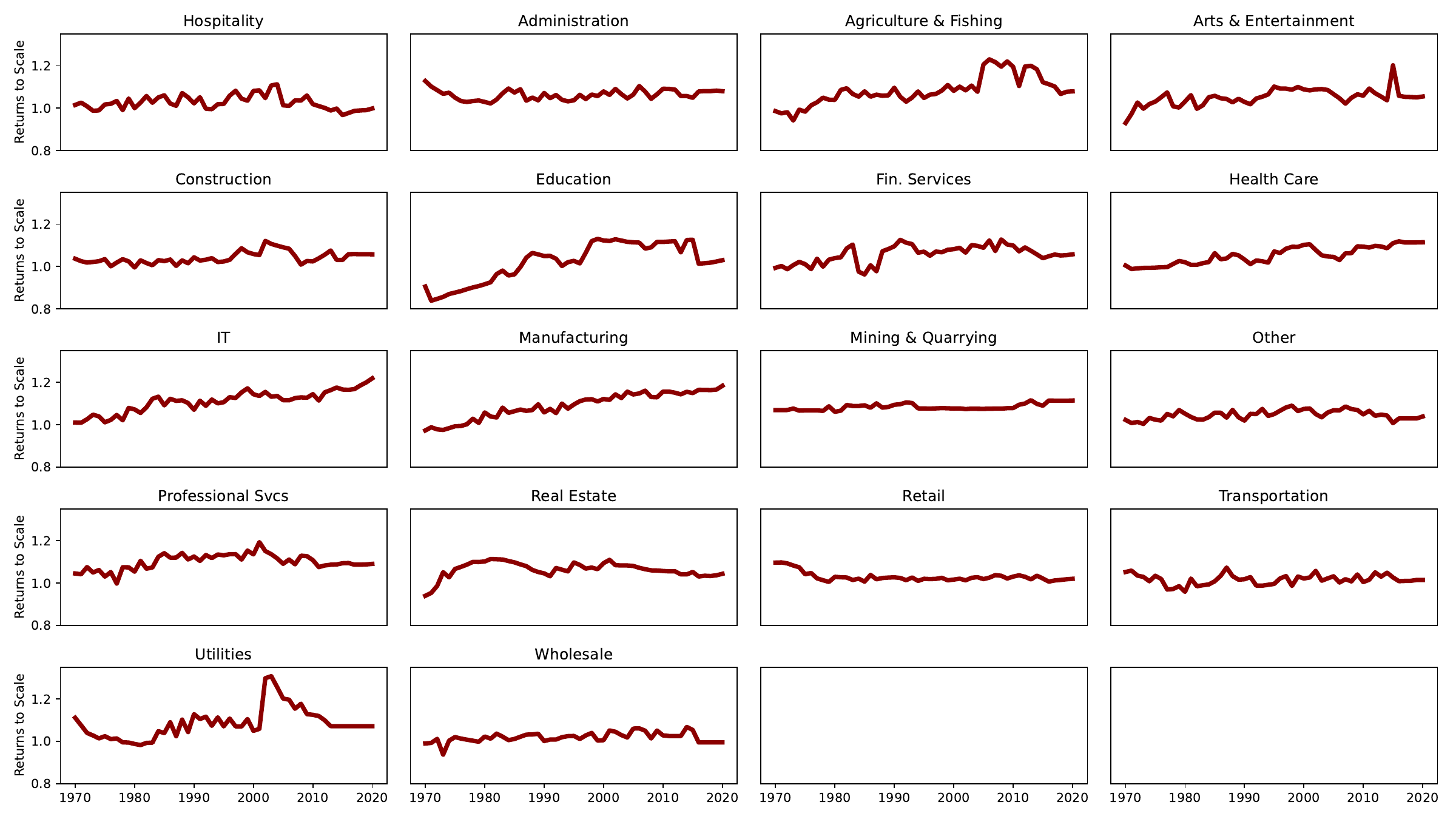}
    \label{fig:RS_ind}
\end{figure}

\newpage 
\subsection{Heterogeneity in Markups}\label{app:markup_het}
Figure \ref{fig:markup_het} provides time series of industry-level markups. As the different panels show, industry-level markups have remained stable for a handful of industries---namely, construction, health care, mining and quarrying, professional services, transportation, arts and entertainment, and wholesale- and retail trade. Other industries such as hospitality, administration, agriculture and fishing, education, financial services, IT, manufacturing, real estate and utilities have experienced large increases in markups over the past fifty years.\par 
\begin{figure}[h!]
    \centering
    \caption{Industry-level Markups, 1970--2020.}
    \includegraphics[scale=0.4]{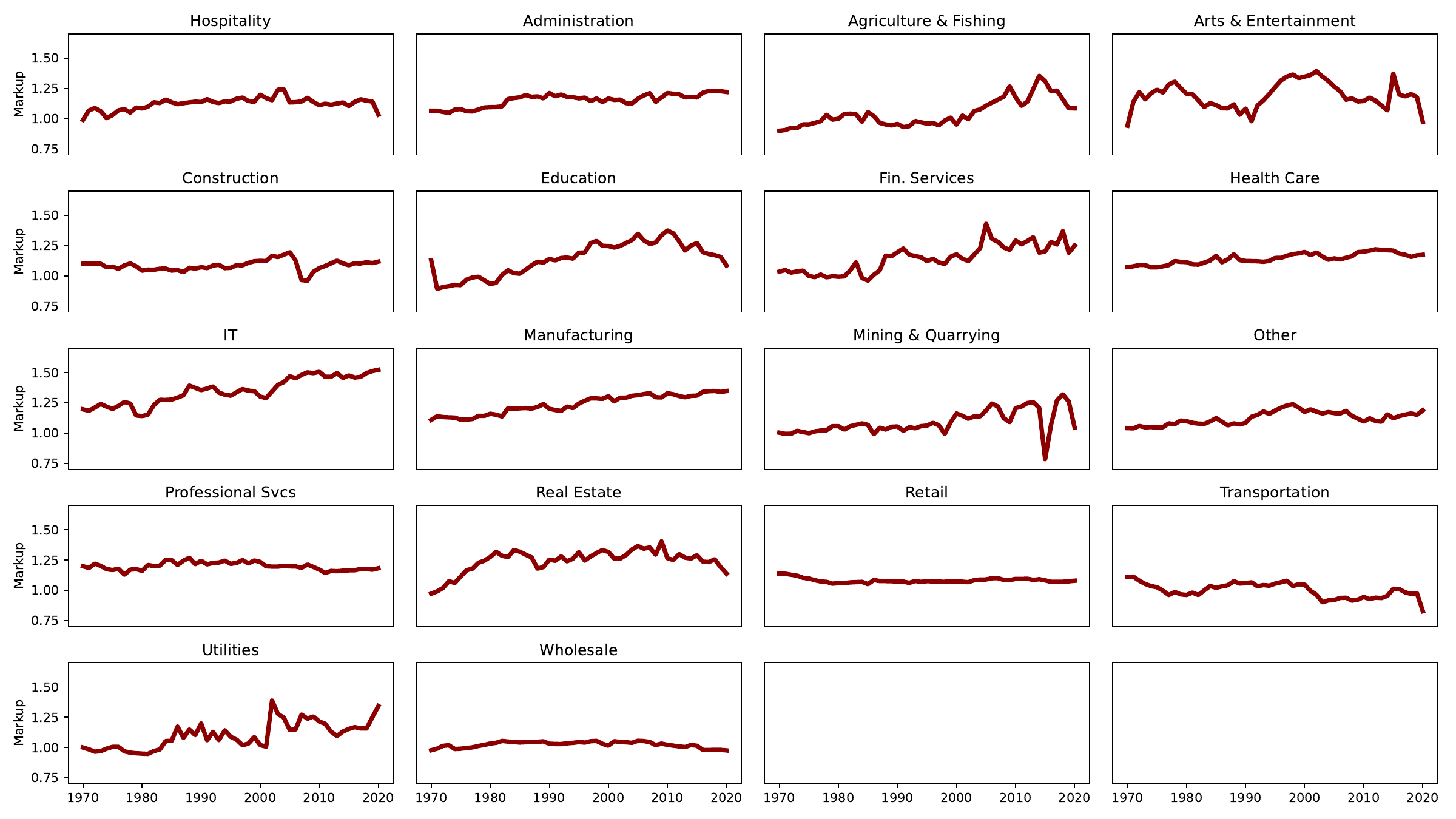}
    \label{fig:markup_het}
\end{figure}

\newpage 
\clearpage
\subsection{Concentration Measures }\label{app:concentration}

In this section, we report Herfindahl-Hirschman Index (HHI) measures of concentration, both at the national and the industry level, for the United States. We compute the HHI as the sum of squared sales shares. That is, 
    \begin{align*}
        \text{HHI}_t &= \sum_{i\in\mathcal I} \omega_{it}^2, \tag{National HHI} \\ 
        \text{HHI}_{jt} &= \sum_{i\in\mathcal I_j} \omega_{it}^2, \tag{Industry-$j$ HHI}
    \end{align*}
where $\omega_{it}$ is the sales share of producer $i$ at time $t$ (i.e., the sales of producer $i$ over total sales), $\mathcal I$ is the set of all firms, and $\mathcal I_j$ is the set of all firms in industry $j$. Thus, total sales in the national HHI is the sum of sales for all firms in the country, whereas total sales in industry $j$'s HHI is the sum of sales of all firms in that particular industry. \par  

The national HHI for the United States, reported in Figure \ref{fig:national_HHI}, indicates that market concentration has decreased from 1970 to 2020. More specifically, concentration has declined from 0.007 in 1970 to 0.004 in 2000, and then increased until 0.005 in 2020. To make sense of these numbers, we provide the interpretation of the HHI suggested by \cite{smith2022evolution}. Perhaps surprisingly, the probability that two dollars, chosen at random, were spent on the same firm was higher in 1970 than in 2020. \par 

\begin{figure}[h!]
    \centering
    \caption{National HHI in the United States, 1970--2020.}
    \vspace{-10pt}
    \includegraphics[scale=0.42]{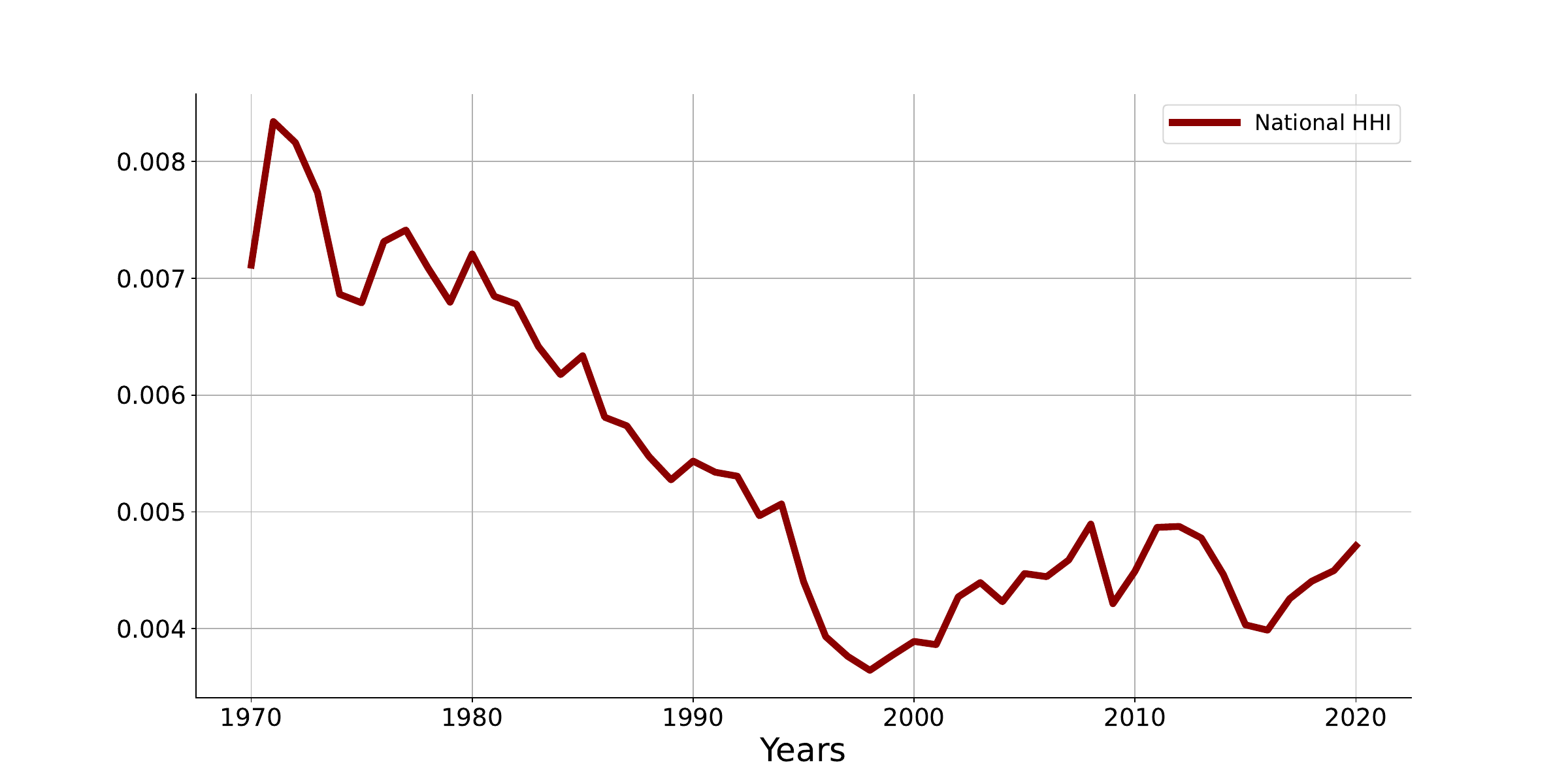}
    \label{fig:national_HHI}
    \begin{minipage}{1\textwidth}
        \scriptsize{\textbf{Figure Notes}. The Herfindahl-Hirschman Index (HHI) is the sum of squared sales weights. Data comes from US Compustat.}  
    \end{minipage}
\end{figure}
\newpage 

Figure \ref{fig:ind_HHI} shows the HHI for different US industries. While market concentration has increased in some industries, it has declined in the majority of industries: mining, manufacturing, transportation, IT, real estate, professional services, administration, education, arts \& entertainment, and hospitality. Consistently with \cite{smith2022evolution}, we find that market concentration, as measured by sales, has sharply increased since the 1980s in the retail sector, which appears to be the exception to the norm. \par 
\begin{figure}[h!]
    \centering
    \caption{Industry-level HHI in the United States, 1970--2020}
    \vspace{-10pt}
    \includegraphics[scale=0.40]{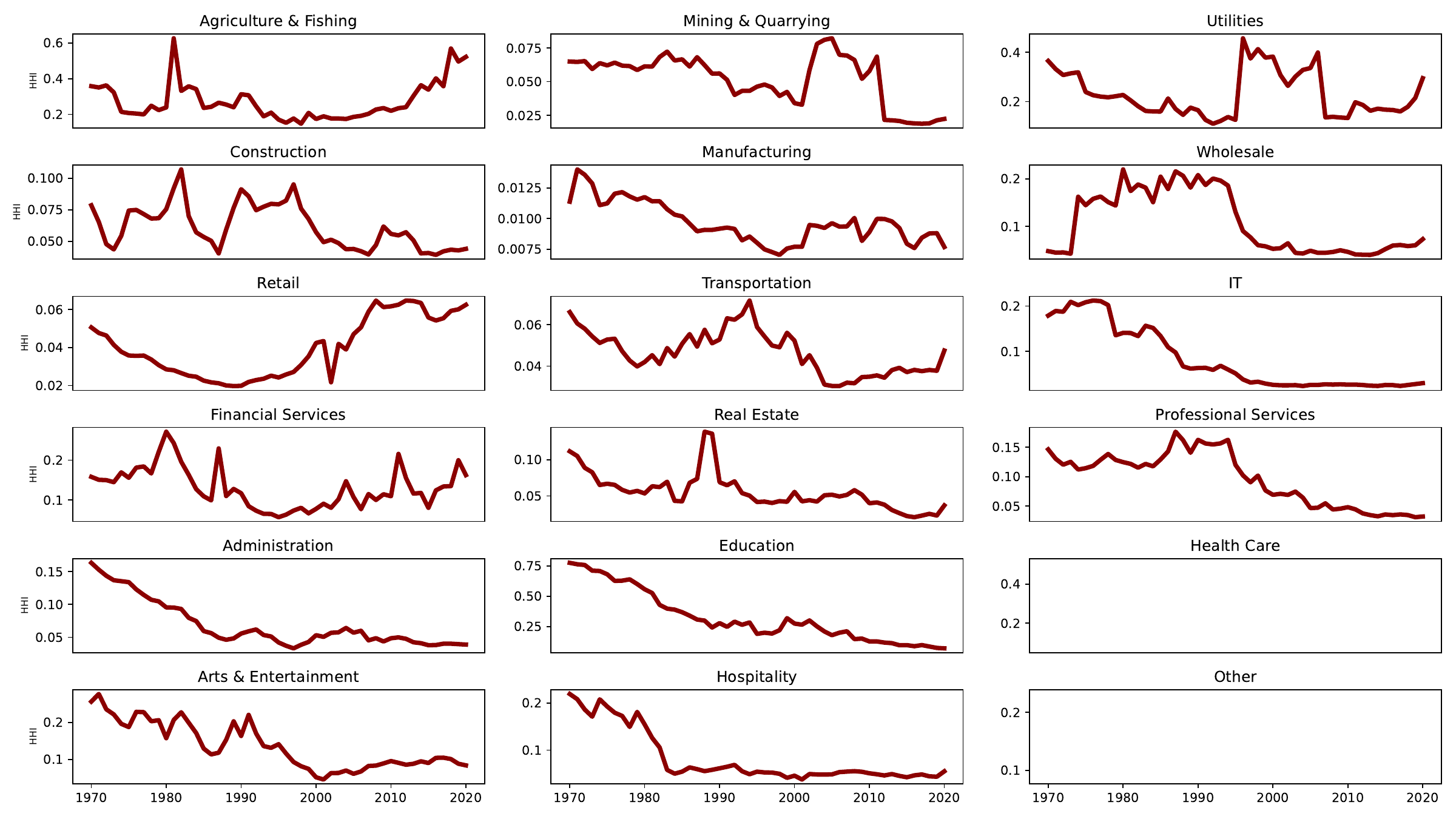}
    \label{fig:ind_HHI}
\end{figure}

\newpage 
\subsection{Income Shares}\label{app:income_shares}
In this section, we discuss the implications of our profit share estimates for the capital share. We compute the labor share from the National Income and Product Income Accounts following \cite{karabarbounis2023perspectives}. That is, we compute the labor share as compensation of employees divided by GDP minus taxes plus subsidies less proprietor's income.\footnote{Implicit in this calculation is the assumption that proprietor's labor share is the same as that of employees.} Under the representativeness assumption of Compustat firms, we can compute the capital share as one minus the sum of the labor share and the micro-aggregated profit share. The results are displayed in Figure \ref{fig:income_shares}. \par 
\begin{figure}[h!]
    \centering
    \caption{Income Shares in the United States, 1970--2020.}
    \vspace{-10pt}
    \includegraphics[scale=0.8]{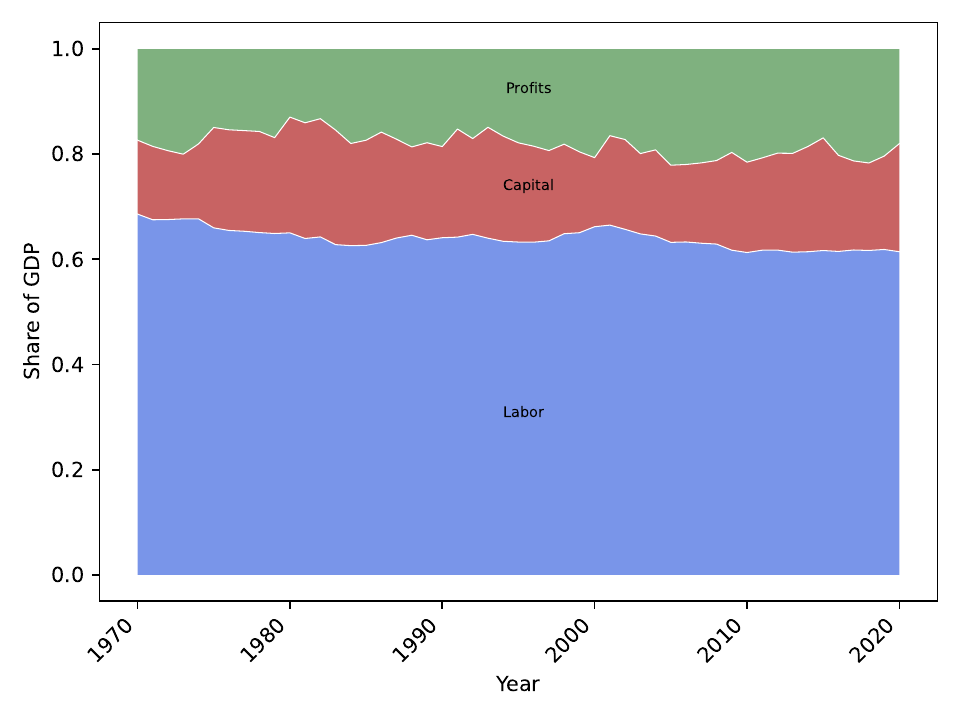}
    \label{fig:income_shares}
\end{figure}

The labor share has declined from approximately 69\% of GDP in 1970 to 62\% in 2020; that is, a decline of 7 percentage points, which mostly occurred from 2000 onwards. The 7 percentage-points decline in the labor share was absorbed by capital and profit shares, which have fluctuated substantially over the period. Because of the large volatility in capital- and profit shares, it is not clear how to allocate the decline in the labor share. Given the constancy of the profit share at around 18\% of GDP that we report for the entire sample period, we view the decline in the labor share as (mostly) driven by an increase in the capital share.\footnote{\cite{rognlie2016deciphering} argues that the capital share increased because of the housing sector.} Overall, despite the large volatility, income shares appear to be relatively stable over time, contrary to what others have suggested, and in line with the famous Kaldor facts. \par 

\newpage
\subsection{Robustness: No Intangible Capital}\label{app:no_itan}

\begin{figure}[h!]
    \centering
    \caption{Aggregate Markup, Aggregate Returns to Scale, and its Covariance}
    \includegraphics[scale=0.35]{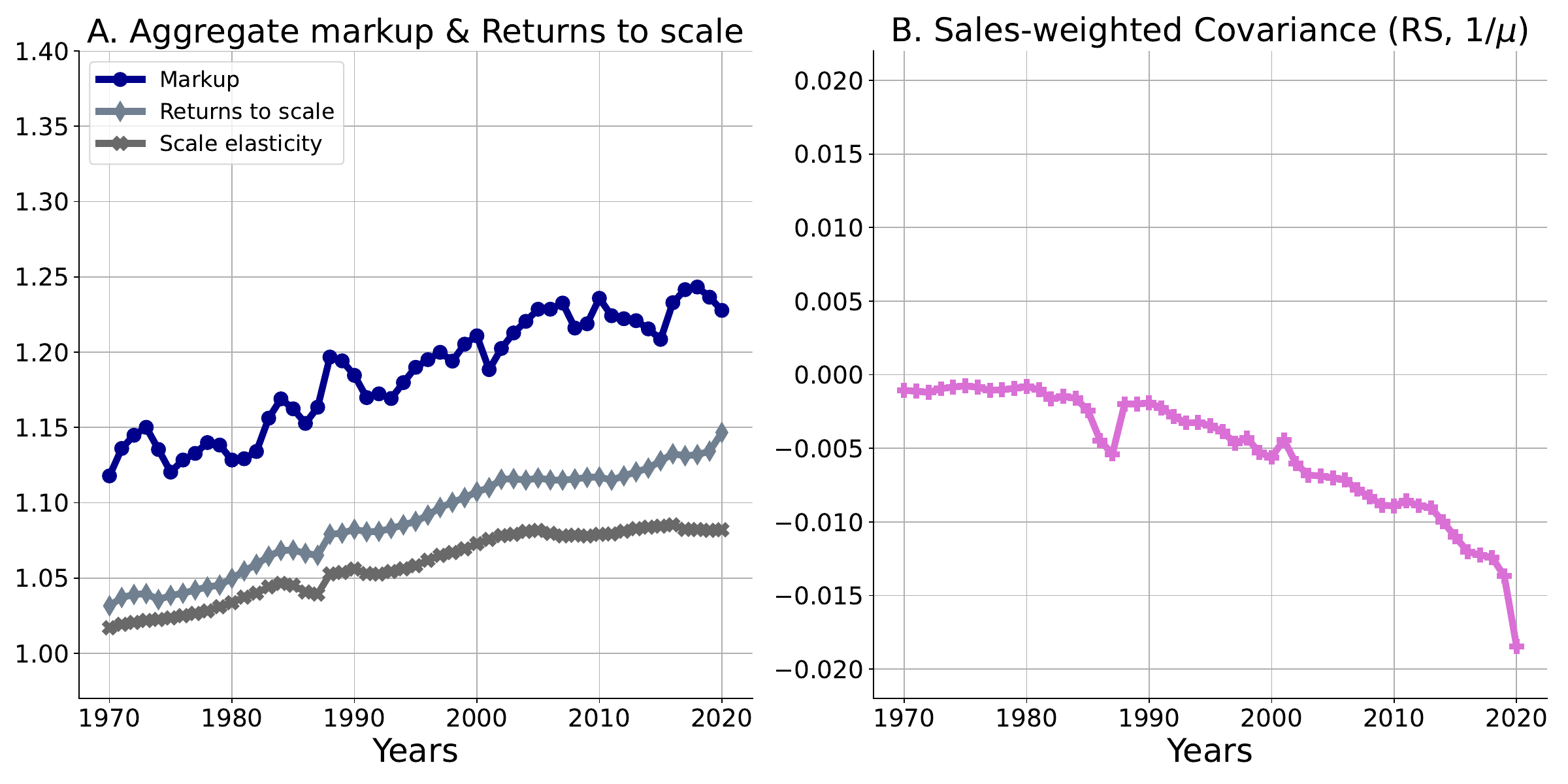}
    \label{fig:MP_nointan}
    \begin{minipage}{1\textwidth}
        \scriptsize{\textbf{Figure Notes}. Firm-level data comes from US Compustat, 1970--2020. In contrast to our benchmark results, where the capital stock is a composite of physical- and intangible capital, here we only consider the stock of physical capital. }  
    \end{minipage}
\end{figure}

\begin{figure}[h!]
    \centering
    \caption{The Micro--Aggregated Profit Share in the United States, 1970--2020.}
    \includegraphics[scale=0.35]{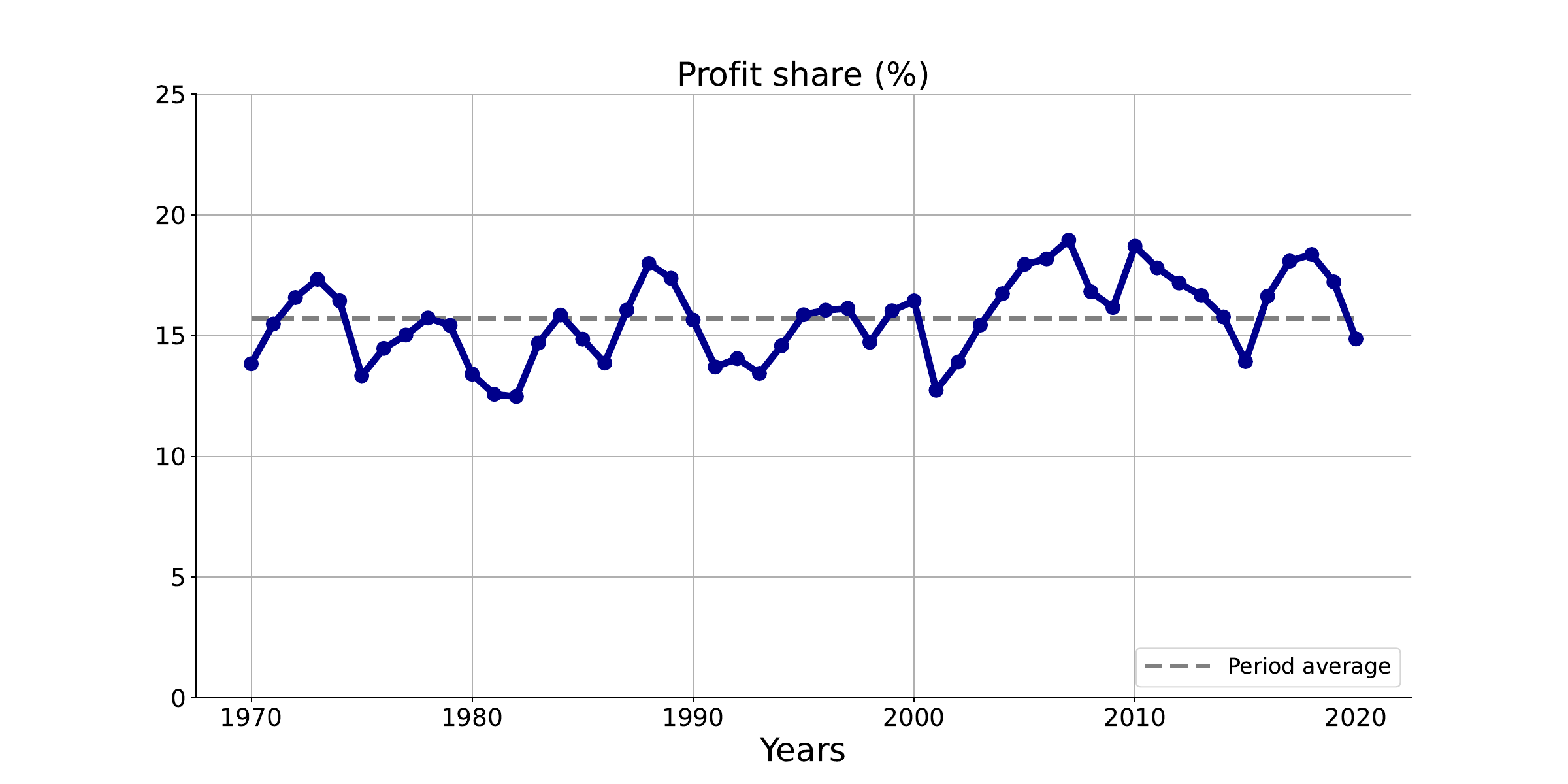}
    \label{fig:profshare_nointan}
    \begin{minipage}{1\textwidth}
        \scriptsize{\textbf{Figure Notes}. Firm-level data comes from US Compustat, 1970--2020. In contrast to our benchmark results, where the capital stock is a composite of physical- and intangible capital, here we only consider the stock of physical capital. }  
    \end{minipage}
\end{figure}

\end{document}